\documentclass[11pt,letterpaper]{article}
\pdfoutput=1
\usepackage[margin=1in]{geometry}
\usepackage[ascii]{inputenc}
\usepackage{mathtools}
\usepackage{mathpazo}

\usepackage[dvipsnames,svgnames]{xcolor}
\usepackage[bookmarksnumbered,linktocpage,hypertexnames=false,colorlinks=true,linkcolor=NavyBlue,urlcolor=NavyBlue,citecolor=ForestGreen,anchorcolor=green,breaklinks=true,pagebackref=true,pdfusetitle]{hyperref}
\usepackage[normalem]{ulem}
\usepackage{bbm,braket,mathrsfs,amsmath,amssymb,amsthm,amsfonts,latexsym,graphicx,multirow,booktabs,bm,xspace,float,mathdots,caption,subcaption,ellipsis,mleftright,setspace,dsfont,faktor,kantlipsum,algorithm,tocloft}
\usepackage[shortlabels]{enumitem}
\usepackage[T1]{fontenc}
\usepackage{textcomp}
\usepackage[english]{babel}
\usepackage[capitalize,nameinlink]{cleveref}
\newcommand{\crefpart}[2]{\hyperref[#2]{\namecref{#1}~\labelcref*{#1}~\ref*{#2}}}
\newtheorem{theorem}{Theorem}[section]
\newtheorem{corollary}[theorem]{Corollary}\AddToHook{env/corollary/begin}{\crefalias{theorem}{corollary}}
\newtheorem{definition}[theorem]{Definition}\AddToHook{env/definition/begin}{\crefalias{theorem}{definition}}
\newtheorem{lemma}[theorem]{Lemma}\AddToHook{env/lemma/begin}{\crefalias{theorem}{lemma}}
\newtheorem{proposition}[theorem]{Proposition}\AddToHook{env/proposition/begin}{\crefalias{theorem}{proposition}}\crefname{proposition}{Proposition}{Propositions}
\AddToHook{env/problem/begin}{\crefalias{theorem}{problem}}
\AddToHook{env/assumption/begin}{\crefalias{theorem}{assumption}}
\AddToHook{env/claim/begin}{\crefalias{theorem}{claim}}
\AddToHook{env/conjecture/begin}{\crefalias{theorem}{conjecture}}
\newtheorem{remark}[theorem]{Remark}\AddToHook{env/remark/begin}{\crefalias{theorem}{remark}}
\newtheorem{example}[theorem]{Example}\AddToHook{env/example/begin}{\crefalias{theorem}{example}}
\newtheorem{question}[theorem]{Question}\AddToHook{env/question/begin}{\crefalias{theorem}{question}}
\newtheorem*{theorem*}{Theorem}
\newtheorem*{corollary*}{Corollary}
\newtheorem*{proposition*}{Proposition}
\AddToHook{cmd/appendix/after}{\crefalias{section}{appendix}}
\numberwithin{equation}{section}
\allowdisplaybreaks[4]
\newcommand{\eps}{\varepsilon}
\newcommand{\ot}{\otimes}
\newcommand{\op}{\oplus}

\renewcommand{\AA}{\mathbb{A}}
\newcommand{\RR}{\mathbb{R}}
\newcommand{\CC}{\mathbb{C}}
\newcommand{\ZZ}{\mathbb{Z}}
\newcommand{\NN}{\mathbb{N}}

\newcommand{\bbH}{\mathbb{H}}
\newcommand{\bbP}{\mathbb{P}}

\newcommand{\diff}{\mathrm{d}}

\newcommand{\MM}{\mathcal{M}}
\newcommand{\VV}{\mathcal{V}}

\newcommand{\qfunc}{F}
\newcommand{\logqfunc}{E}
\DeclareMathOperator*{\argmin}{argmin}
\DeclareMathOperator{\LR}{LR}

\DeclareMathOperator{\poly}{poly}
\DeclareMathOperator{\intr}{int}
\DeclareMathOperator{\dom}{dom}
\DeclareMathOperator{\PD}{PD}
\DeclareMathOperator{\PSD}{PSD}
\DeclareMathOperator{\U}{U}

\DeclareMathOperator{\spec}{spec}
\DeclareMathOperator{\capacity}{cap}
\DeclareMathOperator{\Capacity}{Cap}

\DeclareMathOperator{\dist}{d}
\DeclareMathOperator{\rk}{rk}
\DeclareMathOperator{\GL}{GL}
\DeclareMathOperator{\SL}{SL}
\DeclareMathOperator{\Lin}{L}
\DeclareMathOperator{\Or}{O}
\DeclareMathOperator{\SOr}{SO}
\DeclareMathOperator{\Sym}{S}
\DeclareMathOperator{\diag}{diag}

\DeclareMathOperator{\tr}{tr}
\DeclareMathOperator{\ope}{op}
\DeclareMathOperator{\ncrk}{ncrk}
\DeclareMathOperator{\Lie}{Lie}

\DeclareMathOperator{\Herm}{Herm}

\DeclareMathOperator{\Ten}{Ten}
\DeclareMathOperator{\Shf}{shift}
\DeclareMathOperator{\Sum}{sum}
\DeclareMathOperator{\Hol}{Hol}
\DeclarePairedDelimiter\abs{\lvert}{\rvert}
\DeclarePairedDelimiter\norm{\lVert}{\rVert}

\DeclarePairedDelimiter\parens{\lparen}{\rparen}
\DeclarePairedDelimiter\braces{\lbrace}{\rbrace}

\newcommand{\minproblem}[2]{\text{minimize } #1 \quad \text{over } #2}
\title{Convex optimization on moment polytopes: Hadamard~mirror~descent and efficient algorithms for quantum functionals and other tensor parameters}
\author{M. Levent Do\u{g}an\thanks{LMU Munich \& MCQST, \texttt{\{mahmut.dogan,keiya.sakabe,michael.walter\}@lmu.de}} \and Keiya Sakabe\texorpdfstring{\footnotemark[1]}{} \and Michael Walter\texorpdfstring{\footnotemark[1]}{}}
\date{}
\begin{document}
\begin{titlepage}
\maketitle
\thispagestyle{empty}
\abstract{Convex optimization on polytopes arises in many areas of science. 
When the polytope is given implicitly or has exponentially many vertices and facets, standard methods may not apply or be ineffective.
This is the case for moment polytopes, such as the \emph{entanglement polytopes}, which play a foundational role in quantum information and algebraic complexity.
They give rise to important entanglement measures and tensor parameters such as the \emph{quantum functionals}, yet general effective methods for computing these quantities have been elusive.

In this paper we address this challenge.
We develop a first-order framework called \emph{Hadamard mirror descent} to optimize suitable convex functions over moment polytopes and, more generally, the gradient sets of geodesically convex functions.
It operates locally and does not rely on any explicit description of the polytope.
Our framework extends mirror descent, an effective and widely used framework for convex optimization, from the Euclidean setting to Hadamard manifolds, and is motivated by a recent work by Hirai, which we interpret as a Hadamard version of mirror flow.
Applying the framework to entanglement polytopes yields the first efficient first-order algorithms to compute the quantum functionals and other tensor parameters.
In more detail, our contributions are as follows:
\begin{itemize}
\item We introduce \emph{Hadamard mirror descent}, a framework to optimize general \emph{holonomy-invariant} convex objectives over the gradient sets of geodesically convex functions on Hadamard manifolds, and provide rigorous convergence guarantees under natural hypotheses.
\item We show that this framework applies to \emph{convex optimization on moment polytopes}, which arise from the gradients of geodesically convex Kempf--Ness functions.
\item We obtain the first efficient algorithm for computing quantum functionals: a simple \emph{entropic tensor scaling} method converges in a number of steps polynomial in the input length of the tensor and the inverse accuracy.
\item We obtain analogous algorithms for other important tensor parameters (symmetric quantum functional, $G$-stable rank, and non-commutative rank), as well as for norm minimization on moment polytopes (a key stability parameter in invariant theory).
Before this work, such algorithms were only known in special cases and for the non-commutative rank; for the latter, our algorithm appears conceptually simpler than prior work.
\end{itemize}}
\end{titlepage}
\setcounter{tocdepth}{3}
\setlength{\cftbeforesecskip}{0.75em}
\tableofcontents
\clearpage

\section{Introduction}\label{sec:intro}
Convex optimization on polytopes is a fundamental problem that arises broadly in computer science and mathematics, with applications ranging from graph algorithms to machine learning:
\begin{equation}\label{eq:intro-convex-opt}
    \minproblem{\phi(p)}{p \in D} \quad (D \text{ is a polytope}, \, \phi \text{ is convex}).
\end{equation}
In the classical setting of linear programming, one optimizes a linear function subject to linear constraints that define a polytope.
For example, the famous maximum-flow problem is a linear program on the flow polytope, and bipartite matching reduces to optimization over the Birkhoff polytope (the convex hull of the permutation matrices).
Linear relaxations of \textsc{NP}-hard combinatorial optimization problems, such as arise from the cut-polytope formulation of \textsc{MAX-CUT} or the integer-programming formulation of vertex cover, are ubiquitous and can yield approximation algorithms with provable guarantees.
In computer science, information theory, and physics, one frequently encounters the problem of maximizing the \emph{Shannon entropy} $H(p) \coloneqq -\sum_{i=1}^n p_i\log p_i$, which measures the uncertainty in a distribution.
Entropy maximization with mean constraints is a fundamental tool in approximation algorithms, including algorithms for the traveling salesman problem~\cite{AGMGS-10,vishnoi-singh,vishnoi-straszak}, and for volume approximation of polytopes and spectrahedra~\cite{Barvinok-09,Barvinok-Hartigan-10,Barvinok-12,Barvinok-Rudelson-21,Dogan-Leake-Ravichandran-22}.
In the context of probability theory and statistical physics, the same problem arises through the \emph{maximum-entropy principle} of Jaynes \cite{jaynes1957informationi,jaynes1957informationii}, which asserts that the best guess for an unknown probability distribution is the one that maximizes the entropy subject to the known constraints.

When the polytope is given explicitly, e.g., as a list of vertices or inequalities, standard techniques such as ellipsoid, cutting-plane, or interior-point methods often yield efficient algorithms for convex optimization~\cite{kelley1960cutting,Khachiyan79,Khachiyan80,Renegar-book}.
However, for many interesting problems, the polytope is only given more \emph{implicitly}, such as the polytope of all spanning trees associated to a graph, i.e., the matroid polytope of spanning trees.

In this setting, more advanced techniques are required that do not rely on efficient enumeration of vertices or facets.
One particularly successful strategy is to represent a polytope in terms of the \emph{gradients} of a convex function.
Suppose that $D\subset\RR^n$ is a polytope with vertices $\omega_1,\omega_2,\dots,\omega_m$.
An easy calculation shows that $D$ coincides with the topological closure of the gradient set of a \emph{log-sum-exp} function:
\begin{equation}\label{eq:gradient-trick}
    D = \overline{\{\nabla f(x) \mid x\in\RR^n\}}, \qquad \text{where } f(x) \coloneqq \log\mleft( \sum_{i=1}^m e^{\langle \omega_i,x\rangle} \mright).
\end{equation}
Functions of log-sum-exp form have found numerous applications, such as in machine learning \cite{ml-lse-3,ml-lse-1,ml-lse-2}.
Their optimization falls into the framework of geometric programming \cite{geomet-tut}.
Using \cref{eq:gradient-trick}, we may recast problem~\eqref{eq:intro-convex-opt} as follows:
\begin{equation}\label{eq:gradient-opt}
    \minproblem{\phi\bigl(\nabla f(x) \bigr)}{x \in \RR^n}.
\end{equation}
This reformulation offers two significant advantages.
First, viewing a polytope as (the closure of) the set of gradients of a convex function can provide a \emph{succinct} ``encoding'' of complicated polytopes.
For example, even for a polytope with an exponential number of vertices (such as the convex hull of all spanning trees of a graph), the convex function~$f$ still has only $n$ variables (the number of edges in the aforementioned example).
Moreover, even if $f$ has exponentially many terms when written naively in the log-sum-exp form, it is sometimes possible to represent it compactly and evaluate it efficiently;
for spanning-tree polytopes this is provided by the weighted version of Kirchhoff's theorem \cite{Bollobas-02}; see, e.g., \cite{AGMGS-10}.
More generally, recent works on~entropy maximization consider evaluation access to~$f$ as a generalized counting oracle~\cite{vishnoi-singh,vishnoi-straszak}.

Second, the reformulation in \eqref{eq:gradient-opt} transforms the original constrained optimization problem into an \emph{unconstrained} one, which enables simpler algorithms.
To see this, assume that both~$f$ and the objective~$\phi$ are sufficiently differentiable.
Then the composite function~$\phi(\nabla f(x))$ is differentiable and its derivative can be computed via the chain rule:
\begin{equation}\label{eq:intro-euclidean-descent}
    \left.\frac{\diff}{\diff t}\right|_{t = 0}\phi(\nabla f(x + tY)) = \nabla^2f(x)\left[ \nabla \phi\left(\nabla f(x)\right), Y\right ].
\end{equation}
Because $f$ is convex, its Hessian $\nabla^2 f(x)$ is positive semidefinite, and hence we see from \cref{eq:intro-euclidean-descent} that moving infinitesimally in the direction of $Y \coloneqq -\nabla \phi(\nabla f(x))$ will never increase $\phi(\nabla f(x))$.
This motivates the following first-order descent algorithm:
\begin{equation}\label{eq:intro-mirror-descent}
    x_{k+1} \coloneqq x_k - \eta\nabla \phi(\nabla f(x_k)).
\end{equation}
In fact, under standard hypotheses, this coincides precisely with the celebrated \emph{mirror descent} algorithm~\cite{NY1983} when expressed in dual (i.e., gradient) variables, for the distance-generating function $h = f^\ast$ and the objective function $\phi$.
Foundational results in the modern theory of mirror descent~\cite{BBT2017,LFN2018} state that the step size can be taken as the inverse of the \emph{relative smoothness parameter}~$L$ of $\phi$ with respect to~$h$.
In dual variables, this parameter is defined in terms of a ratio of Bregman divergences as follows:
\begin{equation}\label{eq:intro-relative-smoothness}
    D_\phi(\nabla f(y)\Vert\nabla f(x)) \le L\cdot D_f(x\Vert y) \qquad (x, y \in \RR^n),
\end{equation}
where $D_F(w\Vert z)\coloneqq F(w) - F(z) - \langle\nabla F(z), w - z\rangle$ is the \emph{Bregman divergence} of a convex function~$F$.
For the step size $\eta = 1/L$, mirror descent minimizes~$\phi$ at a rate of $\mathcal{O}(k^{-1})$~\cite{BBT2017,LFN2018}.

\medskip

Remarkably, convex polytopes arise from gradients of convex functions in situations far more general than the Euclidean setting discussed above.
In pure mathematics, this goes back to celebrated convexity theorems in symplectic geometry by Kirwan and Ness--Mumford~\cite{Kirwan1984,NessMumford1984}, generalizing earlier works on special cases by Kostant, Atiyah, Guillemin--Sternberg, and many others~\cite{kostant1973convexity,atiyah1982convexity,guillemin1982convexity}.%
\footnote{The well-known Schur--Horn theorem on the relation between the diagonal entries and the eigenvalues of Hermitian matrices is a special case.}
Much more recently, a line of work in computer science has over the past decade begun to develop a \emph{quantitative} and \emph{algorithmic} framework that unifies and generalizes paradigmatic methods, such as \emph{matrix} and \emph{operator scaling}~\cite{linial1998deterministic,GGOW-16,GGOW-18,AGLOW-18,BGOWW-alternating,BFGOWW2018,Franks2018,BFGOWW2019,FW2020,HNW2023,HiraiSakabe2024,Hirai2025,Hirai-ncrank}.
We describe the setting and the emergence of polytopes more concretely in \cref{sub:intro primer} below, but for now it suffices to say that the corresponding polytopes are known as \emph{moment polytopes} and they arise from convex functions on curved spaces called \emph{Hadamard manifolds} associated with \emph{non-commutative} group actions.
Importantly, they capture foundational problems across mathematics and computer science.
We only give a few examples and refer to \cite{BFGOWW2019} for many more applications:
\begin{itemize}
    \item In algebraic complexity, important tensor parameters are defined by convex optimization over moment polytopes, which in this case are called \emph{entanglement polytopes}.
    An important example is the family of \emph{quantum functionals}~\cite{CVZ2018}; they are universal spectral points in the sense of Strassen, constrain the asymptotic restriction of tensors, and bound complexity-theoretic measures such as the \emph{asymptotic~subrank}~\cite{Burgisser-book,zuiddam2018algebraic}.
    \item In analysis, they arise as \emph{Brascamp--Lieb polytopes}, which characterize the validity of the so-called Brascamp--Lieb inequalities.
    These are integral inequalities that generalize the H\"older and Loomis--Whitney inequalities~\cite{bennett2008brascamp,GGOW-18}.
	\item In quantum information, entanglement polytopes characterize which local quantum states can be obtained from multipartite quantum states in a given entanglement class~\cite{WDGC2013_entanglement}; as a special case, this captures the one-body \emph{quantum marginal} and \emph{$N$-representability problems}~\cite{Klyachko02,Klyachko04,Klyachko06}.
    In particular, natural problems such as marginal entropy maximization are captured by convex optimization on entanglement polytopes.
	\item In geometric invariant theory, a fundamental problem is to determine whether a vector is in the so-called \emph{null cone} of a group action and, if so, to determine its \emph{instability}.
    This is captured by minimizing the Euclidean norm over its moment polytope, as follows from celebrated results by Kempf--Ness and others~\cite{Kempf-78,Kempf-Ness-79,NessMumford1984,Kirwan1984}.
    \item In matrix analysis, the \emph{Horn polytopes} capture the relation between the eigenvalues of Hermitian matrices~$A,B,C$ such that $A+B=C$ \cite{Horn-62,Klyachko-98,Knutson-Tao}; see also \cite{Fulton-00}.
\end{itemize}
Despite its significance, efficient optimization over moment polytopes has proved to be a major challenge.
Moment polytopes often have exponentially many vertices and facets; moreover, their explicit description is rarely known; see, e.g., \cite{vandenberg2025computing} and references therein.
Thus, standard convex optimization techniques often cannot be applied or are inefficient.
To date, there has been major algorithmic progress in important special cases---most notably through operator scaling, tensor scaling, and the broader non-commutative optimization framework \cite{GGOW-16,GGOW-18,Franks2018,AGLOW-18,BGOWW-alternating,BFGOWW2018,BFGOWW2019,HNW2023,Hirai2025}.
While general algorithms are known that address the membership problem in moment polytopes, they do so inefficiently (in particular, this is the case for all examples given above).
Hence these prior works do not imply an efficient algorithm for convex optimization on these polytopes, a major open problem.

\medskip

In this work, we address this challenge: we develop a first-order framework called Hadamard mirror descent to optimize suitable convex functions over moment polytopes and, more generally, the gradients of geodesically convex functions.
Importantly, it operates locally and does not rely on any explicit description of the polytope.
We achieve this by extending mirror descent from the Euclidean setting to Hadamard manifolds.
It is motivated by recent work by Hirai, which we interpret as a Hadamard version of the continuous-time dynamics corresponding to mirror descent.
Applied to moment polytopes, this framework yields efficient first-order algorithms to compute quantum functionals and other tensor parameters.
In more detail, our main contributions~are:
\begin{itemize}
\item We introduce \emph{Hadamard mirror descent}, a framework to optimize a broad class of convex objectives over the gradient set of geodesically convex functions on Hadamard manifolds, and provide rigorous~$\mathcal O(k^{-1})$ convergence guarantees under natural hypotheses.
\item We show that this framework applies to \emph{convex optimization on moment polytopes}, which arise from the gradients of geodesically convex Kempf--Ness functions.
This application requires us to quantitatively bound the relevant complexity parameters of the general framework.
\item We obtain the first efficient algorithm for computing quantum functionals: a simple \emph{entropic tensor scaling} method converges in a number of steps polynomial in the input length of the tensor and the inverse accuracy.
\item We obtain analogous algorithms for other important tensor parameters (symmetric quantum functional, $G$-stable rank, and non-commutative rank), as well as for norm minimization on moment polytopes (a key stability parameter in invariant theory).
Before our work, such algorithms were only known in special cases, and for the non-commutative rank, for which we obtain an algorithm that appears conceptually simpler than prior work.
\end{itemize}
The rest of the introduction is organized as follows:
\Cref{sub:intro primer} explains concretely, using the example of entanglement polytopes, how convex polytopes can emerge from the gradients of convex functions on curved spaces, and sketches how one naturally obtains a descent algorithm generalizing \cref{eq:intro-mirror-descent}.
\Cref{sub:intro framework} describes Hadamard mirror descent in its appropriate generality and states our convergence theorem; we also discuss the associated complexity parameters.
\Cref{sub:convex-moment} applies the preceding to moment polytopes and \cref{sub:convex-entanglement} specializes further to the entanglement polytopes of tensors to obtain efficient algorithms for quantum functionals, the symmetric quantum functional, the $G$-stable rank, and the non-commutative rank.

\subsection{Motivation: Moment polytopes, non-commutative gradients, and mirror descent}\label{sub:intro primer}
Moment polytopes are implicitly defined polytopes that arise from ``nice'' group actions.
To make this concrete and motivate our general approach, we use the \emph{tensor action} as our running example:
the action of the group $\GL\coloneqq  \GL(n_1)\times\cdots\times \GL(n_d)$ on the space $\bigotimes_{\ell=1}^d \CC^{n_\ell} = \CC^{n_1} \ot \cdots \ot \CC^{n_d}$ of $d$-tensors via $(g_1,\dots,g_d)\cdot T\coloneqq (g_1\otimes\dots\otimes g_d)T$.
We denote the positive-definite group elements by~$\PD \coloneqq \PD(n_1)\times\dots\times\PD(n_d)$ and denote by $\langle\cdot,\cdot\rangle$ the inner product on $\bigotimes_{\ell=1}^d \CC^{n_\ell}$.
Our goal is to associate to any tensor~$T\neq0$ a polytope.
To this end, we consider the \emph{Kempf--Ness~function}%
\footnote{Note that this function simply measures the norm of the tensor as we vary over the \emph{orbit} $\GL \cdot T \coloneqq \{ g \cdot T \mid g \in \GL \}$, since $f_T(g^\dagger g) = \log \, \norm{g \cdot T}^2$; this was the original motivation for the definition~\cite{Kempf-Ness-79}.}
\begin{equation}\label{eq:kn T}
    f_T \colon \PD\rightarrow\RR, \quad f_T(x_1,\dots,x_d) = \log\,\langle T, (x_1\otimes\dots\otimes x_d)T\rangle.
\end{equation}
We can think of this as a non-commutative generalization of the log-sum-exp function in \cref{eq:gradient-trick}.
To see the connection, write $x_i = g_i^\dagger e^{\diag(z_i)} g_i$ for $g_i \in \GL(n_i)$ and $z_i \in \RR^{n_i}$.
Then:
\begin{align*}
    f_T(x_1,\dots,x_d) = \log \smash{\sum_{i_1,\dots,i_d}} \abs{S_{i_1,\dots,i_d}}^2 e^{z_{1,i_1} + \cdots + z_{d,i_d}},
    \quad
    \text{where } S \coloneqq (g_1,\ldots,g_d) \cdot T,
\end{align*}
which has the form of \cref{eq:gradient-trick}.
As a consequence, we see that $f_T$ is convex along arbitrary curves of the form~$x_i(t) = g_i^\dagger e^{\diag(z_i + t w_i)} g_i$ ($t \in \RR$).%
\footnote{Note that $\diag(z_i)$ and $\diag(w_i)$ are diagonal matrices and hence they commute.
If we replace them by non-commuting Hermitian matrices, we no longer obtain a log-sum-exp function and indeed convexity will in general fail.
In other words, we \emph{cannot} simply turn $f_T$ into an ordinary Euclidean convex function by the change of variables~$x_i = e^{X_i}$.}
What are these curves? Clearly, they are not straight lines in the naive Euclidean sense of $n_i\times n_i$ matrices.
Instead, they are the \emph{geodesics} if $\PD(n_i)$ is equipped with the \emph{affine-invariant} Riemannian metric.%
\footnote{The affine-invariant Riemannian metric on $\PD(n)$ is given by the formula $g(X,Y)_x \coloneqq \tr[x^{-1} X x^{-1} Y]$ at $x \in \PD(n)$; we regard the tangent vectors~$X,Y$ as Hermitian matrices in~$\Herm(n)$, using that $\PD(n)$ is an open subset of~$\Herm(n)$.}
This metric turns~$\PD(n_i)$ and hence~$\PD$ into a so-called \emph{Hadamard manifold}---a particularly nice kind of Riemannian manifold with nonpositive curvature.
Thus, the Kempf--Ness function~$f_T$ is a \emph{geodesically convex} function on~$\PD$.

Now, how does a polytope emerge from the gradient set of this function?
One immediate challenge we face is that the gradients $\nabla f_T(x)$ for $x \in \PD$ all live in different tangent spaces.
We therefore use the manifold's parallel transport to move them all to the same tangent space, say, the tangent space at the identity~$I \in \PD$, which is naturally identified with the $d$-tuples of Hermitian matrices $\bbH \coloneqq \Herm(n_1)\times\cdots\times\Herm(n_d)$ with the Hilbert--Schmidt inner product.
But now we face the real challenge: unlike in the Euclidean setting, the parallel transport will depend on the choice of curve along which we transport%
\footnote{Arguably there is a canonical choice for this curve in our setting---the geodesic---but this does not yield a convex polytope in~$\bbH$ (in general, no fixed choice of curves works for all tensors).}%
---different curves will give results that differ by the manifold's holonomy, which in the present case amounts to conjugating a matrix tuple in~$\bbH$ by unitary matrices.
To overcome this ambiguity, we can look at unitarily invariant data, namely the eigenvalue spectra.
The result of parallel transporting the gradients to the identity, computing the eigenvalues, and taking the closure has a simple concrete description:%
\footnote{This follows by a short calculation: $\nabla f_T(x) = (x_1^{1/2} \rho_1 x_1^{1/2}, \dots, x_d^{1/2} \rho_d x_d^{1/2})$, where $\rho_\ell = S^{(\ell)} S^{(\ell)\dagger} / \norm S^2$ and $S \coloneqq (x_1^{1/2}, \dots, x_d^{1/2}) \cdot T$; we identify $T_x \PD \cong \bbH$ in the canonical way by using the fact that $\PD \subseteq \bbH$ is an open subset.
In the same coordinates, parallel transport along the geodesic to the identity is given by $x^{-1/2} \nabla f_T(x) x^{-1/2} = (\rho_1,\dots,\rho_d)$.}
\begin{equation}\label{eq:intro-entanglement-polytope}
\begin{aligned}
    \Delta(T)
&\coloneqq \overline{ \{ \spec \tau_{x\to I} \nabla f_T(x) : x \in \PD \} } \\
&= \left\{\left(\spec{\frac{S^{(1)}S^{(1)\dag}}{\|S\|^2}},\ldots, \spec{\frac{S^{(d)}S^{(d)\dag}}{\|S\|^2}} \right) \ \middle|\ 0\neq S \in \overline{\GL\cdot T} \right\} \subseteq \RR_{\downarrow}^{n_1} \!\times\! \cdots \!\times\! \RR_{\downarrow}^{n_d},
\end{aligned}
\end{equation}
where $\tau_{x\to I}$ denotes parallel transport from~$x$ to~$I$;
$\spec(\cdot)$ denotes the sorted eigenvalues of a Hermitian matrix, and the componentwise extension to matrix tuples;
and $S^{(\ell)}$ denotes the $\ell$-th \emph{flattening} obtained by regarding~$S$ as a linear map $S^{(\ell)} \colon \bigotimes_{j\neq \ell} \CC^{n_j}\rightarrow \CC^{n_\ell}$.
Note the similarity to \cref{eq:gradient-trick}.
We also observe that $\Delta(T)$ is a subset of a Cartesian product of probability simplices.
Remarkably, $\Delta(T)$ is a convex polytope---known as the tensor's \emph{entanglement polytope}~\cite{WDGC2013_entanglement}.

General moment polytopes arise in just the same way from any vector in a ``nice'' (rational) representation of a ``nice'' (self-adjoint) matrix group $G \subseteq \GL(n)$.
One considers the \emph{Kempf--Ness function} associated with the vector, which is a \emph{geodesically convex} function on a nonpositively curved space $G \cap \PD(n)$.
Suitably interpreted, the closure of the gradient set of the Kempf--Ness function becomes a convex polytope, known as the vector's \emph{moment polytope}~\cite{Kirwan1984,NessMumford1984,BFGOWW2019}.

\medskip

We now consider optimization over the entanglement polytope.
Let $\phi \colon \RR^{n_1}\times\cdots\times\RR^{n_d}\to\overline{\RR}$ be a convex objective function that we wish to minimize.
Since the entanglement polytope is described in terms of eigenvalue spectra, which carry no intrinsic order, it is natural to assume that $\phi$ is permutation-symmetric in each argument.
Such a function extends to a unitarily invariant convex function $\Phi \colon \bbH\to \overline{\RR}$ by $\Phi(X_1,\dots,X_d)\coloneqq \phi(\spec X_1,\dots,\spec X_d)$.
Then, assuming mild regularity, convex optimization over the entanglement polytope can be written as
\begin{equation}\label{eq:intro phi to Phi}
    \min_{p\in\Delta(T)} \phi(p) = \inf_{x\in\PD}\Phi(\tau_{x\to I} \nabla f_T(x)).
\end{equation}
This reformulates the problem of convex optimization over entanglement polytopes, which have a rather intricate and implicit definition, as an unconstrained optimization over the gradients of the simple Kempf--Ness function---in precise analogy to the transition from \cref{eq:intro-convex-opt} to \cref{eq:gradient-opt}.

Following Hirai~\cite{Hirai2025}, this can also be expressed more abstractly, yet more intrinsically, by extending the objective to the \emph{cotangent bundle}~$T^*\PD$.
To this end, we regard~$\Phi$ as a function~$Q_I$ on the cotangent space $T^*_I\PD$ using the Riemannian metric, and extend it to a function~$Q$ on the entire cotangent bundle via parallel transport ($Q_x \coloneqq Q_I \circ \tau_{I \to x}^*$).
Then:
\begin{equation}\label{eq:intro Phi to Q}
    \inf_{x \in \PD}\Phi(\tau_{x \to I}\nabla f_T(x))
= \inf_{x \in \PD}Q(\diff f_T(x)).
\end{equation}
The advantage of this formulation is that, unlike the gradient, the differential $\diff f_T$ is defined independently of the metric.%
\footnote{The same is true in the Euclidean setting: Problem~\eqref{eq:gradient-opt} is more intrinsically cast as a minimization of $Q\bigl(\diff f(x) \bigr)$ for a convex objective~$Q \colon (\RR^n)^* \to \overline\RR$.}
Motivated by this, Hirai proposed the general minimization problem
\begin{equation}\label{eq:intro-general-minimization}
    \minproblem{Q(\diff f(x))}{x \in \MM},
\end{equation}
where $f$ is a geodesically convex function on an arbitrary Hadamard manifold $\MM$, and $Q$ is a fiberwise convex function on the cotangent bundle that is invariant under parallel transport; we call such a function~$Q$ a \emph{holonomy-invariant convex objective}.
Using this invariance and the chain rule, Hirai showed that (under suitable regularity and domain assumptions) the derivative of~$Q(\diff f(x_t))$ along an arbitrary curve~$x_t$ is given by
\begin{equation}\label{eq:intro q grad vs hessian}
    \left.\frac{\diff}{\diff s}\right|_{s=t} Q(\diff f(x_s)) = \nabla^2 f(x_t)[\nabla^Q f(x_t), \dot x_t],
\end{equation}
where $\nabla^Q f(x)\coloneqq \diff Q_x (\diff f(x)) \in T_x\MM$ is called the \emph{$Q$-gradient} of~$f$.%
\footnote{In other words, the $Q$-gradient of~$f$ is the Riemannian gradient of~$Q(\diff f(x))$ if the manifold~$\MM$ is equipped with the Riemannian metric induced by the Hessian~$\nabla^2 f$ (assuming the latter is everywhere positive definite).}
The geodesic convexity of~$f$ implies that its Riemannian Hessian~$\nabla^2 f$ is positive semidefinite, and hence~$Q(\diff f(x_t))$ is nonincreasing along the following descent flow:
\begin{equation}\label{eq:intro-flow}
    \dot{x}_t = -\nabla^Q f(x_t).
\end{equation}
Hirai proposed this $Q$-gradient flow and established qualitative convergence to the infimum of problem~\eqref{eq:intro-general-minimization} under suitable hypotheses, for a broad class of Hadamard manifolds such as~$\PD(n)$.

The starting point for our work is a new interpretation of Hirai's flow:
we observe that \cref{eq:intro q grad vs hessian} is precisely the manifold analogue of the Euclidean identity in \cref{eq:intro-euclidean-descent}.
Thus, Hirai's flow can be interpreted as a \emph{mirror flow} on a Hadamard manifold.
This connection motivates us to extend the powerful theory of mirror descent to the Hadamard setting.
This allows us to obtain \emph{quantitative} convergence guarantees for both the flow and its discrete-time counterpart, and then apply the resulting framework to moment polytope optimization to derive our concrete algorithmic results.

\subsection{Framework and general results: Hadamard mirror descent}\label{sub:intro framework}
We now present our framework of Hadamard mirror descent and our general convergence results.
Throughout, we consider a Hadamard manifold~$\MM$, that is, a complete, simply connected Riemannian manifold with nonpositive sectional curvature.
Our goal is to address the general minimization problem~\eqref{eq:intro-general-minimization}, which we restate for convenience:
\begin{equation}\label{eq:intro-general-minimization restated}
    \minproblem{Q(\diff f(x))}{x \in \MM}.
\end{equation}
Here, $f$ is a \emph{geodesically convex} function on~$\MM$, meaning it is convex when restricted to any geodesic, and $Q$ is a \emph{holonomy-invariant convex objective} on the cotangent bundle~$T^*\MM$.
As explained above, the latter means that its restrictions $Q_x$ to any cotangent space are lower semicontinuous (l.s.c.) and Euclidean convex, and related to each other by parallel transport (independently of the choice of curves).
We also assume sufficient regularity so that~\eqref{eq:intro-general-minimization restated} and the $Q$-gradient are well-defined; the precise details can be found in \cref{def:holoinvariant,def:grad}.


In \cref{sec:continuous-time}, we formalize our observation that Hirai's $Q$-gradient flow~\eqref{eq:intro-flow} is a Hadamard-manifold generalization of \emph{mirror flow}, the continuous-time version of mirror descent~\cite{NY1983}, and show a general $\mathcal{O}(t^{-1})$ convergence result for \emph{arbitrary} Hadamard manifolds that applies under natural hypotheses; this significantly generalizes the qualitative convergence result of~\cite{Hirai2025} and makes it quantitative.
In \cref{sec:discrete-time}, we consider the corresponding discrete-time method, which we call \emph{Hadamard mirror descent}, as it precisely generalizes the mirror descent update~\eqref{eq:intro-mirror-descent}:
\begin{equation}\label{eq:intro-Hadamard-mirror-descent}
    x_{k+1} \coloneqq \exp_{x_k}( -\eta \nabla^Q f(x_k) ),
\end{equation}
where $\nabla^Q f(x)\coloneqq \diff Q_x (\diff f(x)) \in T_x\MM$ is Hirai's $Q$-gradient and $\exp_x \colon T_x\MM\rightarrow \MM$ the exponential map.

We now explain the quantities that control the convergence of Hadamard mirror descent.
On the one hand, the convergence depends on the choice of step size~$\eta$; we comment more on this below.
On the other hand, it depends on a certain natural function gap.
Let $b^\xi \colon \MM \to \RR$ be a Busemann function: an analog of a linear function on Hadamard manifolds.
The invariance under parallel transport implies that $Q(\xi) \coloneqq Q(-\diff b^\xi(x))$ is well-defined, independently of~$x$.
Let $f^\xi \coloneqq f + b^\xi$ be the shifted geodesically convex function.
Then we have the following result:

\begin{theorem}[Convergence of Hadamard mirror descent, see \cref{thm:main-discrete,lem:relsmooth}]\label{thm:intro-main}
    Under a natural step-size condition (satisfied for $\eta=1/L$ with $L$ the relative smoothness parameter discussed below), we have, for any Busemann function~$b^\xi$ and any~$k\geq1$,
    \[
        Q(\diff f(x_k)) - Q(\xi) \leq \frac{f^\xi(x_0) - \inf f^\xi}{\eta k}.
    \]
    In particular:
    \[
        \lim_{k \to \infty}Q(\diff f(x_k)) = \inf_{x \in \MM}Q(\diff f(x)).
    \]
\end{theorem}
\noindent
In the Euclidean case $\MM = \RR^n$, the quantity $f^\xi(x) - \inf f^\xi$ for the linear function $b^\xi = -\langle \xi, \cdot\rangle$ equals the Bregman divergence $D_h(\xi\Vert\nabla f(x))$ of the distance-generating function $h = f^\ast$ (under natural hypotheses).
Hence, \cref{thm:intro-main} is a natural generalization of the foundational convergence results in Euclidean mirror descent~\cite{BBT2017,LFN2018}.

\Cref{thm:intro-main} serves as the general convergence theorem from which the algorithmic results of the paper are derived.
The remainder of the introduction successively specializes this result, first to convex optimization on moment polytopes, and then further to concrete tensor parameters.
As for the choice of step size~$\eta$, we show that the notion of relative smoothness readily generalizes to this setting when phrased in the dual form~\eqref{eq:intro-relative-smoothness}:
we say that the pair~$(f, Q)$ is \emph{$L$-relatively smooth} if
\begin{equation}
\label{eq:rel-smoothness-ineq}
    D_{Q_x}(\tau_{x\to y}^*\diff f(y)\Vert \diff f(x)) \leq L \cdot D_f^{\MM}(x\Vert y) \quad (\forall x, y \in \MM),
\end{equation}
where $D_f^{\MM}(x\Vert y)\coloneqq f(x)-f(y)-\diff f(y) [ \exp^{-1}_y(x) ]$ is the Riemannian Bregman divergence.
Then, \cref{thm:intro-main} applies with step size~$\eta = 1/L$.\footnote{In fact, we observe that the relative smoothness only needs to hold along the steps actually taken by the algorithm; we formalize this as a natural step-size condition in the body of the paper (see \cref{def:step-size-condition} and \cref{lem:relsmooth}~(a) for its connection to relative smoothness).
This additional flexibility is essential for our application to quantum functionals: the corresponding pair $(f,Q)$ is not $1$-relatively smooth, yet we prove that the step-size condition holds for all~$\eta\in(0,1]$. }
In the Euclidean setting, if $f$ and $Q$ are individually smooth%
\footnote{Throughout this paper, ``smoothness'' refers to the standard notion in convex analysis of Lipschitz-continuity of the gradient (not to differentiability).}
with parameters $L_f$ and $L_Q$, respectively, then relative smoothness holds with~$L=L_f L_Q$.
We show that this criterion generalizes to the manifold setting but with an important twist: the smoothness of~$f$ must be replaced by a new notion we call \emph{dual-smoothness}.%
\footnote{We say that a geodesically convex function~$f$ is $L_f$-dual-smooth if $\frac12 \norm{ \tau^{\ast}_{x \to y}\diff f(y) - \diff f(x) }^2 \leq L_f \cdot D_f^\MM(x\Vert y)$ for all~$x,y\in\MM$.
In the Euclidean case, this is equivalent to $L_f$-smoothness; in general it is a more stringent condition.}
Interestingly, $L_f$-dual-smoothness is not in general equivalent to the corresponding notion along geodesics; a similar phenomenon was observed in the context of Riemannian interior-point methods for the notion of \emph{self-concordance}~\cite{HNW2023}.
We refer to \cref{sub:rel-smooth} for details on relative smoothness.

We briefly comment on other works that offer Riemannian perspectives on mirror descent.
One line of research interprets Euclidean mirror descent as a natural descent algorithm on the Riemannian manifold induced by the distance-generating function.
Alvarez et al.~\cite{Hessian-Riemannian} propose \emph{Hessian Riemannian gradient flow}, a natural descent flow on Riemannian manifolds, as a continuous-time model of Bregman-based descent algorithms, which essentially correspond to mirror descent~\cite{BT2003}.
For the discrete-time algorithm, Raskutti and Mukherjee~\cite{Raskutti-15} point out that mirror descent, expressed in dual coordinates, is equivalent to \emph{natural gradient descent} on the associated Riemannian manifold, a well-known method in information geometry.
From a complementary perspective, Gunasekar et al.~\cite{Gunasekar-21} give a primal space interpretation of mirror descent via discretizations of the Hessian Riemannian gradient flow.
Jiang et al.~\cite{Jiang-26} generalize mirror descent to Riemannian manifolds in another direction: their \emph{Riemannian mirror descent}, as the name suggests, aims to minimize a geodesically convex objective directly defined on a Riemannian manifold.
Our framework is distinct from these approaches: Hadamard mirror descent generalizes the dual expression of mirror descent, minimizing an objective of the form $Q(\diff f(x))$.


\subsection{Convex optimization on moment polytopes}\label{sub:convex-moment}
As we saw concretely in \cref{sub:intro primer}, the bridge between \cref{eq:intro-general-minimization restated} and convex optimization is given by the Kempf--Ness functions.
We now discuss the general setting.
Readers interested primarily in optimization over the entanglement polytopes of tensors can skip ahead to \cref{sub:convex-entanglement}.
Let $G\subseteq \GL(n)$ be a connected Zariski-closed subgroup that is self-adjoint, i.e., $G^\dag=G$.
Then, $P\coloneqq G\cap \PD(n)$ inherits from $\PD(n)$ the structure of a Hadamard manifold; it is a so-called symmetric space with nonpositive curvature~\cite{Helgason1978,Hirai2024}.
We call $P$ a \emph{Hadamard symmetric space}, following \cite{HNW2023}.
Its tangent space at the identity group element~$I \in P$ is naturally identified with the Hermitian part of the Lie algebra, $\bbH \coloneqq \Lie G \cap \Herm(n)$.

Consider a rational representation $\pi \colon G\rightarrow\GL(N)$, and assume $K \coloneqq G \cap \U(n)$ acts unitarily.
We write $g \cdot v \coloneqq \pi(g)v$ for the action of a group element~$g \in G$ on a vector~$v \in \CC^N$.
Generalizing \cref{eq:kn T}, the \emph{Kempf--Ness function} associated with a vector $0\neq v\in\CC^N$ is given by
\[
    f_v \colon P \to \RR,\qquad f_v(x) = \log\langle v,x\cdot v\rangle.
\]
This function is geodesically convex on $P$, and its gradients are captured by the so-called \emph{moment map}~$\mu \colon \CC^N\setminus\{0\} \to \bbH$ associated with the action.
Concretely, for every $x\in P$, $\mu(x^{1/2} \cdot v)$ is the tangent vector dual to $\tau^\ast_{I\rightarrow x}\diff f_v(x)$, where $\tau_{I\rightarrow x} \colon T_I P \to T_x P$ denotes parallel transport.
A celebrated theorem in symplectic geometry~\cite{Kirwan1984,NessMumford1984} states that we obtain a convex polytope if we consider the moment map image of the orbit of~$v$, apply a suitable $K$-invariant function~$s \colon \bbH \to C$ that maps~$\bbH$ into a Euclidean convex cone~$C \subseteq \bbH$ called a \emph{Weyl chamber}, and take the closure.%
\footnote{In \cref{eq:intro-entanglement-polytope}, this operation amounted to computing the eigenvalue spectra of the Hermitian matrix tuple, and the Weyl chamber was given by the tuples of sorted vectors: $C = \RR^{n_1}_\downarrow \times \cdots \times \RR^{n_d}_\downarrow$.}
We discuss this construction in more detail in \cref{sec:momentpolytope}.
The polytope obtained in this way is called the \emph{moment polytope} of~$v$, and it is denoted by~$\Delta(v) \coloneqq \overline{\{ s(\mu(g \cdot v)) \mid g \in G \}}$.

We now consider the problem of convex optimization over the moment polytope:
\begin{align}\label{eq:optimization on moment polytope}
    \minproblem{\phi(p)}{p \in \Delta(v)}.
\end{align}
We assume that the objective $\phi$ appears as a restriction%
\footnote{This condition can be characterized more explicitly:
A convex objective on the moment polytope is the restriction of a $K$-invariant convex function on~$\bbH$ if and only if it extends to a convex function on the \emph{Cartan subspace}~$\AA \coloneqq \mathrm{span}_{\RR}(C)$ that is invariant under the so-called Weyl group; see \cref{sub:objective selfadjoint}.
In the tensor setting of \cref{sub:intro primer,sub:convex-entanglement}, we denote this extension by $\phi$; it is a convex function on $\RR^{n_1}\times\cdots\times\RR^{n_d}$ that is permutation-symmetric in each argument.}
of a $K$-invariant convex objective~$\Phi \colon \bbH \to \overline{\RR}$, so that this can be extended to a holonomy-invariant convex objective $Q \colon T^*P \to \overline{\RR}$ on the cotangent bundle (conversely, any such $Q$ restricts to an objective $\phi$).
Thus, following the same reasoning as in \cref{eq:intro phi to Phi,eq:intro Phi to Q}, we see that \cref{eq:optimization on moment polytope} can be phrased as an instance of the minimization problem~\eqref{eq:intro-general-minimization restated}; cf.~\cite{Hirai2025}:
\begin{align*}
    \min_{p \in \Delta(v)} \phi(p)
= \inf_{g \in G} \phi(s(\mu(g \cdot v)))
= \inf_{g \in G} \Phi(\mu(g \cdot v))
= \inf_{x \in P} Q(\diff f_v(x)).
\end{align*}
Thus, we can use Hadamard mirror descent to solve convex optimization on the moment polytope.
While Hadamard mirror descent is formulated on~$P$, it can equivalently be represented by the following simple update of vectors in~$\CC^N$:
\begin{equation}\label{eq:intro iteration}
    v_{k+1} \coloneqq e^{-\frac\eta2\nabla\Phi(\mu(v_k))}\cdot v_k, \quad v_0 \coloneqq v.
\end{equation}
This update has an appealing feature: it preserves unitary symmetries, which is often important in applications.
For example, if $u \cdot v = v$ for some $u \in K$, then $u \cdot v_k = v_k$ for all $k\in\NN$.
See \cref{prop:symmetry} for the general result.

The gap of the shifted function $f_v^\xi \coloneqq f_v + b^\xi$ that controls the convergence of Hadamard mirror descent can be interpreted as a \emph{capacity} of the type studied in matrix, operator and tensor scaling~\cite{FW2020}.
In the present setting, one can show that it is always finite for suitable~$\xi$.
We show that it is polynomially bounded in the input bit-length of vectors with Gaussian integer (or rational) entries, for broad classes of actions of general linear groups.
We also show that the Kempf--Ness functions are always $D(\pi)^2/2$-dual-smooth, where $D(\pi)$ is a natural complexity parameter called the \emph{weight diameter}%
\footnote{\label{footnote:weight diam vs norm}It can be bounded by twice the \emph{weight norm} of prior work~\cite{BFGOWW2019}. We find a similar improvement for the smoothness of Kempf--Ness functions; see \cref{sub:dualsmooth}.}
of the action (\cref{def:weight diameter}).
In this way we deduce the following result from our general \cref{thm:intro-main}:

\begin{theorem}[Convex optimization on moment polytopes, see \cref{thm:KN-general,cor:KN-Qsmooth}]\label{thm:intro mo po}
    Let $\Phi \colon \bbH\to \overline\RR$ be a l.s.c.\ convex function that is $K$-invariant, differentiable on the interior of its domain, and such that $\mu(g \cdot v)\in\intr\dom\Phi$ for all~$g \in G$; let $\phi$ be its restriction.
    Under a natural step-size condition (satisfied for $\eta\leq2/(D(\pi)^2L)$ if~$\Phi$ is $L$-smooth), Hadamard mirror descent~\eqref{eq:intro iteration} minimizes $\phi$ over the moment polytope $\Delta(v)$ at a rate proportional to the inverse iteration count:
    \[
    0 \leq \phi(s(\mu(v_k))) - \min_{p \in \Delta(v)}\phi(p)
    =   \Phi(\mu(v_k)) - \min_{p \in \Delta(v)}\Phi(p)
    \leq \mathcal O_v\mleft( \frac1{\eta k} \mright).
    \]
    The constant in the $\mathcal{O}_v$ notation captures the logarithm of the norm and a ``capacity'' of the input vector~$v$;
    in typical cases, it is polynomially bounded by its bit-length.
\end{theorem}

\noindent
This establishes that Hadamard mirror descent finds (approximate) minimizers over the moment polytope in a number of steps inversely proportional to the desired accuracy.
We note two caveats regarding this result.
First, it does not directly yield a polynomial-time algorithm in the Turing model because the iteration~\eqref{eq:intro iteration} involves ``continuous'' operations such as matrix exponentials; this requires a detailed precision analysis which we plan to carry out in future work.
Second, it does not encompass fully arbitrary convex optimization problems over moment polytopes, because the $K$-invariance of~$\Phi$ is an additional hypothesis.
However, we will see in the remainder of this introduction that the theorem applies to many natural convex optimization problems, such as those underlying important tensor parameters.

Before considering tensors, we give a simple application of \cref{thm:intro mo po} that addresses an open problem in the literature: quantitative convergence for minimization of the $\ell^2$ or, more precisely, the Frobenius norm $\norm X_{\mathrm F}$ over the moment polytope.
Let $\Phi(X) = \|X\|_{\mathrm{F}}^2/2$, which is 1-smooth and $K$-invariant.
Then, the $Q$-gradient reduces to the ordinary Riemannian gradient of~$f_v$, and the update~\eqref{eq:intro iteration} for the step size~$\eta=2/D(\pi)^2$ reads%
\footnote{Here we assume that $D(\pi)>0$. When $D(\pi)=0$, every group element acts by a scalar multiple of the identity, so the moment polytope is a single point and the algorithmic problem is trivial; moreover, any finite step size is allowed.}
\begin{equation}\label{eq:intro-gradientdescent}
    v_{k+1} \coloneqq e^{-\frac1{D(\pi)^2}\mu(v_k)}\cdot v_k, \qquad v_0 \coloneqq v.
\end{equation}
Then \cref{thm:intro mo po} implies the following result.

\begin{corollary}[Norm minimization on the moment polytope, see \cref{thm:moment-norm-conv}]\label{cor:intro-moment-norm-conv}
    The iteration \eqref{eq:intro-gradientdescent} achieves the following convergence rate, where $D(\pi)$ is the weight diameter\footref{footnote:weight diam vs norm} of the representation:
    \begin{equation*}
        \|\mu(v_k)\|_{\mathrm{F}}^2 - \min_{p \in \Delta(v)}\|p\|_{\mathrm{F}}^2 = \mathcal{O}_v\mleft( \frac{D(\pi)^2} k \mright).
    \end{equation*}
    The constant in the $\mathcal{O}_v$ notation has the same dependence on~$v$ as in \cref{thm:intro mo po}.
\end{corollary}

\noindent
This result improves prior work, which applied only in special cases or did not give quantitative convergence guarantees.
B\"{u}rgisser et al.~\cite{BFGOWW2019} proposed the gradient descent algorithm~\eqref{eq:intro-gradientdescent} with step size~$\eta=1/L$, where $L$ is any upper bound on the smoothness parameter of~$f_v$ (we improve their bound on the smoothness parameter in \cref{prop:kn-smooth}); they showed a best-iterate bound of~$\mathcal{O}(k^{-1})$ assuming the minimum is zero, which is exactly when $f_v$ is bounded below (such~$v$ are called \emph{semistable}).
For the same step sizes, Hirai and Sakabe~\cite{HiraiSakabe2024} showed that $\|\mu(v_k)\|_{\mathrm{F}}^2$ converges to its infimum even if it is positive, but did not provide a quantitative rate.
Our \cref{cor:intro-moment-norm-conv} proves a convergence rate of~$\mathcal{O}(k^{-1})$ that applies in general, regardless of the value of the optimization problem, for the step size~$\eta=2/D(\pi)^2$ arising from the use of the dual-smoothness of the Kempf--Ness function.

We remark that the above also gives an algorithmic proof and interpretation of a widely used semistability criterion by Luna~\cite[Cor.~2]{Luna-75} using conservation of symmetry; see \cref{rem:Luna}.

\subsection{Convex optimization on entanglement polytopes of tensors}\label{sub:convex-entanglement}
We now return to the concrete case of the tensor action.
As discussed in \cref{sub:intro primer}, the moment polytopes in this case are known as \emph{entanglement polytopes}; they are given concretely by \cref{eq:intro-entanglement-polytope}.
This setting has attracted attention due to its applications in algebraic complexity theory, most notably to tensor rank and the exponent of matrix multiplication, and in quantum information, as discussed earlier in this introduction; hence the development of effective algorithms to optimize over entanglement polytopes is a significant open problem.
We can obtain such algorithms from the Hadamard mirror descent framework by specializing the results of the preceding subsection.

For the tensor action, $G = \GL \coloneqq \GL(n_1)\times\cdots\times\GL(n_d)$, $K = \U(n_1)\times\cdots\times\U(n_d)$, $\bbH = \Herm(n_1)\times\cdots\times\Herm(n_d)$, and $\pi(g_1,\dots,g_d) T \coloneqq (g_1 \ot \cdots \ot g_d) T$; the moment map is given by
\begin{equation}\label{eq:intro mu tensors}
    \mu \colon \bigotimes_{\ell = 1}^d\CC^{n_\ell}\setminus\{0\}\rightarrow\bbH,\quad
    \mu(T) = \left(\rho_1,\dots,\rho_d\right), \quad\text{where } \rho_\ell \coloneqq \frac{T^{(\ell)}T^{(\ell)\dag}}{\norm{T}^2}.
\end{equation}
Here, we recall that $T^{(\ell)}$ denotes the $\ell$-th \emph{flattening} of~$T$.
In the language of quantum information, the tensor $T$ represents a (pure) quantum state of $d$ particles, and the $\rho_1,\dots,\rho_d$ describe the quantum states of the individual particles.
The Weyl chamber is given by $C = \RR^{n_1}_\downarrow \times \cdots \times \RR^{n_d}_\downarrow$, and the abstract map $s \colon \bbH \to C$ is simply the map $\spec(\cdot)$ that sends a tuple of Hermitian matrices to their ordered eigenvalue spectra, as in \cref{sub:intro primer}.
Thus, the entanglement polytopes~\eqref{eq:intro-entanglement-polytope} can be written as
$\Delta(T) = \{ \spec(\mu(S)) \mid 0 \neq S \in \overline{\GL \cdot T} \}$,
in agreement with the general formula for the moment polytope given in \cref{sub:convex-moment}.
Because each $\rho_\ell$ is positive semidefinite and has trace one, its eigenvalues form a probability distribution.
Thus, $\Delta(T)$ is a subset of a Cartesian product of probability simplices, as already stated in \cref{sub:intro primer}.

Let us consider convex optimization on the entanglement polytope of a tensor~$T$:
\[
    \minproblem{\phi(p)}{p \in \Delta(T)},
\]
where $\phi \colon \RR^{n_1}\times\cdots\times\RR^{n_d}\to\overline{\RR}$ is a convex objective function that is permutation-symmetric in each argument.
Such a function extends to a unitarily invariant convex function $\Phi \colon \bbH\to \overline{\RR}$ and further to a holonomy-invariant convex objective.
Then the update~\eqref{eq:intro iteration} reads concretely:
\begin{equation}\label{eq:intro iteration tensors}
    T_{k+1} \coloneqq e^{-\frac\eta2\nabla\Phi(\rho_1,\dots,\rho_d)}\cdot T_k \quad
    \text{ where } (\rho_1,\dots,\rho_d) \coloneqq \mu(T_k),
    \quad
    T_0 \coloneqq T.
\end{equation}
The sequence $T_k$ inherits all local unitary symmetries of $T$ (as well as permutation symmetries, for suitable objectives); see \cref{rem:tensor symmetries}.
Since the weight diameter of the tensor action is bounded by $D(\pi) \leq \sqrt{2d}$, \cref{thm:intro-main} along with strong capacity bounds implies the following general~result:

\begin{theorem}[Convex optimization on entanglement polytopes, see \cref{cor:entanglementpolytope,thm:tensor_f+b}]\label{thm:intro ent po}
    Let $0\neq T\in \bigotimes_{\ell=1}^d \CC^{n_\ell}$ be a tensor.
    Let $\phi \colon \RR^{n_1}\times\cdots\times\RR^{n_d}\to\overline{\RR}$ be a l.s.c.\ convex function, permutation-symmetric in each argument and differentiable on the interior of its domain, with $\spec\mu(g \cdot T) \in\intr\dom\phi$ for all~$g \in \GL$.
    Under a natural step-size condition (satisfied for $\eta\leq1/(dL)$ if~$\phi$ is $L$-smooth), Hadamard mirror descent~\eqref{eq:intro iteration tensors} minimizes $\phi$ over the entanglement polytope $\Delta(T)$:
    \[
    0 \leq \phi(\spec(\mu(T_k))) - \min_{p \in \Delta(T)}\phi(p) \leq \mathcal O_T\mleft( \frac1{\eta k} \mright).
    \]
    When $T$ has Gaussian integer entries, the constant in the $\mathcal O_T$ notation is polynomial in its bit-length.
\end{theorem}

\noindent
This and all tensor results discussed in this subsection apply more generally to tensor tuples, but for clarity of exposition we only discuss the case of single tensors.

As a simple example, $\ell^2$-norm minimization corresponds to the objective $\phi(p) = \sum_{\ell}\|p_\ell\|^2/2$.
We obtain the following iteration as a special case of~\eqref{eq:intro-gradientdescent}:
\begin{equation}\label{eq:tensor grad descent}
    T_{k+1} \coloneqq \mleft(\bigotimes_{\ell = 1}^de^{-\rho_\ell/(2d)}\mright) T_k,\quad \text{where}\ (\rho_1, \ldots, \rho_d) = \mu(T_k),
\end{equation}
with stronger guarantees than previously known (see \cref{cor:intro-moment-norm-conv} and the discussion following it).
In the remainder, we present applications to paradigmatic tensor parameters:
the quantum functionals, the symmetric quantum functional, the $G$-stable rank, and the non-commutative~rank.

\subsubsection{Quantum functionals}\label{ssub:quantum_functionals}
In 1986, Strassen introduced the \emph{asymptotic spectrum} of tensors to analyze the computational complexity of multilinear problems, with the exponent of matrix multiplication as a main motivation \cite{Strassen-86,Strassen-88,Strassen-05}.
Central to his theory are \emph{universal spectral points}---functionals on tensors that are normalized on the unit tensor, monotone under restriction, additive under direct sum and multiplicative under Kronecker product.
For decades, constructing non-trivial universal spectral points remained a primary open problem in algebraic complexity theory.

Quantum functionals, introduced by Christandl, Vrana, and Zuiddam~\cite{CVZ2018}, provided the first family of non-trivial universal spectral points for tensors over the complex numbers.
These functionals are defined via entropy maximization over the entanglement polytope, and they are parameterized by probability distributions~$\theta \in \Theta \coloneqq \{ \theta \in \RR^d_{\geq0} \mid \sum_{i=1}^d \theta_i=1\}$.
Namely, the \emph{quantum functional}~$\qfunc_\theta$ with parameter~$\theta$ is, for any nonzero tensor $T\in\bigotimes_{\ell=1}^d \CC^{n_\ell}$, expressed as a maximization problem over the entanglement polytope $\Delta(T)$:
\begin{equation}
\label{eq:intro-def-quantum-functional}
\qfunc_\theta(T) \coloneqq e^{\logqfunc_\theta(T)}, \quad \logqfunc_\theta(T) \coloneqq \max_{p\in\Delta(T)}\; H_\theta(p),
\end{equation}
where $H_\theta(p) \coloneqq H_\theta(p_1,\dots,p_d) \coloneqq \sum_{\ell=1}^d \theta_\ell H(p_\ell)$ and $H(p_\ell) \coloneqq -\sum_{i=1}^{n_\ell} p_{\ell,i} \log p_{\ell,i}$ is the \emph{Shannon entropy};%
\footnote{\label{footnote:qfunc_base}In~\cite{CVZ2018}, the entropy was defined using the binary logarithm. Here we use the natural logarithm, so that $\qfunc_\theta(T) \coloneqq e^{\logqfunc_\theta(T)}$ coincides with the original definition of the quantum functionals.}
this is natural because the components $p_\ell$ are always probability distributions, as observed earlier.
One also puts $F_\theta(0)\coloneqq 0$ by convention.
The quantity $\logqfunc_\theta(T)$ is called the \emph{logarithmic quantum functional}.
The quantum functionals are known to satisfy Strassen's axioms of spectral points, namely, they are monotone under restriction, normalized on unit tensors, additive and multiplicative.
As a result, they also upper bound the asymptotic subrank.

We show that by applying Hadamard mirror descent to the convex objective $\phi = -H_\theta$, we obtain a finite-step descent algorithm for approximating the logarithmic quantum functionals.
Without loss of generality, we may assume that $T$ is \emph{concise}: each component of $\mu(T)$ is nonsingular and hence a positive definite matrix.%
\footnote{\label{footnote:intro-concise}One can always reduce to this situation efficiently, with only a polynomial increase in the bit-length; see \cref{footnote:concise}.}
Then the condition in \cref{thm:intro ent po} that $\spec\mu(g \cdot T) \in \intr\dom\phi$ for all~$g \in \GL$ is satisfied.
The question of step size is more delicate because the negative entropy is \emph{not} $L$-smooth for any~$L>0$.
Instead, we directly show that the natural step-size condition in \cref{thm:intro ent po} is satisfied for~$\eta=1$;
this amounts to a novel \emph{quantum entropy inequality}, which we establish using mathematical techniques from quantum information theory (\cref{thm:quantum}).
We now compute the update rule for this step size.
The unitarily invariant extension of~$\phi$ is given by $\Phi(\rho)\coloneqq-\sum_{\ell=1}^d \theta_\ell H(\rho_\ell)$, where $H(\rho_\ell)$ denotes the \emph{von Neumann entropy}.
We calculate its gradient as $(\nabla\Phi(\rho))_\ell=\theta_\ell(\log \rho_\ell+I_{n_\ell})$ for $\ell\in [d]$.
Thus, up to an irrelevant overall scalar, the update~\eqref{eq:intro iteration tensors} takes the following form, which we call \emph{entropic tensor scaling}:
\begin{equation}\label{eq:intro-entropic}
    T_{k+1} \coloneqq \left( \bigotimes_{\ell=1}^d \rho_\ell^{-\theta_\ell/2} \right) T_k,\quad \text{where } (\rho_1,\dots,\rho_d) = \mu(T_k), \quad T_0 \coloneqq T.
\end{equation}
To the best of our knowledge, entropic tensor scaling provides the first rigorous algorithm to approximate the quantum functionals for arbitrary tensors with a polynomial iteration complexity.
This is formalized in the following theorem, obtained from \cref{thm:intro ent po} and the discussion above.

\begin{theorem}[Quantum functionals, see \cref{thm:entropic tensor scaling}]
\label{thm:intro-quantum-functional}
    Let $0\neq T\in \bigotimes_{\ell=1}^d \CC^{n_\ell}$ be a concise tensor and let $\theta \in \Theta$ be a probability distribution.
	Then, entropic tensor scaling \eqref{eq:intro-entropic}, run for $k\geq1$ steps, returns a tensor~$T_k$ such that
    \[
	   0\leq \logqfunc_\theta(T) - H_\theta(\spec(\mu(T_k))) \leq \mathcal O_T\mleft( \frac1k \mright).
	\]
    When $T$ has Gaussian integer entries, the constant in the $\mathcal O_T$ notation is polynomial in its bit-length.
\end{theorem}

Entropic tensor scaling also provides a new solution to the \emph{tensor scaling problem}, which explains our terminology.
Given a tensor~$T$, this problem asks to find a ``scaling'' $g\cdot T$ such that $\spec(\mu(g \cdot T)) \approx (\bm{1}_{n_1}/n_1, \ldots, \bm{1}_{n_d}/n_d)$ whenever this is possible (such tensors are called \emph{semistable}); equivalently, the quantum marginals $\mu(g\cdot T)\approx(I_{n_1}/n_1, \ldots, I_{n_d}/n_d)$ should be approximately uniform.
This is a natural generalization of the matrix and operator scaling problems, and it finds applications in, e.g., quantum information and tensor networks~\cite{VDDB2003,tensor_canonical_form}.
The problem can be solved by a simple alternating minimization algorithm (called \emph{tensor scaling}) in the spirit of Sinkhorn's algorithm; it was proposed in~\cite{VDDB2003} and analyzed rigorously in~\cite{BFGOWW2018}.
Because uniform distributions are characterized by having maximum entropy, entropic tensor scaling with uniform $\theta$ also solves this problem.
It offers two advantages over the original method:
(1)~it preserves not only the tensor's local unitary but also permutation symmetries between the tensor factors, which is important for applications;
(2)~it also provides meaningful guarantees for tensors that are \emph{unstable}, i.e., not semistable (it computes the quantum functional, whose value quantifies the instability).
The same advantages hold for the gradient descent method~\eqref{eq:tensor grad descent} (which computes the minimum $\ell^2$ norm instead of the maximum entropy) using our \cref{cor:intro-moment-norm-conv}.

One can also use \cref{thm:intro-quantum-functional} to obtain \emph{structural} insight into the nature of the quantum functionals, by using
\begin{equation}\label{eq:qfunc via ets}
   \qfunc_\theta(T) = \lim_{k\rightarrow\infty} e^{H_\theta(\spec\mu(T_k))}
\end{equation}
and the form of the entropic tensor scaling algorithm in \cref{eq:intro-entropic}.
For example, we obtain a new direct proof that the quantum functionals are universal spectral points.
As discussed in \cite{CVZ2018}, it follows directly from the definition and the elementary properties of the moment polytope that $F_{\theta}$ is monotone, normalized on the unit tensor, super-multiplicative, and super-additive.
In contrast, sub-multiplicativity and sub-additivity are rather more involved.
These two properties were proved in \cite{CVZ2018} using a representation-theoretic dual description of the moment polytope.

Entropic tensor scaling allows us to bypass sophisticated representation-theoretic reasoning and prove the multiplicativity and sub-additivity of $F_{\theta}$ as a direct consequence of \cref{thm:intro-quantum-functional}.
Let $S\in\bigotimes_{\ell=1}^{d} \CC^{m_\ell}$ and $T\in\bigotimes_{\ell=1}^d \CC^{n_\ell}$ be two tensors; as discussed, we may assume they~are~concise.
A~direct calculation gives $(S\otimes T)^{(\ell)} = S^{(\ell)}\otimes T^{(\ell)}$ ($\forall \ell\in [d]$) for the flattenings of $S\otimes T$.
Hence,
\[
\mu(S\otimes T) = \left( \mu(S)_1 \otimes \mu(T)_1 , \dots,\mu(S)_d \otimes \mu(T)_d \right),
\]
and we see the entropic tensor scalings~\eqref{eq:intro-entropic} of $S\otimes T$, $S$, and $T$ are related by $(S\otimes T)_{k}=S_k\otimes T_k$.
Consequently,
\[
    \qfunc_\theta(S\otimes T)
= \lim_{k\rightarrow\infty} e^{H_\theta\left(\spec\mu\left( (S\otimes T)_{k} \right)\right)}
= \lim_{k\rightarrow\infty} e^{H_\theta\left(\spec\mu\left( S_{k} \right)\right)+H_\theta\left(\spec\mu\left( T_{k} \right)\right)}
= \qfunc_\theta(S) \qfunc_\theta(T)
\]
using only \cref{eq:qfunc via ets} and the additivity $H(p\otimes q)=H(p)+H(q)$ of the Shannon entropy.
The proof of sub-additivity works in much the same way; we refer to \cref{thm:algorithmicproof} for all formal details.

Two other works have recently given alternative proofs that the quantum functionals are universal spectral points without relying on representation theory.
Hirai~\cite{Hirai2025} uses the continuous-time flow and a duality argument; the work~\cite{SDW2026} shows that the quantum functionals coincide with Strassen's support functionals (which are naturally sub-multiplicative and~sub-additive).

\subsubsection{Symmetric quantum functional}
There is also a symmetric variant of the quantum functionals, called the \emph{symmetric quantum functional} and introduced in~\cite{CFTZ-22}.
It is known to be a spectral point in the asymptotic spectrum of symmetric tensors, that is, it is monotone under symmetric restriction, normalized on the unit tensor, and additive and multiplicative on symmetric tensors.
For symmetric tensors, the symmetric quantum functional reduces to the ordinary quantum functional (for uniform~$\theta$), but in general the two quantities differ.
In the following we show that it can be computed by a simple descent algorithm with polynomial iteration complexity, analogous to \cref{ssub:quantum_functionals}.

We briefly sketch the definition of the symmetric quantum functionals and then present our results.
The relevant group action is the \emph{tensor power action} of $G = \GL(n)$ on $(\CC^n)^{\ot d}$, the space of $d$-tensors with all local dimensions equal to~$n$;
this action is given by $\pi(g) T \coloneqq g^{\ot d} T$.
Then the moment map computes the sum of the quantum marginals in \cref{eq:intro mu tensors}:
\[
    \mu_{\Sym} \colon (\CC^n)^{\ot d} \setminus \{0\} \rightarrow \Herm(n), \quad
    \mu_{\Sym}(T) = \sum_{\ell=1}^d \rho_\ell, \quad\text{where } \rho_\ell \coloneqq \frac{T^{(\ell)}T^{(\ell)\dag}}{\norm{T}^2}.
\]
The corresponding moment polytopes are called \emph{symmetric entanglement polytopes}; they are given by $\Delta_{\Sym}(T) = \{ \spec(\mu_{\Sym}(S)) \mid 0 \neq S \in \overline{\GL(n) \cdot T} \}$.
Their elements are vectors $p\in\RR^{n}_{\geq 0}$ with $\sum_{i=1}^n p_i=d$, so division by~$d$ results in a probability distribution.
The \emph{symmetric quantum functional} is defined by entropy maximization over the distributions thus obtained, in analogy to \cref{eq:intro-def-quantum-functional}:
\begin{equation*}
    \qfunc_{\Sym}(T) \coloneqq e^{\logqfunc_{\Sym}(T)}, \quad \logqfunc_{\Sym}(T) \coloneqq \max_{p\in\Delta_{\Sym}(T)}\; H(p/d),
\end{equation*}
where $H$ denotes the \emph{Shannon entropy};
$\logqfunc_{\Sym}(T)$ is called the \emph{logarithmic symmetric quantum functional}.

We now show that an application of the Hadamard mirror descent framework provides an algorithm to compute the symmetric quantum functional.
Without loss of generality, we may again assume that $\mu_{\Sym}(T)$ is nonsingular;\footnote{In analogy with the quantum functionals, one can always reduce to this situation efficiently; see \cref{footnote:symmetric-concise}.}
this holds, for example, if $T$ is concise.
Then we obtain from~\eqref{eq:intro iteration} the following \emph{symmetric entropic tensor scaling} algorithm:
\begin{equation}\label{eq:intro-sym-entropic}
    T_{k + 1} \coloneqq \left(\rho_\text{avg}^{-1/(2d)}\otimes\cdots\otimes\rho_\text{avg}^{-1/(2d)}\right)T_k, \quad \text{where}\quad\rho_\text{avg} \coloneqq \frac{\mu_{\Sym}(T_k)}d,\qquad T_0 \coloneqq T.
\end{equation}
For symmetric tensors, the above update coincides with the usual entropic tensor scaling~\eqref{eq:intro-entropic} with uniform weights $\theta=(1/d,\dots,1/d)$.
In general, the two iterations differ: instead of using the individual quantum marginals, symmetric entropic tensor scaling uses their \emph{average}.
However, we find that~\eqref{eq:intro-sym-entropic} coincides with entropic tensor scaling for a suitable tensor \emph{tuple} constructed from cyclic shifts of the tensor~$T$, which is useful for the analysis.

Our next theorem gives an iteration bound that is polynomial in the input size of the tensor;
it yields an approximation for the symmetric quantum functional~$\qfunc_{\Sym}(T)$ even if~$T$ is not symmetric.

\begin{theorem}[Symmetric quantum functional, see \cref{thm:symmetric-quantum}]\label{thm:intro-symmetric-quantum-functional}
    Let $0\neq T\in (\CC^n)^{\otimes d}$ be a $d$-tensor such that $\mu_{\Sym}(T)$ is nonsingular.
	Then symmetric entropic tensor scaling \eqref{eq:intro-sym-entropic}, run for $k\geq1$ steps, returns a tensor $T_k$ such that
    \[
        0 \le \logqfunc_{\Sym}(T) - H(\spec(\mu_{\Sym}(T_k))/d) \leq \mathcal O_T\mleft( \frac1k \mright).
    \]
    When $T$ has Gaussian integer entries, the constant in the $\mathcal O_T$ notation is polynomial in its bit-length.
\end{theorem}

\subsubsection{$G$-stable rank}
Next, we consider the $G$-stable rank, a tensor parameter introduced by Derksen~\cite{Derksen-22}.
It provides, for example, a multiplicative approximation to the \emph{slice rank}, a notion that has found numerous applications from algebraic complexity theory to additive combinatorics~\cite{Alon-13,Tao-16,Tao-Sawin-16,Blasiak-17,Kleinberg-18}.
The $G$-stable rank~$\rk^G_\alpha(T)$ of a tensor~$T$ is a measure of its \emph{instability}; it is given by an optimization over all one-parameter subgroups of $\GL \coloneqq \GL(n_1) \times \cdots \times \GL(n_d)$ that asymptotically drive~$T$ to the zero tensor.
It is parameterized by a vector~$\alpha\in\RR^d_{>0}$---the cases $\alpha=(1,\ldots,1)$ and $\alpha=(\frac{1}{n_1},\ldots,\frac{1}{n_d})$ are of particular interest due to their connection to slice rank and semistability, respectively~\cite[Props.~4.9 and~2.6]{Derksen-22}.
For tensors over the complex numbers, the reciprocal of the $G$-stable rank can be viewed as a convex optimization problem over the entanglement polytope~\cite{Derksen-22,Hirai2025}:
\[
    \frac1{\rk^G_\alpha(T)}
= \min_{p\in\Delta(T)}\max_{\ell\in [d]} \frac{\norm{p_\ell}_{\infty}}{\alpha_\ell}
= \min_{p\in\Delta(T)}\max_{\ell\in [d], i \in [n_\ell]} \frac{p_{\ell,i}}{\alpha_\ell}.
\]
While the objective $\phi(p)\coloneqq\max_{\ell\in [d],i \in [n_\ell]} \frac{p_{\ell,i}}{\alpha_\ell}$ is a convex function, it is not smooth and also does not satisfy our differentiability hypotheses.
Accordingly, we cannot apply \cref{thm:intro ent po} directly.
Instead, we apply a \emph{smoothing technique}~\cite{Nesterov-smoothing}; we consider a smooth convex function~$\phi_\delta$ which gives an additive approximation to $\phi$:
\[
    \phi_\delta(p)-\delta \log\mleft(\textstyle\sum_{\ell=1}^d n_\ell\mright) \leq \phi(p) \leq \phi_\delta(p);
\]
the parameter~$\delta>0$ controls both the approximation and the smoothness of the objective~$\phi_\delta$.
By applying Hadamard mirror descent~\eqref{eq:intro iteration tensors} to the smoothed objective~$\phi_\delta$ for sufficiently small~$\delta$ and a suitable step size, we obtain an algorithm that computes an estimate~$\phi_\delta(\spec\mu(T_k))$ for the inverse of the $G$-stable rank to any desired additive error.
We only state the algorithmic result and refer to \cref{sub:G-stable-rk} for the precise iteration and choice of step size.

\begin{theorem}[$G$-stable rank, see \cref{thm:Gstable}]
    Let $0\neq T\in \bigotimes_{\ell=1}^d \CC^{n_\ell}$ be a tensor and let $\alpha \in \RR^d_{>0}$.
    Then, $k = \mathcal O_T(d \log(n_{\Sum}) /(\alpha_{\min}\varepsilon)^2)$ iterations of Hadamard mirror descent approximate~$\rk_\alpha^G(T)^{-1}$ to additive error~$\eps>0$,
    where $n_{\Sum} \coloneqq \sum_{\ell=1}^d n_\ell$ and $\alpha_{\min} \coloneqq \min_\ell \alpha_\ell$.
When $T$ has Gaussian integer entries, the constant in the $\mathcal O_T$ notation is polynomial in its bit-length.
\end{theorem}

\subsubsection{Non-commutative rank}
Let $\mathcal A = (A_1,\dots,A_m)$ be a matrix tuple in~$(\CC^{n \times n})^{\op m}$.
The \emph{non-commutative rank} $\ncrk(\mathcal{A})$ is defined as the rank of the corresponding symbolic matrix in non-commuting indeterminates:
\[
    \ncrk(\mathcal{A}) \coloneqq \rk(A(x)), \qquad A(x) \coloneqq x_1 A_1 + x_2 A_2 +\dots + x_m A_m,
\]
where $\rk(A(x))$ is the matrix rank over the \emph{free skew field} $\CC\ (\!\!\!\!<\!\!x_1,x_2,\dots,x_m\!\!>\!\!\!\!)$ \cite{Cohn1995}.
See also \cite[Thm.~1.17]{GGOW20} and \cite[Def.~2.1]{FSG-22} for several equivalent definitions and motivations for this fundamental and influential notion.
A celebrated result states that the non-commutative rank is computable in deterministic polynomial time, e.g., when the matrix entries are rational numbers encoded in binary~\cite{GGOW20,IQS-18,Hamada-Hirai-21}.
Here we give a new method for computing the non-commutative rank directly by optimization on a suitable moment polytope.
Without loss of generality, we may assume that $\mathcal{A}$ satisfies the following \emph{kernel condition} (see \cref{sub:ncrk}):
\begin{equation}\label{eq:kernel-condition-intro}
    \bigcap_{i \in [m]}\ker A_i = \{0\} \qquad\text{or}\qquad \bigcap_{i \in [m]}\ker A^\dagger_i = \{0\}.
\end{equation}
Then the non-commutative rank admits a formula as a convex optimization problem over the moment polytope, as recently proved by Hirai~\cite{Hirai2025,Hirai-ncrank}:
\[
\frac{2}{n}\, \left(n-\ncrk(\mathcal{A})\right) = \min_{p\in\Delta_{\mathrm{LR}}(\mathcal{A})} \phi(p),\qquad \text{where}\  \phi(p)\coloneqq\norm*{p_1 - \frac{\mathbf{1}_n}{n}}_1 + \norm*{p_2 - \frac{\mathbf{1}_n}{n}}_1.
\]
Here, $\Delta_{\mathrm{LR}}(\mathcal{A})$ denotes the moment polytope of the matrix tuple $\mathcal{A}$ with respect to the \emph{left-right action} of~$\GL(n)\times\GL(n)$.
This is a special case of the tensor action on tensor \emph{tuples}; see \cref{sub:ncrk}.
Unfortunately, as in the case of the $G$-stable rank, the convex objective~$\phi$ is not smooth, but we can again construct, for every $\delta>0$, a natural smoothing $\phi_\delta$ that provides an additive approximation:
\[
    \phi_\delta(p) - 2n\delta \leq \phi(p) \leq \phi_\delta(p).
\]
Hence, by applying Hadamard mirror descent to $\phi_\delta$ for sufficiently small~$\delta$ and for a suitable step size, we can approximate $\ncrk(\mathcal{A})$ to arbitrarily small additive error.
Since $\ncrk(\mathcal{A})$ is always an integer, this suffices to compute it \emph{exactly} by rounding.
In this way we obtain the following theorem.
For concreteness, here we only state it for matrix tuples with Gaussian integer entries:

\begin{theorem}[Non-commutative rank, see \cref{thm:ncrk}]
\label{thm:nc-rank-intro}
    Let $\mathcal{A}$ be an $m$-tuple of $n\times n$ matrices with Gaussian integer entries that satisfies the kernel condition~\eqref{eq:kernel-condition-intro}.
    Then, $\ncrk(\mathcal{A})$ can be computed using $\poly(n,\braket{\mathcal{A}})$ iterations of Hadamard mirror descent applied to the smoothed objective~$\phi_\delta$ on the moment polytope~$\Delta_{\mathrm{LR}}(\mathcal{A})$, with suitable choices of~$\delta$ and~$\eta$, followed by rounding, where $\braket{\mathcal A}$ denotes the bit-length.
\end{theorem}

\noindent
Concretely, the Hadamard mirror descent iterations for the matrix tuples~$\mathcal A_k = (A^{(k)}_1,\dots,A^{(k)}_m)$~read
\begin{equation*}
    A^{(k+1)}_i \coloneqq e^{-\frac{\delta}4 \sigma_L \left(\sigma_L^2 + \delta^2I\right)^{-\frac12}} A^{(k)}_i \, e^{-\frac{\delta}4\sigma_R\left(\sigma_R^2 + \delta^2I\right)^{-\frac12}},
    \quad
    \sigma_L \coloneqq \sum_{j=1}^{m} \tfrac {A_j^{(k)} A_j^{(k)\dagger}} {\norm{\mathcal A_k}^2} - \tfrac I n, \;
    \sigma_R \coloneqq \sum_{j=1}^{m} \tfrac {A_j^{(k)\dagger} A_j^{(k)}} {\norm{\mathcal A_k}^2} - \tfrac I n,
\end{equation*}
where $\delta=1/(2n^2)$, corresponding to the step size~$\eta=\delta/2$.

Our algorithm can be viewed as a conceptual descendant of the operator-scaling-based approach in~\cite{GGOW20}, but it directly computes $\ncrk(\mathcal{A})$, whereas the algorithm in \cite{GGOW20} operates by reduction to the decision problem $\ncrk(\mathcal{A})\geq r$.
See the discussion after \cref{thm:ncrk} for a comparison between the algorithms in \cite{IQS-18,Hamada-Hirai-21} and our algorithm.

\subsection{Organization of the paper}
In \cref{sec:preliminaries}, we collect preliminary results about Hadamard manifolds and the Riemannian geometry of positive definite matrices and Hadamard symmetric spaces obtained from self-adjoint subgroups.
In \cref{sec:generalized-gradients}, we introduce the concept of \emph{holonomy-invariant convex objectives} which we seek to minimize over the gradient set of geodesically convex functions on Hadamard manifolds, and give several key examples.
We also introduce Hirai's $Q$-gradient, whose negative provides a natural descent direction.
\Cref{sec:continuous-time} analyzes the \emph{Hadamard mirror flow}---the natural descent flow defined by the negative $Q$-gradient.
We give an interpretation in terms of mirror flow and extend Hirai's result by proving a general $\mathcal{O}(t^{-1})$ convergence rate.
In \cref{sec:discrete-time}, we introduce and analyze \emph{Hadamard mirror descent}---the discretization of the Hadamard mirror flow---and demonstrate that it generalizes Euclidean mirror descent.
Furthermore, we discuss a natural step-size condition that yields a $\mathcal{O}(k^{-1})$ convergence rate, and we show that this condition follows from a natural and easy-to-verify notion of relative smoothness.
In \cref{sec:momentpolytope}, we apply Hadamard mirror descent to convex optimization on \emph{moment polytopes}.
We discuss Kempf--Ness function and show that Hadamard mirror descent using Kempf--Ness functions allows us to optimize convex objectives on the moment polytope.
We also give a capacity-based lower bound that controls the constant in the convergence bound of Hadamard mirror descent.
We also specialize this discussion to convex optimization on the entanglement polytopes of (tuples of) tensors, as well as to the moment polytopes of homogeneous polynomial actions of the general linear group.
In \cref{sec:applications}, we explore various applications of these general results, such as computing (ordinary and symmetric) quantum functionals, $G$-stable ranks, and non-commutative ranks, and we also obtain a new convergence result for norm minimization on moment polytopes.
In \cref{app:matrix scaling}, we formalize matrix scaling as marginal entropy maximization.
Although this is an example of Euclidean mirror descent, it serves as a simple illustration of how scaling problems can be treated in our framework.
\Cref{app:technical proofs} contains some technical proofs omitted from the main part of the paper.

\section{Preliminaries}\label{sec:preliminaries}
In this paper, we write $[n] \coloneqq \{1, 2, \ldots, n\}$ and~$S_n$ for the symmetric group.
We let~$\CC^\times \coloneqq \CC\setminus\{0\}$.
We denote the extended real line by $\overline\RR \coloneqq \RR \cup \{\infty\}$.
We write $\RR^n_{\downarrow} \coloneqq \{(x_1, \ldots, x_n) \in \RR^n \mid x_1 \ge \cdots \ge x_n\}$ for the closed convex cone of vectors with entries sorted in nonincreasing order.
Let $\VV$ be a finite-dimensional complex vector space.
We denote by~$\Lin(\VV)$ the space of linear endomorphisms, by~$\GL(\VV) \subseteq \Lin(\VV)$ the general linear group, and by $\SL(\VV) \subseteq \GL(\VV)$ the special linear group.
Now suppose $\VV$ carries a Hermitian inner product.
For an operator~$X \in \Lin(\VV)$, we write~$X^\dagger$ for the adjoint, $\norm{X}_{\mathrm{F}} \coloneqq (\tr X^\dagger X)^{1/2}$ for the \emph{Frobenius norm}, which is induced by the Hilbert--Schmidt inner product~$(X,Y) \mapsto \tr X^\dagger Y$, and $\norm{X}_{\tr} \coloneqq \tr{\sqrt{X^\dagger X}}$ for the \emph{trace norm} (also called nuclear norm);
the \emph{operator norm} (or spectral norm) is denoted~$\norm X_{\ope} = \max_{\norm v=1} \norm{X v}$.
We let $\U(\VV) \subseteq \GL(\VV)$ denote the unitary group, $\Herm(\VV) \subseteq \Lin(\VV)$ the Hermitian endomorphisms of~$\VV$, $\PSD(\VV) \subseteq \Herm(\VV)$ the positive-semidefinite operators, and $\PD(\VV) \subseteq \PSD(\VV)$ the positive-definite ones.
For $A, B \in \Herm(\VV)$, we write $A \succeq B$ if $A - B$ is positive semidefinite and $A \succ B$ if $A - B$ is positive definite.
We abbreviate $\GL(n) \coloneqq \GL(\CC^n)$, $\U(n) \coloneqq \U(\CC^n)$, and so forth.
We write $I$ for identity operators and matrices.
The above-mentioned groups naturally admit a Lie group structure.
We denote by $\Lie G$ the Lie algebra of a Lie group~$G$.

\subsection{Hadamard manifolds}
Throughout, let $\MM$ be a \emph{Hadamard manifold}: a complete and simply connected Riemannian manifold with nonpositive sectional curvature;
we refer the reader to \cite[\S{}V.4]{Sakai1996} and \cite[II]{BridsonHaefliger1999} for an introduction to Hadamard manifolds.
Let $\dist(\cdot, \cdot)$ denote the distance induced by the Riemannian metric of~$\MM$.
Any two points~$x,y\in\MM$ are connected by a unique geodesic; its length is equal to~$\dist(x,y)$.
A real-valued function on~$\MM$ is said to be \emph{geodesically convex} if it is convex along any geodesic.
We write $\tau_{\gamma,s\to t}$ for the parallel transport along a parameterized curve~$\gamma$ from~$\gamma(s)$ to~$\gamma(t)$, and $\tau_{x\to y} \colon T_x\MM \to T_y\MM$ for the parallel transport of tangent vectors along the unique geodesic from~$x$ to~$y$.
Then the dual maps $\tau_{\gamma,s\to t}^* \colon T_{\gamma(t)}^*\MM \to T_{\gamma(s)}^*\MM$ and $\tau_{y \to x}^*\colon T_x^*\MM \to T_y^*\MM$ describe the parallel transport of cotangent vectors.
Parallel transport along closed loops is in general nontrivial due to the manifold's curvature.
This is captured by the \emph{holonomy groups}%
\footnote{Because $\MM$ is simply connected, the restricted and the full holonomy groups coincide.}
\[
    \Hol_x(\MM) \coloneqq \{\tau_{\gamma,0\to1} \in \SOr(T_x\MM)\mid \gamma\colon[0,1]\to\MM \text{ piecewise smooth}, \gamma(0) = \gamma(1) = x\},
\]
which act on $T_x\MM$ and $T^*_x\MM$.
We call two geodesic rays (not necessarily of unit speed) $\gamma(t)\coloneqq\exp_x(tv)$, $\beta(t)\coloneqq\exp_y(tu)$ \emph{asymptotic} if $\dist( \gamma(t) , \beta(t) )$ remains bounded for~$t\to\infty$.
This defines an equivalence relation on the set of all geodesic rays, denoted $\gamma\sim \beta$.
Let $C\MM^\infty$ denote the set of equivalence classes.
For any $x \in \MM$, the natural map $T_x\MM\rightarrow C\MM^{\infty},\; u \mapsto [t \mapsto \exp_x(tu)]$ is a bijection, identifying~$C\MM^\infty$ with~$T_x\MM$.
There is a natural norm on $C\MM^\infty$ by defining $\norm{\xi}\coloneqq \norm{\dot\gamma} \equiv \norm{u}$ for any geodesic ray~$\gamma(t) = \exp_x(tu)$ in the equivalence class~$\xi \in C\MM^\infty$.
For a geodesic ray~$\gamma(t)$, we define the (unnormalized) \emph{Busemann function}
\[
    b^\gamma \colon \MM \to \RR, \quad
    b^\gamma(x) \coloneqq \lim_{t\to\infty} \parens*{ \norm{\dot\gamma} \dist(\gamma(t), x) - \norm{\dot\gamma}^2 t },
\]
which is geodesically convex and $\|\dot\gamma\|$-Lipschitz continuous.
Its differential on~$\gamma$ is given by
\begin{equation}\label{eq:diff busemann}
    \diff b^\gamma(\gamma(t)) = -\flat(\dot \gamma(t)),
\end{equation}
where $\flat \colon T\MM \to T^*\MM$ is the natural isomorphism induced by the Riemannian metric~\cite[Prop.~3.1]{HeintzeImHof1977}.
For $\xi = [\gamma] \in C\MM^\infty$, we also define $b^\xi \coloneqq b^\gamma$, but we caution that this is well-defined only up to an additive constant.
In particular, the differential~$\diff b^\xi$, differences~$b^\xi(x) - b^\xi(y)$, and the boundedness and minimizers of functions~$f + b^\xi$ are still well-defined.

\subsection{Positive-definite matrices}\label{sub:pd}
An important example of a Hadamard manifold is $\PD(n)$, the set of $n\times n$ Hermitian positive definite matrices, equipped with the so-called affine-invariant Riemannian metric.
To define the latter, we use that $\PD(n)$ is an open subset of~$\Herm(n)$, and hence at every point we can canonically identify the tangent space $T_x \PD(n) \cong \Herm(n)$; then the \emph{affine-invariant Riemannian metric} is given by the formula~$(X,Y)_x \coloneqq \tr[x^{-1} X x^{-1} Y]$, so $\norm{X}_x = \norm{x^{-1/2} X x^{-1/2}}_{\mathrm{F}}$.
The exponential map on $\PD(n)$ reads~$\exp_x(Z) = x^{1/2} e^{x^{-1/2} Z x^{-1/2}} x^{1/2}$, its inverse is given by $\log_x(y) = x^{1/2} \log(x^{-1/2} y x^{-1/2}) x^{1/2}$ (in particular, $\exp_I$ and $\log_I$ are the matrix exponential and logarithm, respectively), and the distance between two points $x,y\in\PD(n)$ is given by $\dist(x, y) = \norm{\log_x(y)}_x = \norm{\log(x^{-1/2} y x^{-1/2})}_{\mathrm{F}}$.
Note that $x \mapsto gxg^\dagger$ is an isometry of~$\PD(n)$ for any~$g \in \GL(n)$.
The parallel transport of tangent vectors along the geodesic from the identity~$I$ to~$x \in \PD(n)$ is given by~$\tau_{I\to x}(X) = x^{1/2} X x^{1/2}$.
More generally, the parallel transport of tangent vectors along the geodesic from~$x$ to~$y$ is given by $\tau_{x\to y}(X) = x^{1/2} (x^{-1/2} y x^{-1/2})^{1/2} x^{-1/2} X x^{-1/2} (x^{-1/2} y x^{-1/2})^{1/2} x^{1/2}$, and its dual map~$\tau_{x\to y}^*$ gives the parallel transport of cotangent vectors along the geodesic from~$y$ to~$x$.

\subsection{Hadamard symmetric spaces and self-adjoint subgroups}\label{sub:selfconjugate}
We obtain a broad class of Hadamard manifolds, which we call Hadamard symmetric spaces, by intersecting~$\PD(n)$ with a self-adjoint subgroup of~$\GL(n)$.
A Zariski-closed subgroup $G \subseteq \GL(n)$ is said to be \emph{self-adjoint} if~$G^\dagger = G$, i.e., $g^\dagger\in G$ for every~$g\in G$.
This notion goes back to~\cite{mostow1955self} and provides a concrete model for complex reductive algebraic groups; see~\cite[\S{}2.2]{Wallach2017} and \cite{BFGOWW2019} for an introduction (where such subgroups are called \emph{symmetric}).
For example, $\GL(n)$, $\SL(n)$, and the subgroups of diagonal matrices in these are all self-adjoint subgroups, as are their products (see, e.g., \cref{ex:gl} below).

If $G \subseteq \GL(n)$ is a self-adjoint subgroup, the positive definite elements $P \coloneqq G \cap \PD(n)$ form a totally geodesic submanifold of $\PD(n)$~\cite[Thm.~2.12]{Wallach2017}.
In fact, it is a special kind of Hadamard manifold: a symmetric space with nonpositive curvature~\cite{Helgason1978,Hirai2024,HNW2023}.
We will call $P$ a \emph{Hadamard symmetric space}, following \cite{HNW2023}.
Since~$P \subseteq \PD(n)$ is a totally geodesic submanifold, the exponential map and parallel transport on~$P$ match those in $\PD(n)$, and hence they are given by the formulae presented in \cref{sub:pd}.
Let~$K \coloneqq G \cap \U(n)$, a maximal compact subgroup of~$G$.
Then $P \cong G/K$ because $K$ is the stabilizer of $I \in P$ under the transitive action of~$G$ on~$P$ given by~$g \cdot x \coloneqq g x g^\dagger$.
Under our canonical identification~$T_I \PD(n) \cong \Herm(n)$, the tangent space at the identity identifies with the Hermitian elements in the Lie algebra: $T_I P \cong \bbH \coloneqq i\Lie K = \Lie G \cap \Herm(n)$.
Using the Riemannian metric, which at $I \in P$ is simply the Hilbert--Schmidt inner product, we can also identify~$T_I^* P \cong T_I P$, hence~$T_I^* P \cong \bbH$.
In this paper, we consider only \emph{connected} self-adjoint subgroups; their maximal compact subgroups are also connected~\cite[Thm.~2.16]{Wallach2017}.
Then, identifying $T_I P \cong \bbH$, the holonomy group acts by the adjoint action of~$K$:%
\footnote{In general, the holonomy group at the identity~$\Hol_I(P)$ is a Lie group whose Lie algebra is given by the adjoint action of $[\Lie K, \Lie K] \subseteq \Lie K$; see \cite[III~(6.7), IV~Prop.~6.7 and Thm.~6.9]{Sakai1996}.
For the compact Lie group~$K$, its Lie algebra decomposes as a direct sum $\Lie K = [\Lie K, \Lie K] \op \mathfrak z$ of its semisimple part and center.
Thus, the adjoint action of $[\Lie K, \Lie K]$ coincides with the adjoint action of $\Lie K$, and it follows that $\Hol_I(P)$ is given by the adjoint action of the identity component of~$K$, which, for connected~$K$, is all of~$K$.}
\[
    \Hol_I(P) = \{ \bbH \ni X \mapsto u X u^\dagger \in \bbH \mid u \in K \}.
\]

Choose a maximal commuting subspace~$\AA\subseteq\bbH$; 
this is called a \emph{Cartan subspace}.
Every adjoint orbit~$\{uXu^\dagger\mid u\in K\}$ intersects~$\AA$; see~\cite[Cor.~2.20]{Wallach2017}.
In general, this intersection contains more than one element, but it is always finite, as we will discuss now.
Define the \emph{Weyl group} as $W\coloneqq N_K(\AA)/Z_K(\AA)$, where $N_K(\AA)\coloneqq \{u\in K\mid u\AA u^\dagger=\AA\}$ is the normalizer and $Z_K(\AA)\coloneqq \{u\in K\mid uXu^\dagger=X \; (\forall X\in\AA)\}$ the centralizer of~$\AA$.
The Weyl group acts naturally on~$\AA$: we write $w \cdot X$ for this action, which is induced by the adjoint action of~$K$.
Then there is a bijection between the $K$-orbits in~$\bbH$ and the $W$-orbits in~$\AA$:
for any~$X \in \AA$, we have $\{u X u^\dagger \mid u \in K \} \cap \AA = W \cdot X$~\cite[Prop.~3.61]{Wallach2017}.
Remarkably, the Weyl group is finite~\cite[Cor.~3.60]{Wallach2017}, hence this intersection is finite.
In fact, it is a finite reflection group, so one can always choose as a fundamental domain a closed, convex cone~$C\subseteq\AA$ called a \emph{(closed) Weyl chamber}.
Thus, every Weyl group orbit and hence every adjoint orbit~$\{u Xu^\dagger \mid u \in K\}$ intersects the Weyl chamber~$C$ in a single point, which we denote by~$s(X)$.
This defines a continuous surjection $s\colon\bbH\rightarrow C$.


We now describe these notions concretely in the following important example.

\begin{example}\label{ex:gl}
    The group $\GL(n_1)\times\cdots\times\GL(n_d)$ can be naturally regarded as a subgroup of~$\GL(N)$, where $N \coloneqq n_1+\cdots+n_d$, via the block-diagonal embedding.
    As such, it is clearly Zariski-closed and self-adjoint.
    In this case, we have $P = \PD(n_1)\times\cdots\times\PD(n_d) \subseteq \PD(N)$, $K = \U(n_1)\times\cdots\times\U(n_d) \subseteq \U(N)$, and $\bbH = \Herm(n_1)\times\cdots\times\Herm(n_d) \subseteq \Herm(N)$.
    We take as a Cartan subspace the tuples of diagonal Hermitian matrices: $\AA=\diag(\RR^{n_1})\times\cdots\times\diag(\RR^{n_d}) \subseteq \bbH$.
    The Weyl group can be identified with the tuples of permutation matrices: $W \cong S_{n_1}\times\cdots\times S_{n_d}$.
    This reflects the fact that two diagonal matrices are conjugate under~$\U(n)$ if and only if their diagonal entries are permutations of each other.
    Hence, we can choose as a Weyl chamber $C = \diag(\RR^{n_1}_{\downarrow})\times\cdots\times\diag(\RR^{n_d}_{\downarrow}) \subseteq \AA$.
    Then, the map $s \colon \bbH \to C$ sends a tuple of Hermitian matrices to the tuple of diagonal matrices in the Weyl chamber with the same spectra.
    In other words, $s = (\diag,\dots,\diag)\circ\spec$, where $\spec \colon \bbH \to \RR^{n_1}_{\downarrow}\times\cdots\times\RR^{n_d}_{\downarrow}$ sends a tuple of Hermitian matrices to their spectra, with the eigenvalues listed in nonincreasing order.
\end{example}

We end this section with a lemma about asymptotic geodesic rays in~$P$.

\begin{lemma}\label{lem:asymptoticrays}
    Any geodesic ray $t \mapsto g^\dagger e^{tX}g$ (where $g \in G, X \in \bbH$) on $P$ is asymptotic to $t \mapsto e^{tu^\dagger Xu}$ for some~$u \in K$.
\end{lemma}

\noindent
The proof of the lemma can be found in \cref{app:technical proofs}.

\section{Holonomy-invariant objectives and generalized gradients}\label{sec:generalized-gradients}
In this section, we recall Hirai's notion of $Q$-gradient in a slight variant of his setup~\cite{Hirai2025}.
The context is the problem of minimizing $Q(\diff f)$, where $f$ is a geodesically convex function on a Hadamard manifold, and $Q$ is an objective function on the cotangent bundle.
Roughly speaking, $Q$ should be convex on each cotangent space and invariant under holonomies; we define the precise conditions in \cref{sub:objective}.
Then, under suitable regularity assumptions, a generalized gradient called the $Q$-gradient of~$f$ can be defined so that~$Q(\diff f)$ is nonincreasing under the corresponding flow; see \cref{sub:generalized_gradient}.%
\footnote{\label{footnote:Q}The main difference between our setup and that of \cite{Hirai2025} is as follows:
Hirai considers the problem of minimizing $Q(\diff f)$, but defines the associated $Q$-gradient using $Q^2/2$ instead of $Q$.
Since he additionally assumes that $Q$ is nonnegative (``norm-like''), minimizing $Q$ and $Q^2/2$ are equivalent.
Our $Q$ plays the role of his $Q^2/2$, which we find slightly more convenient.
In particular, we do not require that $Q$ is nonnegative.}
Throughout this section, $\MM$ denotes a Hadamard manifold.

\subsection{Holonomy-invariant convex objectives}\label{sub:objective}
We consider the following class of objectives.

\begin{definition}[Holonomy-invariant convex objective]\label{def:holoinvariant}
A \emph{holonomy-invariant convex objective} is a function $Q \colon T^*\MM \to \overline\RR$ on the cotangent bundle satisfying the following properties:
\begin{enumerate}[leftmargin=4em,noitemsep,label={(Q\arabic*)}]
\item\label{it:Q1}
For every $x\in\MM$, its restriction $Q_x \coloneqq Q|_{T_x^*\MM}$ is a lower-semicontinuous (l.s.c.) and Euclidean convex function on the cotangent space $T^*_x\MM$.
\item\label{it:Q2}
$Q$ is invariant under parallel transport in the following sense:
for every piecewise smooth curve $\gamma \colon [0,1]\rightarrow\MM$, we have $Q_{\gamma(1)} = Q_{\gamma(0)} \circ \tau^*_{\gamma,0\to1}$.
\end{enumerate}
To later define an associated gradient (as opposed to subgradient) notion, we further require:
\begin{enumerate}[leftmargin=4em,noitemsep,label={(Q\arabic*)},resume]
\item\label{it:Q3} For every $x \in \MM$, $\dom Q_x \coloneqq \{\alpha \in T_x^*\MM \mid Q_x(\alpha) < \infty\}$ has a nonempty interior, and~$Q_x$ is differentiable on $\intr\dom Q_x$.
\end{enumerate}
Such an objective is called \emph{$L$-smooth} if its restrictions~$Q_x$ are $L$-smooth, i.e., if each $Q_x$ is finite-valued and its differential is $L$-Lipschitz continuous with respect to the Riemannian metric and its dual norm.%
\footnote{The differential of $Q_x \colon T^*_x\MM \to \overline\RR$ at a point $\alpha \in \intr\dom Q_x$ can be identified with a tangent vector: $\diff Q_x(\alpha) \in T^*_\alpha(T^*_x\MM) \cong T_x\MM$.
Then $Q_x$ is $L$-smooth iff $\dom Q_x = T^*_x\MM$ and $\norm{\diff Q_x(\alpha) - \diff Q_x(\beta)}_x \leq L \norm{\alpha-\beta}^*_x$ for all $\alpha,\beta\in T^*_x\MM$.}
\end{definition}

The name \emph{holonomy-invariant} originates from the following observation:
By property~\ref{it:Q2}, $Q$~is uniquely defined by its restriction~$Q_x$ to the cotangent space at an arbitrary point~$x\in\MM$ and, moreover, invariant under the dual action of the holonomy group~$\Hol_x(\MM)$.
Conversely, if for an arbitrary fixed~$x\in\MM$ one has a function $Q_x \colon T^*_x\MM \to \overline{\RR}$ that satisfies the following conditions, then it must be the restriction of a unique holonomy-invariant convex objective:
\begin{enumerate}[leftmargin=4em,noitemsep,label={($Q_x$\arabic*)}]
\item\label{it:Q1'} $Q_x$ is l.s.c.\ and convex.
\item\label{it:Q2'} $Q_x$ is invariant under the holonomy group: $Q_x = Q_x\circ h^*$ for any $h \in \Hol_x(\MM)$.
\item\label{it:Q3'} $\dom Q_x$ has a nonempty interior and $Q_x$ is differentiable on $\intr \dom Q_x$.
\end{enumerate}
Concretely, its extension $Q\colon T^*\MM \to \overline\RR$ can be defined by $Q_y \coloneqq Q_x \circ \tau_{x \to y}^*$.
Moreover, $Q$ is $L$-smooth if and only if this is the case for its restriction~$Q_x$ at an arbitrary point~$x\in\MM$.

\begin{example}\label{eq:norm squared}
For any Hadamard manifold~$\MM$, we can take $Q = \frac12\norm{\cdot}^2$, with $\norm{\cdot}$ the dual norm on the cotangent bundle induced by the Riemannian metric.
In this case, $Q$ is $1$-smooth,
and the $Q$-gradient defined in \cref{def:grad} below coincides with the ordinary Riemannian gradient.
\end{example}

\begin{example}\label{ex:euclidean}
If $\MM = \RR^n$ is a Euclidean space, we can identify all tangent spaces with~$\RR^n$, and parallel transport is trivial; thus \ref{it:Q2} simply requires all~$Q_x$ to be the same.
Thus a holonomy-invariant convex objective is parameterized by a single convex function~$Q \colon (\RR^n)^* \to \overline\RR$.
After identifying $\RR^n$ with its dual, we obtain a function $\phi \colon \RR^n \to \overline\RR$, as in \cref{sec:intro}.
See also \cref{ex:euclidean as selfadjoint} below.
\end{example}

We end this subsection with the following observation: properties~\ref{it:Q1} and \ref{it:Q2} imply that if we evaluate~$Q$ on the differential of a Busemann function, the function value is independent of the point on $\MM$.

\begin{lemma}\label{lem:Q-busemann}
    Let $Q: T^*\MM \to \overline{\RR}$ be a function that satisfies \ref{it:Q1} and \ref{it:Q2}.
    Then for any pair of asymptotic geodesic rays $\exp_x(t u)\sim \exp_y(tv)$ we have $Q_x(\flat u)=Q_y(\flat v)$, where $\flat \colon T\MM \to T^*\MM$ is the isomorphism induced by the Riemannian metric.
    Hence, for $\xi\in C\MM^\infty$ and $b^\xi$ an associated~Busemann~function,
    \begin{equation}\label{eq:Fxi}
        Q(\xi)\coloneqq Q_x(-\diff b^\xi(x))
    \end{equation}
    is well-defined (i.e., independent of the choice of $x$).
\end{lemma}

\noindent We prove \cref{lem:Q-busemann} in \cref{app:technical proofs}.

\begin{example}\label{ex:Q_xi}
To justify the ``$Q(\xi)$'' notation, consider the case of $\MM = P \coloneqq G \cap \PD(n)$.
As~$T_I P = \bbH$, every geodesic ray is asymptotic to one of the form $t \mapsto e^{tX}$ ($X \in \bbH$), and hence each equivalence class $\xi \in CP^\infty$ is represented by such a geodesic; this leads to a natural identification $CP^\infty \cong \bbH$.
Under this identification, $Q(\xi)$ equals $\Phi(X)$ (where $\Phi$ is defined in the next subsection), which justifies the notation.
\end{example}

\subsection{Objectives on Hadamard symmetric spaces}\label{sub:objective selfadjoint}
We now specialize the discussion to the Hadamard manifolds obtained from a connected self-adjoint subgroup~$G \subseteq \GL(n)$.
Recall from \cref{sub:selfconjugate} that the symmetric space~$P \coloneqq G\cap\PD(n)$ is a Hadamard manifold.
We also saw that the cotangent space at the identity element is canonically identified with $\bbH \coloneqq \Lie G \cap \Herm(n)$, with the action of the holonomy group corresponding to the adjoint action of $K \coloneqq G \cap \U(n)$.
Thus, holonomy-invariant objectives can be specified by functions $\Phi \colon \bbH \to \overline{\RR}$ satisfying the following conditions, which are obtained by translating~\ref{it:Q1'}--\ref{it:Q3'}:
\begin{enumerate}[leftmargin=4em,noitemsep,label={($\Phi$\arabic*)}]
    \item\label{it:Phi1} $\Phi$ is l.s.c.\ and convex.
    \item\label{it:Phi2} $\Phi$ is $K$-invariant: $\Phi(u X u^\dagger) = \Phi(X)$ for all $u \in K$ and $X \in \bbH$.
    \item\label{it:Phi3} $\dom\Phi$ has a nonempty interior and $\Phi$ is differentiable on $\intr\dom\Phi$.
\end{enumerate}
The function $\Phi$ coincides with $Q_I$ under the identification $\bbH \cong T^*_I P$.

By condition~\ref{it:Phi2}, $\Phi$ is uniquely determined by its restriction to a Weyl chamber, and any function on a Weyl chamber extends to a unique function~$\Phi$ satisfying~\ref{it:Phi2}.
However, if we start with a convex function on the Weyl chamber, its extension will \emph{not} in general be convex.
For example, if~$P=\PD(n)$ then $\bbH=\Herm(n)$, and a Weyl chamber is given by $C\cong\RR_\downarrow^n$.
Thus, $x \mapsto -x_1$ is a convex (even linear) function on~$C$, but its extension to~$\Herm(n)$ is given by $\Phi(X) = -\lambda_{\max}(X)$, which is non-convex for $n\geq2$.
In fact, the obstruction to convexity can already be seen by considering diagonal matrices or, equivalently, real vectors (the function $x \mapsto -\max\{x_1,\dots,x_n\}$ is nonconvex on~$\RR^n$).
This observation is general: if, instead, we start with a convex Weyl-group invariant function on the entire Cartan subspace, then its extension will always be convex.

We will now formalize and prove this observation.
Consider the restriction~$\phi \coloneqq \Phi|_{\AA}$ of a function~$\Phi \colon \bbH \to \overline\RR$ to the Cartan subspace.
Then, the conditions~\ref{it:Phi1}--\ref{it:Phi3} on $\Phi$ respectively imply conditions on $\phi$ as follows:
\begin{enumerate}[leftmargin=4em,noitemsep,label={($\phi$\arabic*)}]
    \item\label{it:F1} $\phi$ is l.s.c.\ and convex.
    \item\label{it:F2} $\phi$ is invariant under the Weyl group: $\phi(w \cdot X) = \phi(X)$ for all $X \in \AA$ and $w \in W$.
    \item\label{it:F3} $\dom\phi$ has a nonempty interior and $\phi$ is differentiable on $\intr \dom\phi$.
\end{enumerate}
Furthermore, if $\Phi$ is $L$-smooth, then so is~$\phi$.

Conversely, given a function~$\phi$ satisfying these properties, we can always extend it to a function~$\Phi$ satisfying the corresponding properties, and the smoothness parameter is also preserved.
This is the content of the following proposition.
Its proof draws on the general results of~\cite{lewis1996group}.%
\footnote{Lewis considers the abstract setting of ``normal decomposition systems''---in our case~$(\bbH,K,s)$, where~$K$ acts by conjugation, and~$(\AA,W,s|_{\AA})$.
He explicitly discusses the case of real semisimple groups in~\cite[\S{}4]{lewis1996group}; in particular, this applies when our~$G$ is complex semisimple, viewed as a real group.
The complex reductive case follows by splitting off the center, on which~$K$ acts trivially;
alternatively, it follows from the compact Lie group case discussed in \cite[Ex.~1.3]{tam1999extension}, as noted in \cite[remark following Prop.~3.4]{lewis2000convex}.
The important case~$\bbH=\Herm(n)$ is also treated in \cite{Lewis1996}.}

\begin{proposition}\label{prop:FandQ}
Let $\phi \colon \AA \to \overline\RR$ be a function that satisfies \ref{it:F1}--\ref{it:F3}.
Define its $K$-invariant extension $\Phi \colon \bbH\to \overline{\RR}$ by~$\Phi(X)\coloneqq \phi(s(X))$.
Then:
\begin{enumerate}
    \item[(a)] The function $\Phi$ satisfies \ref{it:Phi1}--\ref{it:Phi3}, and hence it determines a holonomy-invariant convex objective~$Q$.
    \item[(b)] Furthermore, $s^{-1}(\intr\dom\phi)=\intr\dom\Phi$, and the gradient of~$\Phi$ at any~$X \in \intr\dom\Phi$ can be computed as follows: $\nabla\Phi(X) = u^\dagger \nabla\phi(u X u^\dagger) u$ for any $u \in K$ such that $u X u^\dagger \in \AA$.
    \item[(c)] If $\phi$ is additionally $L$-smooth, then $\Phi$ is $L$-smooth, and hence so is~$Q$.
\end{enumerate}
\end{proposition}
\begin{proof}
(a)~Property \ref{it:Phi2} holds by definition of~$\Phi$, because~$s$ is constant on adjoint orbits.
Because $\phi$ is Weyl group invariant, $\phi$ is convex if and only if~$\Phi$ is convex, and $\phi$ is l.s.c.\ if and only if~$\Phi$ is l.s.c., by \cite[Thm.~4.3]{lewis1996group}.
Thus, \ref{it:F1} and \ref{it:F2} imply property \ref{it:Phi1}.
Now, $s^{-1}(\dom\phi) = \dom\Phi$.
Because $\dom\phi$ is Weyl group invariant, it follows that $s^{-1}(\intr\dom\phi) = \intr\dom\Phi$ by~\cite[Thm.~5.4~(iii)]{lewis1996group}.
This implies that $\dom\Phi$ has nonempty interior; this is because the interior of $\dom\phi$ is nonempty and Weyl group invariant, hence it meets the Weyl chamber~$C$, and $s$ is surjective onto~$C$.
Moreover, we obtain that $\Phi$ is differentiable on $\intr\dom\Phi$.
This follows from the preceding and \cite[Thm.~6.1]{lewis1996group}, which states that if a convex and Weyl group invariant function~$\phi$ is differentiable at~$s(X)$, then $\Phi$ is differentiable at~$X$.
Thus, \ref{it:F1}, \ref{it:F2} and \ref{it:F3} imply property \ref{it:Phi3}.

(b)~We already proved the statement about the interiors of the domains.
The formula for the gradient follows from~\cite[Thm.~6.1]{lewis1996group}, which states it in the special case where~$u X u^\dagger \in C$; the general case follows from the $W$-invariance of~$\phi$, which implies that~$\nabla\phi$ is $W$-equivariant.

(c)~The convex function $\phi$ is $L$-smooth iff $L\norm{\cdot}_{\mathrm F}^2/2 - \phi(\cdot)$ is convex.
By applying part~(a) to $L\norm{\cdot}_{\mathrm F}^2/2-\phi(\cdot)$, we see that $L\norm{\cdot}_{\mathrm F}^2/2-\Phi(\cdot)$ is convex.
This implies that $\Phi$ is also $L$-smooth.
\end{proof}

\noindent
Despite its simple proof, part~(c) seems little known even in the case of~$\PD(n)$.
For instance, it appears that the smoothness analysis of functions of matrices in \cite[\S{}2]{Nesterov-smoothing} could be simplified by using this property in the case of the Frobenius norm.

We specialize the above to the important symmetric space of positive definite matrix tuples.

\begin{example}\label{ex:P phi}
    Consider the Hadamard symmetric space $P = \PD(n_1)\times\cdots\times\PD(n_d)$, as in \cref{ex:gl}.
    Then, holonomy-invariant convex objectives~$Q$ are parameterized by functions
    $\Phi \colon \Herm(n_1)\times\cdots\times\Herm(n_d) \to \overline\RR$
    satisfying the following properties:
    \begin{enumerate}[leftmargin=4em,noitemsep,label={($\Phi$\arabic*)}]
        \item $\Phi$ is l.s.c.\ and convex.
        \item $\Phi$ is unitarily invariant: $\Phi(u_1 X_1 u_1^\dagger, \dots u_d X_d u_d^\dagger) = \Phi(X_1,\dots,X_d)$ for all $u_\ell \in \U(n_\ell)$ and~$X_\ell \in \Herm(n_\ell)$.
        \item $\dom\Phi$ has a nonempty interior and $\Phi$ is differentiable on $\intr\dom\Phi$.
    \end{enumerate}
    Clearly, such functions are uniquely determined by their behavior on the Weyl chamber~$C$, which consists of tuples of diagonal matrices with nonincreasing diagonal entries.

    \Cref{prop:FandQ} gives a useful characterization in terms of their restriction to the Cartan subspace~$\AA$, which consists of all tuples of diagonal matrices.
    Identifying $\AA \cong \RR^{n_1} \times \cdots \times \RR^{n_d}$ and using the fact that the Weyl group acts by independent permutations of each block, we arrive at the following characterization:
    holonomy-invariant convex objectives are parameterized by functions
    $\phi \colon \RR^{n_1} \times \cdots \times \RR^{n_d} \to \overline\RR$
    such that
    \begin{enumerate}[leftmargin=4em,noitemsep,label={($\phi$\arabic*)}]
        \item\label{it:phi1gl} $\phi$ is l.s.c.\ and convex.
        \item\label{it:phi2gl} $\phi$ is symmetric in each argument: $\phi(P_1p_1, \ldots, P_dp_d) = \phi(p_1, \ldots, p_d)$ for all~$p_\ell \in \RR^{n_\ell}$ and all permutation matrices~$P_\ell$.
        \item\label{it:phi3gl} $\dom\phi$ has a nonempty interior and $\phi$ is differentiable on $\intr \dom\phi$.
    \end{enumerate}
    Furthermore, the corresponding objective~$Q$ is $L$-smooth if and only if~$\phi$ is $L$-smooth.
\end{example}

We now give some concrete examples of holonomy-invariant convex objectives on~$\PD(n)$, which form the building blocks of the objectives underlying our main applications.

\begin{example}\label{ex:F examples}
Let $\MM=\PD(n)$.
By \cref{ex:P phi}, we can define holonomy-invariant convex objectives in terms of functions $\phi \colon \RR^n \to \overline\RR$.
We give three examples (for each, it is clear that~\ref{it:F1}--\ref{it:F3} hold) along with their unitarily invariant extensions~$\Phi \colon \Herm(n) \to \overline\RR$:
\begin{enumerate}
\item[(a)] The \emph{(negative) Shannon entropy} and its unitarily invariant extension, the \emph{(negative) von Neumann entropy} (by common convention denoted by the same symbol):
\begin{align*}
-H(x) \coloneqq
\begin{cases}
    \sum_{j=1}^n x_j \log x_j & \text{ if } x \in \RR^n_{\geq0},\\
    \infty & \text{ otherwise,}
\end{cases}
\qquad
-H(X) \coloneqq
\begin{cases}
    \tr[X \log X] & \text{ if } X \succeq 0,\\
    \infty & \text{ otherwise,}
\end{cases}
\end{align*}
with the usual convention that $0\log0=0$.

\item[(b)] For any fixed $\delta>0$, the functions
\[
    \phi_\delta(x) \coloneqq \delta\log \sum_{j=1}^n e^{x_j/\delta},
    \qquad
    \Phi_\delta(X) \coloneqq \delta\log \tr e^{X/\delta}
\]
provide smooth approximations to the $\ell^\infty$-norm on the nonnegative orthant and to the operator norm on the positive semidefinite matrices, respectively.
\item[(c)] For any fixed $\delta>0$, the functions
\[
    \phi_\delta(x) \coloneqq  \sum_{j=1}^n \sqrt{\left(x_j - \tfrac1n\right)^2 + \delta^2},
    \qquad
    \Phi_\delta(X) \coloneqq \tr\sqrt{\left(X - \frac{I}n\right)^2 + \delta^2I}
\]
provide smooth approximations to the $\ell^1$-distance to the uniform distribution and to the trace-norm distance to the normalized identity, respectively.
\end{enumerate}
These objectives are used as building blocks for the quantum functionals (\cref{sub:quantum,sub:symmetric_quantum}), the $G$-stable rank (\cref{sub:G-stable-rk}) and the non-commutative rank (\cref{sub:ncrk}), respectively.
\end{example}

\begin{example}\label{ex:euclidean as selfadjoint}
The Euclidean setting also arises as a special case of the setup of this section.
Let~$G \coloneqq \diag((\CC^\times)^n) \subseteq \GL(n)$ denote the group of invertible diagonal matrices, a Zariski-closed self-adjoint subgroup.
Then~$P = \diag(\RR_{>0}^n)$.
We can identify~$P$ with~$\RR_{>0}^n$, with the log-Euclidean metric~$\dist(x,y) = \norm{\log x - \log y}_2 = \sqrt{\sum_{j=1}^n (\log x_j - \log y_j)^2}$, and further with the Euclidean space~$\RR^n$ (using the entrywise logarithm).
Thus, $\bbH \cong \RR^n$.
Since $G$ is commutative, $\AA = \bbH$, $\Phi = \phi$, and the adjoint action and Weyl group are trivial; hence~\ref{it:Phi2} and~\ref{it:F2} are vacuous.
Thus we see that objectives are specified by functions $\phi \colon \RR^n \to \overline\RR$ satisfying conditions~\ref{it:F1} and~\ref{it:F3};
we have recovered \cref{ex:euclidean}.
\end{example}

\subsection{Generalized gradient}\label{sub:generalized_gradient}
The goal of this paper is to solve minimization problems of the following kind:
\begin{equation}\label{eq:general_minimization}
    \minproblem{Q(\diff f(x))}{x \in \MM},
\end{equation}
where $f$ is a geodesically convex function and~$Q$ is an objective as above.
We recall the central notion introduced by Hirai.\footref{footnote:Q}

\begin{definition}[$Q$-gradient, \cite{Hirai2025}]\label{def:grad}
Let $Q \colon T^*\MM\to\overline\RR$ be a holonomy-invariant convex objective and let $f \colon \MM\rightarrow\RR$ be a continuously differentiable geodesically convex function such that $\diff f(x) \in \intr\dom Q_x$ for all~$x\in \MM$.
The \emph{$Q$-gradient of $f$} is the continuous vector field, denoted~$\nabla^Q f$, and defined as follows:
\begin{equation}\label{eq:grad}
    \nabla^Q f(x)
\coloneqq
  \diff(Q_x)(\diff f(x))
\in T_x^{**}\MM \cong T_x\MM
\qquad (x \in \MM).
\end{equation}
\end{definition}

\noindent
Note that $\diff (Q_x)$ is a covector field on $\intr\dom Q_x \subseteq V \coloneqq T_x^*\MM$, hence $\diff(Q_x)(\diff f(x)) \in T_{\diff f(x)}^* V$; because $T_{\diff f(x)}^* V \cong V^* \cong T_x^{**}\MM \cong T_x\MM$, this can be canonically identified with a tangent vector in~$T_x\MM$.

The following lemma reveals the fundamental nature of this definition:
by the geodesic convexity of~$f$, the negative $Q$-gradient provides a natural candidate descent direction for problem~\eqref{eq:general_minimization}.

\begin{lemma}[{cf.\ \cite[Lem.~3.9]{Hirai2025}}]\label{lem:dFdf}
Let $Q$ and $f$ be as in \cref{def:grad}, with $f$ twice differentiable.
For any continuously differentiable curve $x_t$ in~$\MM$, we have
\begin{align*}
    \frac \diff {\diff t} Q(\diff f(x_t)) = \nabla^2 f(x_t)[\dot x_t, \nabla^Q f(x_t)].
\end{align*}
In other words:
\[ \diff(Q \circ \diff f)(x) = \nabla^2 f(x)[\cdot, \nabla^Q f(x)]  \qquad (x \in \MM). \]
\end{lemma}
\begin{proof}
We recall the argument in the proof of \cite[Lemma~3.9]{Hirai2025}.
We have
\begin{align*}
    \frac{\diff}{\diff s}\Big|_{s=t} Q(\diff f(x_s))
&=   \frac{\diff}{\diff s}\Big|_{s=t} Q_{x_t}(\tau^*_{x_\bullet,t\to s} \diff f(x_s))
=   \diff(Q_{x_t})(\diff f(x_t)) [ \nabla_{\dot x_t} \diff f ] \\
&=   (\nabla_{\dot x_t} \diff f)[\nabla^Q f(x_t)]
=   (\nabla \diff f)(x_t)[\dot x_t, \nabla^Q f(x_t)]
=   \nabla^2 f(x_t)[\dot x_t, \nabla^Q f(x_t)],
\end{align*}
where the first equality follows from property~\ref{it:Q2} (the invariance of~$Q$ under parallel transport), the second from the chain rule and the definition of the covariant derivative, the third uses the definition of the $Q$-gradient (\cref{eq:grad}) along with the identification $T_{x_t} \MM \cong T^{**}_{x_t} \MM$, and we conclude by rewriting using the definition of the covariant derivative and Hessian.
\end{proof}

We end this section with a regularity lemma that will be used below.
We omit its proof, which is standard.

\begin{lemma}\label{lem:C1 Q-gradient}
Let $f$ and $Q$ be as in \cref{def:grad}.
Suppose that $f$ is twice continuously differentiable on~$\MM$ and that $Q_x$ is twice continuously differentiable on~$\intr\dom Q_x$ for every (equivalently, some)~$x \in \MM$.
Then the $Q$-gradient~$\nabla^Q f$ is a continuously differentiable vector field.
In particular, it is locally Lipschitz.
\end{lemma}

\begin{example}\label{ex:euclidean q-gradient}
In the Euclidean setting, where $\MM = \RR^n$, we can regard $Q$ as a function~$Q \colon (\RR^n)^* \to \overline{\RR}$, or equivalently~$\phi \colon \RR^n \to \overline{\RR}$, as discussed in \cref{ex:euclidean as selfadjoint}.
Then the $Q$-gradient is simply
\begin{align}\label{eq:euclidean q-gradient}
    \nabla^Q f(x) = \diff Q(\diff f(x)) = \nabla \phi(\nabla f(x)),
\end{align}
where we interpret the differential $\diff f(x)$ as an element of~$T_x^*\MM \cong (\RR^n)^*$, and $\nabla f$ and~$\nabla\phi$ are the usual Euclidean gradients.
See also the discussion in \cref{sec:intro}.
\end{example}

\section{Hadamard mirror flow}\label{sec:continuous-time}
In this section, we study the descent flow associated to the $Q$-gradient, first defined by Hirai, who also proved that $Q(\diff f)$ is monotonically nonincreasing along the flow~\cite{Hirai2025}.
We recall his definition in \cref{sub:flow and euclidean} and give a novel interpretation as a \emph{Hadamard mirror flow}.
In \cref{sub:flow convergence}, we then establish a $\mathcal{O}(t^{-1})$ convergence result (\cref{thm:main-continuous}).
Throughout this section, $\MM$ denotes a Hadamard manifold, and $f$ and $Q$ are functions satisfying the hypotheses of \cref{lem:C1 Q-gradient}.

\subsection{Definition and Euclidean interpretation}\label{sub:flow and euclidean}
We first recall Hirai's $Q$-gradient flow in our setting:
it is the descent flow associated with the negative of the $Q$-gradient~\cite{Hirai2025}.
Below we show that in the Euclidean setting this flow has a mirror flow interpretation;
accordingly we also call it the \emph{Hadamard mirror flow}.

\begin{definition}[Hadamard mirror flow]\label{def:mirror flow}
Let $f$, $Q$ be functions satisfying the hypotheses of \cref{lem:C1 Q-gradient}.
Then the \emph{Hadamard mirror flow}, or $Q$-gradient flow, for problem~\eqref{eq:general_minimization} is defined as the unique maximal forward solution~$x_t$ to the first-order differential equation
\begin{align}\label{eqn:q-gradientflow}
    \dot x_t = -\nabla^Q f(x_t) \qquad (0 \leq t < T_{\max}) 
\end{align}
for a given initial point $x_0 \in \MM$, where $T_{\max} \in (0,\infty]$.
\end{definition}

By our assumption, the $Q$-gradient is a locally Lipschitz vector field (\cref{lem:C1 Q-gradient}).
Then the existence of a unique maximal solution with $T_{\max}>0$ follows from standard ODE theory.
Now Hirai's descent lemma for the $Q$-gradient flow follows immediately from \cref{lem:dFdf}:

\begin{lemma}[{cf. \cite[Lem.~3.9]{Hirai2025}}]\label{lem:descentHirai}
    The function $t \mapsto Q(\diff f(x_t))$ is nonincreasing.
\end{lemma}
\begin{proof}
It follows at once from \cref{lem:dFdf}, the definition of the Hadamard mirror flow, and the geodesic convexity of~$f$ that
\begin{equation*}
    \frac{\diff}{\diff t} Q(\diff f(x_t))
=   (\nabla^2 f)[\dot x_t, \nabla^Q f(x_t)]
=  -(\nabla^2 f)[\dot x_t, \dot x_t]
\leq 0.
\qedhere
\end{equation*}
\end{proof}

In the Euclidean case, the flow~\eqref{eqn:q-gradientflow} can be interpreted as a continuous-time dynamical system that served as the original motivation for mirror descent~\cite[\S{}3.1.4]{NY1983}.%
\footnote{In \cref{sec:discrete-time} below, we consider the discrete-time variant of the flow, which can be interpreted as a Hadamard generalization of mirror \emph{descent} (in dual coordinates).
Readers familiar with mirror descent may prefer to skip ahead and consult \cref{sub:descent and euclidean}.}
We briefly describe this connection using our notation.
Nemirovski and Yudin consider the problem of minimizing a convex function~$Q$ on a reflexive Banach space~$E$.
To this end, they choose a function $f \colon E^* \to \RR$ and consider the differential equation for~$x_t \in E^*$ given by
\[ \dot x_t = -\diff Q(\diff f(x_t)), \]
where we consider $\diff f(x_t) \in T_{x_t}^*(E^*) \cong E^{**} \cong E$.
If $E = (\RR^n)^*$, then $\MM \coloneqq E^* \cong \RR^n$; using \cref{eq:euclidean q-gradient} we see that the above is nothing but the mirror flow~\eqref{eqn:q-gradientflow} specialized to the Euclidean setting.
Their key observation is that $f^\xi \coloneqq f - \xi$ (for $\xi \in E \cong E^{**}$) can serve as a ``Lyapunov function'':
\begin{equation}\label{eq:nemi yudin}
  \frac {\diff}{\diff t} f^\xi(x_t)
= ( \diff f(x_t) - \xi ) [\dot x_t]
= \diff Q(\diff f(x_t))\mleft[ \xi - \diff f(x_t) \mright]
\leq Q(\xi) - Q(\diff f(x_t)),
\end{equation}
using the convexity of~$Q$.
Thus, $f^\xi(x_t)$ is decreasing as long as $Q(\diff f(x_t)) > Q(\xi)$.
Assuming the flow is defined at all times and $f^\xi$ is bounded below, it follows that $\liminf_{t\to\infty} Q(\diff f(x_t)) \leq Q(\xi)$.
In particular, if $\xi$ is a minimizer, then $\liminf_{t\to\infty} Q(\diff f(x_t))$ is the minimum.

\subsection{Convergence analysis}\label{sub:flow convergence}
We now analyze the convergence of~$Q(\diff f(x_t))$ under the Hadamard mirror flow.
The key idea will be to control the decrease of the objective using the geodesically convex ``Lyapunov function''
\begin{align}\label{eq:lyapunov}
    f^\xi \coloneqq f + b^\xi
\end{align}
and the quantity~$Q(\xi)$ defined in \cref{eq:Fxi}, where~$\xi \in C\MM^\infty$ and $b^\xi$ is a corresponding Busemann function.
This generalizes the approach by Nemirovski and Yudin since in the Euclidean setting, the Busemann function~$b^\xi$ simply becomes a linear function $b^\xi = -\xi$ (up to an arbitrary constant); here we use the same notation as in the previous subsection, $\MM = E^*$, and we identify the equivalence class~$\xi$ of rays in~$\MM$ with a vector $\xi \in E\cong E^{**}$.

We remark that the function $f^\xi$ is related to the asymptotic Legendre--Fenchel conjugate $f^*$ introduced in~\cite[(2.16)]{Hirai2024} by $\inf_{x \in \MM} f^\xi(x) = -f^*(\xi)$ (up to the choice of aditive constant in~$b^\xi$).

We first present a lemma that generalizes \cref{eq:nemi yudin} to the Hadamard setting.

\begin{lemma}[$f^\xi$ decreases by the gap of $Q$]\label{lem:continuous-smallf}
    For all $\xi \in C\MM^\infty$,
    \[
        \frac{\diff}{\diff t} f^\xi(x_t) \leq Q(\xi) - Q(\diff f(x_t)).
    \]
\end{lemma}
\begin{proof}
This follows by direct computation:
    \begin{align*}
        \frac{\diff}{\diff t}f^\xi(x_t)
        &= \left(\diff f(x_t) + \diff b^\xi(x_t)\right)\left[\dot x_t\right]\\
        &= -\diff (Q_{x_t})(\diff f(x_t))\left[\diff f(x_t) + \diff b^\xi(x_t)\right]\\
        &\leq Q_{x_t}(-\diff b^\xi(x_t)) - Q_{x_t}(\diff f(x_t)) \\
        &= Q(\xi) - Q(\diff f(x_t)),
    \end{align*}
    where the second equality follows from the definition of the Hadamard mirror flow~\eqref{eqn:q-gradientflow}, the inequality from the Euclidean convexity of~$Q_{x_t}$, and the last equality follows from the definition of $Q(\xi)$ in \cref{eq:Fxi}.
\end{proof}

We are now ready to prove a convergence result.
To this end, we assume that the Hadamard mirror flow~$x_t$ is defined for all times~$t\geq0$; this holds under natural assumptions, see \cite[Lem.~3.11]{Hirai2025}.
Then, by combining \cref{lem:descentHirai,lem:continuous-smallf}, we have the following convergence bound:

\begin{theorem}[Convergence of Hadamard mirror flow]
\label{thm:main-continuous}
    Assume that the Hadamard mirror flow~\eqref{eqn:q-gradientflow} is defined for all~$t\geq0$.
    Then, for all $\xi \in C\MM^\infty$ and $t>0$,
    \begin{equation}\label{eq:main continuous first}
        Q(\diff f(x_t)) - Q(\xi) \le \frac{f^\xi(x_0) - \inf_{x \in \MM} f^\xi(x)}t.
    \end{equation}
    In particular, it holds that
    \begin{equation}\label{eq:main continuous second}
        \lim_{t \to \infty}Q(\diff f(x_t)) = \inf_{x \in \MM}Q(\diff f(x))
    \end{equation}
    (both sides can be $-\infty$).
\end{theorem}
\begin{proof}
    The inequality~\eqref{eq:main continuous first} is established by the following calculation:
    \[
        f^\xi(x_0) - \inf_{x \in \MM} f^\xi(x) \ge f^\xi(x_0) - f^\xi(x_t) \ge \int_0^t\left(Q(\diff f(x_s)) - Q(\xi)\right)\diff s \ge t\left(Q(\diff f(x_t)) - Q(\xi)\right),
    \]
    where the second inequality follows from \cref{lem:continuous-smallf}, and the last follows from \cref{lem:descentHirai}.

    To prove \cref{eq:main continuous second}, let $x \in \MM$ be arbitrary (and not to be confused with the curve~$x_t$) and consider the class~$\xi \in C\MM^\infty$ of the geodesic ray~$t \mapsto \exp_x(t \nabla f(x))$, where $\nabla f(x) \coloneqq \flat^{-1} \diff f(x)$ and $\flat \colon T\MM \to T^*\MM$ denotes the natural isomorphism induced by the Riemannian metric.
    By \cref{eq:diff busemann}, we have $\diff b^\xi(x) = -\diff f(x)$.
    In other words $\diff f^\xi(x) = 0$.
    Hence, by geodesic convexity, we have $\inf_{y \in \MM}f^\xi(y) = f^\xi(x) > -\infty$.
    On the other hand, we have $Q(\xi) = Q(\diff f(x))$ by \cref{lem:Q-busemann}.
    Together, and using \cref{eq:main continuous first}, we find that
    \begin{align*}
        Q(\diff f(x_t)) - Q(\diff f(x))
    =   Q(\diff f(x_t)) - Q(\xi)
    \leq \frac{f^\xi(x_0) - \inf_{y \in \MM}f^\xi(y)}{t}
    \underset{t \to \infty}{\longrightarrow} 0;
    \end{align*}
    the numerator is finite by the observation above.
    Since $x \in \MM$ was arbitrary, this implies that
    \[ \liminf_{t\to\infty} Q(\diff f(x_t)) \leq \inf_{x \in \MM} Q(\diff f(x)). \]
    The limit inferior exists (but may be $-\infty$) because $Q(\diff f(x_t))$ is nonincreasing by \cref{lem:descentHirai}.
    Along with the pointwise lower bound $Q(\diff f(x_t)) \geq \inf_{x \in \MM} Q(\diff f(x))$, \cref{eq:main continuous second} follows.
\end{proof}


\section{Hadamard mirror descent}\label{sec:discrete-time}
In this section, we study \emph{Hadamard mirror descent}, the natural discrete-time variant of the Hadamard mirror flow.
In \cref{sub:descent and euclidean} we define the update formally and show that it can be interpreted as a Hadamard generalization of the well-known mirror-descent method (expressed in dual variables), justifying its name.
In \cref{sub:descent convergence}, we state a step-size condition that implies a $\mathcal{O}(1/k)$ convergence result (\cref{thm:main-discrete}).
This condition is ensured by a natural \emph{relative smoothness} assumption, which in particular holds when $Q$ is smooth and $f$ is dual-smooth, as we discuss in \cref{sub:rel-smooth}.
These assumptions are motivated by the Euclidean theory and easy to verify in practice.
Throughout, $\MM$ denotes a Hadamard manifold, and $f$, $Q$ are functions satisfying the hypotheses of \cref{def:grad}.

\subsection{Definition and Euclidean interpretation}\label{sub:descent and euclidean}

We first define the discrete-time variant of \cref{def:mirror flow}.

\begin{definition}[Hadamard mirror descent]
Let $f$, $Q$ be functions satisfying the hypotheses of \cref{def:grad}.
Then the \emph{Hadamard mirror descent}, or $Q$-gradient descent, for problem~\eqref{eq:general_minimization} with step size~$\eta>0$ is defined as the iteration
\begin{align}\label{eq:descent}
    x_{k+1} \coloneqq \exp_{x_k}\mleft( -\eta \nabla^Q f(x_k) \mright) \qquad (k \in \NN)
\end{align}
for a given initial point~$x_0 \in \MM$.
\end{definition}

We remark that Hirai (without assuming differentiability of~$Q$) defined a $Q$-\emph{sub}gradient method and asked for natural conditions under which convergence can be guaranteed~\cite[(3.16), Quest.~3.21]{Hirai2025}.
For the special case of the $Q$-gradient method, we give such conditions in \cref{sub:descent convergence}.

Analogously to the discussion in \cref{sub:flow and euclidean}, in the Euclidean case, the update~\eqref{eq:descent} can be understood as a (dual) formulation of \emph{mirror descent}, justifying our terminology.
We recall the standard setting of mirror descent as follows (see, e.g., \cite[\S{}9]{Beck2017}).
Let~$E$ be a Euclidean space, and let~$h \colon E \to \overline\RR$ be a strongly convex and essentially smooth function, often called the \emph{distance-generating function} in this context.
The associated \emph{Bregman divergence} is defined by
\begin{equation}\label{eq:euclidean bregman}
    D_h(w\Vert z) \coloneqq h(w) - \left( h(z) + \diff h(z)[w-z] \right) \qquad (w \in \dom h, z \in \intr\dom h).
\end{equation}
Then the mirror descent procedure for the objective~$Q \colon E \to \overline{\RR}$ and step size~$\eta$ reads as follows:
\begin{equation*}
    z_{k+1} \coloneqq \argmin_{z \in E} \braces*{ Q(z_k) + \diff Q(z_k)[z - z_k] + \frac1{\eta} D_h(z\Vert z_k) }.
\end{equation*}
The first-order optimality condition for this minimization is $\diff Q(z_k) + \frac1{\eta} \left( \diff h(z_{k+1}) - \diff h(z_k) \right) = 0$, that is, $\diff h(z_{k+1}) = \diff h(z_k) - \eta \diff Q(z_k)$.
If we instead use the ``dual'' variables $x_k \coloneqq \diff h(z_k) \in E^*$, we arrive at the following iteration in ``gradient space'' (or rather differential space):
\begin{align}\label{eq:mirror in grad space}
    x_{k+1} = x_k - \eta \diff Q\mleft( \diff(h^*)(x_k) \mright),
\end{align}
where~$h^* \colon E^* \to \overline{\RR}$ denotes the Legendre--Fenchel conjugate of~$h$, which satisfies $\diff(h^*) = (\diff h)^{-1}$~(e.g., \cite[Cor.~E.1.4.4]{HL2001}).
Let $E=(\RR^n)^*$ and set $\MM \coloneqq E^* \cong \RR^n$.
Using \cref{eq:euclidean q-gradient}, we see that the iteration~\eqref{eq:mirror in grad space} coincides with the Hadamard mirror descent~\eqref{eq:descent} specialized to the Euclidean setting.

From \cite[Thm.~1]{BBT2017}, \cite[Thm.~3.1]{LFN2018} it is known that the step size in mirror descent can be chosen as~$\eta \le 1/L$, where $L$ is the so-called \emph{relative smoothness parameter} of $Q$ with respect to $h$:
\begin{equation}\label{eq:relsmooth-Bregman primal}
    D_{Q}( w \Vert z ) \leq L \cdot D_{h}(w \Vert z).
\end{equation}
Note that this equation simply states that $Lh - Q$ is convex.
In dual coordinates, we have the following expression for the Bregman divergence:
\begin{equation*}
    D_h(\diff(h^*)(y)\Vert \diff(h^*)(x)) = D_{h^*}(x\Vert y).
\end{equation*}
Hence, identifying~$f = h^*$, we can write \cref{eq:relsmooth-Bregman primal} as
\begin{align}\label{eq:relsmooth-Bregman dual}
    D_{Q}(\diff f(y)\Vert \diff f(x)) \leq L \cdot D_f(x\Vert y).
\end{align}
In \cref{sub:rel-smooth}, we will see that this reformulation generalizes appropriately to the Hadamard setting and implies convergence for Hadamard mirror descent.

\subsection{Convergence analysis}\label{sub:descent convergence}
In this section, we present a general analysis of the convergence of Hadamard mirror descent under the following hypothesis.

\begin{definition}[Step-size condition]
\label{def:step-size-condition}
We say that the sequence~\eqref{eq:descent} satisfies the step-size condition if
\begin{equation}\label{eq:step-size-condition}
    \eta\left(Q(\diff f(x_{k+1})) - Q(\diff f(x_k))\right) \le f(x_k) - f(x_{k+1}) + \diff f(x_k)\mleft[\exp_{x_k}^{-1}x_{k+1}\mright]
\qquad (\forall k \in \NN).
\end{equation}
\end{definition}

\noindent
In \cref{sub:rel-smooth} we will discuss a natural and easy-to-verify condition that implies \cref{eq:step-size-condition}.

The step-size condition allows us to proceed analogously to \cref{sub:flow convergence}.
First, we note that the condition immediately implies the following analog of \cref{lem:descentHirai}.

\begin{lemma}[$Q$ is nonincreasing]\label{lem:rel-largeF}
    Under the step-size condition~\eqref{eq:step-size-condition}, for all $k\in\NN$,
    \[
        Q(\diff f(x_{k+1})) \leq Q(\diff f(x_k)).
    \]
\end{lemma}
\begin{proof}
This follows from the step-size condition~\eqref{eq:step-size-condition}, whose right-hand side is nonpositive by the geodesic convexity of~$f$ (along the geodesic through $x_k$ and $x_{k+1}$).
\end{proof}

Next, we prove the analog of \cref{lem:continuous-smallf}.

\begin{lemma}[$f^\xi$ decreases by the gap of $Q$]\label{lem:rel-smallf}
    Under the step-size condition~\eqref{eq:step-size-condition}, for all $\xi \in C\MM^\infty$ and~$k\in\NN$,
    \[
        \frac1\eta \left( f^\xi(x_{k+1}) - f^\xi(x_k) \right) \leq Q(\xi) - Q(\diff f(x_{k+1})),
    \]
    where we recall that $f^\xi \coloneqq f + b^\xi$ as defined in \cref{eq:lyapunov}.
\end{lemma}
\begin{proof}
We compute:
\begin{align*}
&\qquad \frac1\eta \left( f^\xi(x_{k+1}) - f^\xi(x_k) \right) \\
&= \frac1\eta \left( f(x_{k+1}) - f(x_k) + b^\xi(x_{k+1}) - b^\xi(x_k) \right) \\
&\leq Q(\diff f(x_k)) - Q(\diff f(x_{k+1})) + \frac1\eta \diff f(x_k)\mleft[\exp_{x_k}^{-1}x_{k+1}\mright]
    - \frac1\eta \diff b^\xi(x_{k+1})\mleft[\exp_{x_{k+1}}^{-1}x_k\mright] \\
&= Q(\diff f(x_k)) - Q(\diff f(x_{k+1})) + \frac1\eta \diff f(x_k)\mleft[\exp_{x_k}^{-1}x_{k+1}\mright]
    + \frac1\eta \diff b^\xi(x_{k+1})\mleft[\tau_{x_k\to x_{k+1}} \exp_{x_k}^{-1}x_{k+1}\mright] \\
&= Q(\diff f(x_k)) - Q(\diff f(x_{k+1})) - \diff(Q_{x_k})(\diff f(x_k))\mleft[ \diff f(x_k) + \tau_{x_k\to x_{k+1}}^* \diff b^\xi(x_{k+1}) \mright] \\
&= Q_{x_k}(\diff f(x_k)) + \diff(Q_{x_k})(\diff f(x_k))\mleft[-\tau^*_{x_k\to x_{k+1}}\diff b^\xi(x_{k+1})-\diff f(x_k)\mright] - Q(\diff f(x_{k+1})) \\
&\leq Q_{x_k}\mleft(- \tau^*_{x_k\to x_{k+1}}\diff b^\xi(x_{k+1})\mright) - Q(\diff f(x_{k+1})) \\
&= Q(\xi) - Q(\diff f(x_{k+1})),
\end{align*}
where the first equality holds by substituting the definition of~$f^\xi$;
the subsequent inequality follows from the step-size condition~\eqref{eq:step-size-condition} and the geodesic convexity of~$b^\xi$;
the second equality holds because $\exp_w^{-1}(z) = -\tau_{z\to w}(\exp_z^{-1}(w))$ for all~$z,w \in \MM$;
the third equality follows from the definition of Hadamard mirror descent~\eqref{eq:descent}, i.e., $\exp_{x_k}^{-1}x_{k+1} = -\eta \nabla^Q f(x_k)$, and of the $Q$-gradient~\eqref{eq:grad};
next we rearrange and use the notation for the restrictions $Q_{x_k} \coloneqq Q|_{T_{x_k}^*\MM}$;
the last inequality holds by the Euclidean convexity of~$Q_{x_k}$;
and the last equality uses the parallel transport invariance of~$Q$ and the definition of~$Q(\xi)$ (see \cref{eq:Fxi}).
\end{proof}

We thus obtain the following convergence theorem analogously to \cref{thm:main-continuous}.

\begin{theorem}[Convergence of Hadamard mirror descent under the step-size condition]
\label{thm:main-discrete}
    Under the step-size condition~\eqref{eq:step-size-condition}, for all $\xi \in C\MM^\infty$ and $k\geq1$,
    \begin{equation}\label{eq:main discrete first}
        Q(\diff f(x_k)) - Q(\xi) \le \frac{f^\xi(x_0) - \inf_{x \in \MM} f^\xi(x)}{\eta k}.
    \end{equation}
    In particular, it holds that
    \begin{equation}\label{eq:main discrete second}
        \lim_{k \to \infty}Q(\diff f(x_k)) = \inf_{x \in \MM}Q(\diff f(x))
    \end{equation}
    (both sides can be $-\infty$).
\end{theorem}
\begin{proof}
    The inequality~\eqref{eq:main discrete first} is proved by the following calculation:
    \[
        f^\xi(x_0) - \inf_{x \in \MM} f^\xi(x) \ge f^\xi(x_0) - f^\xi(x_k) \ge \eta\sum_{i = 1}^k\left(Q(\diff f(x_i)) - Q(\xi)\right) \ge \eta k\left(Q(\diff f(x_k)) - Q(\xi)\right),
    \]
    where the second inequality follows from \cref{lem:rel-smallf} and a telescoping sum, and the last follows from \cref{lem:rel-largeF}.
    Then \cref{eq:main discrete second} is shown in the same manner as in \cref{thm:main-continuous}.
\end{proof}

\begin{remark}\label{rem:hypothesis}
    To obtain \cref{thm:main-discrete}, the hypothesis $\diff f(x) \in \intr\dom Q_x$ in \cref{def:grad} is not required for every $x \in \MM$, but only for the iterates $x = x_k$ of the Hadamard mirror descent, so that the iteration~\eqref{eq:descent} is well-defined.
    This fact will be used in \cref{sub:tensor_power_action}.
\end{remark}

\subsection{Relative smoothness and dual smoothness}\label{sub:rel-smooth}
In this section, we identify a natural \emph{relative smoothness} parameter that implies the step-size condition~\eqref{eq:step-size-condition} when it is finite.
The definition is a generalization of the eponymous parameter in the Euclidean setting.
We also show that relative smoothness is implied by smoothness of the objective~$Q$ and \emph{dual-smoothness} of the function~$f$; we define the latter notion below.
We first define a Riemannian generalization of the Bregman divergence and then define relative smoothness.

\begin{definition}[Riemannian Bregman divergence]
    Let $f \colon \MM \to \RR$ be a geodesically convex differentiable function.
    Then its \emph{Riemannian Bregman divergence} is defined as
    \begin{equation}\label{eq:riemannian bregman}
        D_f^{\MM}(x\Vert y) \coloneqq f(x) - \left(f(y) + \diff f(y)\mleft[\exp_{y}^{-1}x\mright]\right) \geq 0 \quad (x, y \in \MM).
    \end{equation}
\end{definition}

When $\MM$ is a Euclidean space, we recover the usual formulation of the Bregman divergence;
in this case we omit the superscript, as in~\eqref{eq:euclidean bregman}.

\begin{definition}[Relative smoothness]
    Let $f$ and $Q$ satisfy the hypotheses of \cref{def:grad}.
    Then the pair~$(f, Q)$ is said to be \emph{$L$-relatively smooth} if for all~$x, y \in \MM$,
    \begin{equation}\label{eq:relsmooth}
        D_{Q_x}(\tau_{x \to y}^*\diff f(y)\Vert\diff f(x)) \leq L\cdot D_f^{\MM}(x\Vert y),
    \end{equation}
    where the left-hand side uses the Euclidean Bregman divergence (defined as in \eqref{eq:euclidean bregman}) and the right-hand side the Riemannian Bregman divergence~\eqref{eq:riemannian bregman}.
    The number~$L>0$ is called the relative smoothness parameter.
\end{definition}

For Euclidean space $\MM=\RR^n$, the relative smoothness condition~\eqref{eq:relsmooth} reduces to~\eqref{eq:relsmooth-Bregman dual}, which in turn is equivalent to the standard definition~\eqref{eq:relsmooth-Bregman primal} of relative smoothness of~$Q$ with respect to the Legendre-Fenchel conjugate~$h=f^\ast$ of~$f$ (under the assumptions discussed therein).

The next result states that, in the general Hadamard setting, relative smoothness implies the step-size condition when applied to the Hadamard mirror descent sequence for any step size~$\eta\leq1/L$.
In particular, our convergence result applies.

\begin{proposition}[Relative smoothness implies step-size condition]\label{lem:relsmooth}~
\begin{enumerate}
    \item[(a)] The step-size condition~\eqref{eq:step-size-condition} is equivalent to~\eqref{eq:relsmooth} with $x = x_k$, $y = x_{k+1}$, and $L = 1/\eta > 0$.
    \item[(b)] If the pair $(f, Q)$ is $L$-relatively smooth, the step-size condition~\eqref{eq:step-size-condition} holds for all $\eta\leq 1/L$ and hence \cref{thm:main-discrete} applies.
\end{enumerate}
\end{proposition}
\begin{proof}
(a)~We first compute the two Bregman divergences in~\eqref{eq:relsmooth} for~$x = x_k$ and~$y = x_{k+1}$:
\begin{align*}
    D^\MM_{f}(x_k \Vert x_{k+1})
    &= f(x_k) - f(x_{k+1}) - \diff f(x_{k+1})\mleft[\exp_{x_{k+1}}^{-1} x_k\mright]\\
    &= f(x_k) - f(x_{k+1}) + \mleft(\tau^*_{x_k\to x_{k+1}}\diff f(x_{k+1})\mright)\mleft[\exp_{x_k}^{-1} x_{k+1}\mright],
\end{align*}
because $\exp^{-1}_w(z) = -\tau_{z\to w}(\exp_z^{-1}(w))$ for all $z,w\in\MM$, while
\begin{align*}
&\qquad D_{Q_{x_k}}(\tau^*_{x_k \to x_{k+1}}\diff f(x_{k+1}) \Vert \diff f(x_k)) \\
&= Q_{x_k}(\tau^*_{x_k \to x_{k+1}}\diff f(x_{k+1})) - \left( Q_{x_k}(\diff f(x_k)) + \diff (Q_{x_k})(\diff f(x_k))\mleft[ \tau^*_{x_k \to x_{k+1}}\diff f(x_{k+1}) - \diff f(x_k) \mright] \right) \\
&= Q(\diff f(x_{k+1})) - Q(\diff f(x_k)) + \frac1\eta \left( \tau^*_{x_k \to x_{k+1}}\diff f(x_{k+1}) - \diff f(x_k) \right)\mleft[ \exp_{x_k}^{-1} x_{k+1} \mright],
\end{align*}
where the second equality uses the parallel transport invariance of~$Q$ as well as the definition of the Hadamard mirror descent~\eqref{eq:descent} and the $Q$-gradient~\eqref{eq:grad}: we have $\exp_{x_k}^{-1} x_{k+1} = -\eta \nabla^Q f(x_k) = -\eta \diff (Q_{x_k})(\diff f(x_k))$.
If we substitute the above into~\eqref{eq:relsmooth} and use $L=\frac1\eta$, the term proportional to~$(\tau^*_{x_k\to x_{k+1}}\diff f(x_{k+1}))[\exp_{x_k}^{-1}(x_{k+1})]$ cancels out and we obtain the step-size condition~\eqref{eq:step-size-condition}.

(b)~This follows from part~(a) and the observation that $L$-relative smoothness implies $L'$-relative smoothness for any~$L' \ge L$.
\end{proof}

In the Euclidean case, it is well-known that if $f$ is $L_f$-smooth (equivalently, the distance-ge\-ne\-ra\-ting function~$h=f^*$ is $L_f^{-1}$-strongly convex) and separately $Q$ is~$L_Q$-smooth, then the pair~$(f,Q)$ is relatively smooth with parameter~$L_f L_Q$.
This provides an easy-to-check criterion.
In the Hadamard setting, we replace the smoothness of~$f$ by the following condition.

\begin{definition}[Dual smoothness]\label{def:dual smooth}
    A geodesically convex differentiable function~$f \colon \MM \to \RR$ is called \emph{$L_f$-dual-smooth} if for all $x,y \in \MM$,
    \begin{equation}\label{eqn:dual-smooth}
        \frac12 \norm{ \tau^{\ast}_{x \to y}\diff f(y) - \diff f(x) }^2 \leq L_f \cdot D_f^\MM(x\Vert y),
    \end{equation}
    where $\norm{\cdot}$ denotes the norm of cotangent vectors induced by the Riemannian metric.
    The quantity $L_f>0$ is called the dual smoothness parameter of~$f$.
\end{definition}

\noindent
This definition should be contrasted with \emph{$L_f$-smoothness}: it can be defined as $D_f^{\MM}(y\Vert x) \leq \frac {L_f}2 \dist(x,y)^2$, or equivalently
\begin{align}\label{eq:smoother}
\norm{ \tau^{\ast}_{x \to y}\diff f(y) - \diff f(x) }
\leq L_f \cdot \dist(x,y)
    \quad (x, y \in \MM),
\end{align}
or simply as Euclidean $L_f$-smoothness along any unit-speed geodesic~$\gamma \colon \RR \to \MM$.
For twice differentiable~$f$, the latter means that
\begin{align}\label{eq:smooth}
    \frac{\diff^2}{\diff t^2} f(\gamma(t)) \leq L_f.
\end{align}
In the Euclidean case ($\MM = \RR^n$), smoothness and dual-smoothness coincide:
a function is $L_f$-smooth if and only if it is $L_f$-dual-smooth~\cite[Thm.~5.8]{Beck2017}.
In the Hadamard setting, these parameters are in general distinct:
while any $L_f$-dual-smooth function is $L_f$-smooth,%
\footnote{
To see this, suppose that $f$ is $L_f$-dual-smooth, and apply \cref{eqn:dual-smooth} twice, once with the roles of~$x$ and~$y$ reversed:
we obtain
$\norm{ \tau^{\ast}_{x \to y}\diff f(y) - \diff f(x) }^2
\leq L_f \cdot ( D_f^\MM(x\Vert y) + D_f^\MM(y\Vert x) )
= L_f \cdot (\tau^{\ast}_{x \to y}\diff f(y) - \diff f(x))[\exp_{x}^{-1}y]
\leq L_f \cdot \norm{ \tau^{\ast}_{x \to y}\diff f(y) - \diff f(x) } \cdot\dist(x,y)$.
Thus we have established \cref{eq:smoother}.}
the converse is not true in general (with the same parameter); we give an explicit counterexample in \cref{ex:countersmooth}.
In particular, unlike $L$-smoothness, $L$-dual-smoothness of a function does not in general follow from $L$-dual-smoothness of its restrictions to geodesics (since the latter is equivalent to $L$-smoothness).

Using this notion, we can establish a Hadamard generalization of the Euclidean criterion.
It will be applied to the problem of convex optimization on moment polytopes (\cref{sec:momentpolytope,sec:applications}).

\begin{proposition}
\label{prop:dual-to-relative}
    If $f$ is $L_f$-dual-smooth and $Q$ is $L_Q$-smooth, then $(f, Q)$ is $L_f L_Q$-relatively smooth.
    In particular, for any $\eta\leq 1/(L_f L_Q)$, \cref{thm:main-discrete} applies.
\end{proposition}
\begin{proof}
    This is now immediate:
    \begin{equation*}
       D_{Q_x}(\tau^{\ast}_{x\rightarrow y}\diff f(y)\Vert \diff f(x) ) \leq \frac{L_Q}2\|\tau^*_{x\to y}\diff f(y) - \diff f(x)\|^2 \leq L_fL_Q \cdot D_f^{\MM}(x\Vert y),
    \end{equation*}
    where the first inequality follows from the $L_Q$-smoothness of the Euclidean convex function~$Q_x$, and the second inequality follows from the $L_f$-dual-smoothness of~$f$.
\end{proof}

\section{Convex optimization on moment polytopes}\label{sec:momentpolytope}
In this section, we explore a central application of Hadamard mirror descent---convex optimization on moment polytopes.
In \cref{sub:momentpolytope-general}, we recall the definition of moment polytopes and their characterization in terms of Kempf--Ness functions.
In \cref{sub:convex_optimization_on_moment_polytopes}, after recalling Hirai's approach via holonomy-invariant convex objectives, we give general $\mathcal{O}(k^{-1})$ convergence results for convex optimization on moment polytopes using Hadamard mirror descent; we also state analogous results for the mirror flow.
In \cref{sub:dualsmooth}, we show that Kempf--Ness functions are dual-smooth, with a parameter that is naturally computed from the group action.
We also comment on the relation to smoothness, sharpening bounds from the literature.
The remaining subsections are devoted to specific actions.
In \cref{sub:entanglementpolytope}, we consider the tensor action and the associated \emph{entanglement polytopes}, which underlie many applications.
In \cref{sub:tensor_power_action}, we discuss the tensor power action and explain how this case can be reduced to the tensor action.
In \cref{sub:GLaction}, we consider arbitrary homogeneous polynomial actions of~$\GL(n)$.
In all these settings, the convergence of Hadamard mirror descent is polynomial in the input bit-length.

\subsection{Moment polytopes and Kempf--Ness functions}\label{sub:momentpolytope-general}
Let~$G \subseteq \GL(n)$ be a connected self-adjoint Zariski-closed subgroup, as in \cref{sub:selfconjugate,sub:objective selfadjoint}.
Let~$P \coloneqq G \cap \PD(n)$ denote the associated Hadamard symmetric space, $K \coloneqq G \cap \U(n)$ the maximal compact subgroup, $\bbH \coloneqq i\Lie K = \Lie G \cap \Herm(n)$ the Hermitian part of the Lie algebra (which is naturally identified with the tangent space at the identity $\bbH \cong T_IP$), and let $C \subseteq \bbH$ denote a Weyl chamber; recall the map $s \colon \bbH \to C$ such that $\{u X u^\dagger : u \in K\} \cap C = \{s(X)\}$ for all~$X \in \bbH$.

Let $\pi \colon G \to \GL(\VV)$ be a rational representation of~$G$ on a finite-dimensional complex vector space~$\VV$.
We may equip~$\VV$ with a $K$-invariant inner product $\langle\cdot,\cdot\rangle$, so $\pi(K) \subseteq \U(\VV)$;
we denote the associated norm by $\norm\cdot \coloneqq \sqrt{\langle\cdot,\cdot\rangle}$.
We write $g \cdot v \coloneqq \pi(g)v$ for~$g \in G$ and~$v \in \VV$, so $G \cdot v$ denotes the orbit of~$v$.
Let $\Pi \colon \Lie G \to \Lin(\VV)$ denote the corresponding Lie algebra representation:
$\Pi(X)\coloneqq\partial_{t=0} \pi(e^{tX})$ for $X\in\Lie G$.
We observe that $\pi(g^\dagger) = \pi(g)^\dagger$ for $g \in G$ and $\Pi(\bbH) \subseteq \Herm(\VV)$ because $\pi(K) \subseteq \U(\VV)$.
We can then define the moment map and the moment polytope.

\begin{definition}[Moment map]
For a group and representation as above, the \emph{moment map} is defined as
\begin{equation}\label{eq:moment map}
    \mu \colon \VV\setminus\{0\}\rightarrow \bbH, \quad \tr\mleft[ \mu(v) X \mright] = \frac{\langle v, \Pi(X) v\rangle}{\langle v,v\rangle}
    \quad \forall X\in \bbH.
\end{equation}
We remark that there are several equivalent variants of this definition in the literature.
For example, the domain can also be taken as the projective space~$\bbP(\VV)$, since the moment map depends on its argument~$v$ only through the corresponding ray~$[v]$.
It is also natural to use the dual space $\bbH^*$ as the codomain.
\end{definition}

The moment map is $K$-equivariant in the following way: for all $v \in \VV\setminus\{0\}$ and $u \in K$, we~have
\begin{align}\label{eq:mu equivariance}
    \mu(u \cdot v) = u \mu(v) u^\dagger.
\end{align}
This holds because $K$ acts unitarily and we have $\Pi(u^\dagger X u) = \pi(u)^\dagger \Pi(X) \pi(u)$ for all~$X \in \bbH$.

In general, the image of the moment map need not be convex (if the group is non-commutative).
However, one obtains a polytope by intersecting the image with the Weyl chamber~$C$; equivalently, we may apply the map~$s \colon \bbH \to C$.
More precisely, one can obtain a convex polytope for any $G$-invariant irreducible projective subvariety.
For our purposes it suffices to consider orbit closures.

\begin{definition}[Moment polytope]
For a group and representation as above, we define the \emph{moment polytope}~$\Delta(v)$ of a vector~$v \in \VV\setminus\{0\}$,
\begin{equation}\label{eq:momentpolytope}
\begin{aligned}
    \Delta(v)
&\coloneqq \left\{ \mu(w) \mid [w] \in \overline{G \cdot [v]} \right\} \cap C
= \left\{ s(\mu(w)) \mid [w] \in \overline{G \cdot [v]} \right\} \\
&= \overline{ \left\{ s(\mu(g \cdot v)) \mid g \in G \right\} }
= \overline{ \left\{ s(\mu(x \cdot v)) \mid x \in P \right\} } \subseteq C.
\end{aligned}
\end{equation}
This is also called the moment polytope of the orbit closure in projective space.
As above, we write $[w] \in \bbP(\VV)$ for the point in projective space corresponding to a nonzero vector~$w\in\VV$.
The first two characterizations are equivalent by equivariance~\eqref{eq:mu equivariance}, the third follows from the second by continuity and the compactness of orbit closures in projective space, and the last follows using the polar decomposition, equivariance~\eqref{eq:mu equivariance}, and the $K$-invariance of~$s$:
it holds that $s(\mu(g \cdot v)) = s(u \mu(x \cdot v) u^\dagger) = s(\mu(x \cdot v))$ where~$g = ux$ is the polar decomposition ($u \in K, x \in P$).
\end{definition}

Remarkably, $\Delta(v)$ is indeed a convex polytope:

\begin{theorem}[\cite{NessMumford1984,Kirwan1984}]
    For any $v\in\VV\setminus\{0\}$, $\Delta(v)$ is a rational convex polytope.
\end{theorem}

The moment map can be interpreted as the gradient of a geodesically convex function, called the Kempf--Ness function.
This is an important insight that is well-known in the literature~\cite{Kempf-78,NessMumford1984,Kirwan1984}.

\begin{definition}[Kempf--Ness function]
The Kempf--Ness function of a vector $v \in \VV\setminus\{0\}$ is defined as
\begin{align}\label{eq:kn}
    f_v \colon P \to \RR, \quad f_v(x) \coloneqq \log\langle v, x\cdot v\rangle.
\end{align}
\end{definition}

The following properties are well-known (see, e.g., \cite{BFGOWW2019,HiraiSakabe2024}).

\begin{lemma}[Gradient vs.\ moment map]\label{lem:KNderivative}
    Let $v \in \VV\setminus\{0\}$. Then:
    \begin{enumerate}
        \item[(a)] The Kempf--Ness function $f_v$ is geodesically convex.
        \item[(b)] The differential of~$f_v$ is directly related to the moment map:
        for every $x\in P$, $\tau^*_{I \to x}\diff f_v(x)$ is identified with~$\mu(x^{1/2}\cdot v)$ by the Hilbert--Schmidt inner product on~$T_I P \cong \bbH$ (the Riemannian metric at~$I \in P$).
        In other words, for every $X \in \bbH$ we have
        \[
            \diff f_v(x)[\tau_{I \to x} X]
        = \tau^*_{I \to x}\diff f_v(x)[X]
        = \tr[ \mu(x^{1/2}\cdot v) X]
        = \frac {\braket{w, \Pi(X) w}} {\braket{w, w}},
        \]
        where $w\coloneqq x^{1/2}\cdot v$.
    \end{enumerate}
\end{lemma}

\subsection{Convex optimization on moment polytopes}\label{sub:convex_optimization_on_moment_polytopes}
We now turn to convex optimization on the moment polytope~$\Delta(v)$ of a vector~$v\in\VV\setminus\{0\}$.
Let~$\phi \colon \AA \to \overline{\RR}$ be a function that satisfies \ref{it:F1}--\ref{it:F3}.
We assume that $s(\mu(g\cdot v)) \in \intr\dom\phi$ for all~$g \in G$.
We consider the minimization of $\phi$ on the moment polytope:
\begin{equation}\label{eq:min_over_momentpolytope}
    \minproblem{\phi(p)}{p \in \Delta(v)}.
\end{equation}
Note that there always exists a minimizer since $\phi$ is l.s.c.\ and $\Delta(v)$ is compact.

We may extend~$\phi$ to an objective~$\Phi \coloneqq \phi \circ s \colon \bbH \to \overline{\RR}$.
Then $\Phi$ satisfies \ref{it:Phi1}--\ref{it:Phi3} by \cref{prop:FandQ}, as well as $\mu(g\cdot v) \in \intr\dom\Phi$ for all $g \in G$.
This extension is natural: $\Phi$ is a convex function on the entire moment map image, yet its $K$-invariance implies that it only depends on the moment polytope.
Importantly, as pointed out by Hirai, the minimization problem~\eqref{eq:min_over_momentpolytope} can be reduced to a problem of the general form~\eqref{eq:general_minimization}~\cite{Hirai2025}.%
\footnote{Hirai stated this reduction in the case where $G$ is a product of $\GL(n)$'s, but his discussion readily generalizes to any connected self-adjoint subgroup $G$.}
Let~$Q \colon T^*P \to \overline{\RR}$ denote the holonomy-invariant convex objective corresponding to~$\phi$ and~$\Phi$.
Then we have, for any $x \in P$,
\begin{equation}\label{eq:Q_momentmap}
    Q(\diff f_v(x))
=   Q_I(\tau_{I\to x}^* \diff f_v(x))
=   \Phi(\mu(x^{1/2} \cdot v))
=   \phi(s(\mu(x^{1/2} \cdot v))),
\end{equation}
where the first equality holds by the invariance of~$Q$ under parallel transport,
the second equality follows from~\cref{lem:KNderivative} and the fact that~$Q_I$ is identified with~$\Phi$ using the Riemannian metric at~$I \in P$,
and the last equality follows from the definition of~$\Phi$ in terms of~$\phi$.
By taking the infimum over~$x \in P$ and using the last expression in \cref{eq:momentpolytope}, we obtain the following fundamental relation:

\begin{proposition}
Let $\phi\colon \AA \to \overline\RR$ be a function that satisfies~\ref{it:F1}--\ref{it:F3}, with associated holonomy-invariant convex objective~$Q\colon T^*P\to\overline{\RR}$.
Let $v \in\VV\setminus\{0\}$ be such that~$s(\mu(g\cdot v)) \in \intr\dom\phi$ for all~$g \in G$.
Then:
\[
    \inf_{x \in P}Q(\diff f_v(x)) = \min_{p \in \Delta(v)}\phi(p).
\]
\end{proposition}
\begin{proof}
The inequality $\inf_{x \in P}Q(\diff f_v(x)) \ge \min_{p \in \Delta(v)}\phi(p)$ is direct from $\{s(\mu(x\cdot v))\mid x \in P\} \subseteq \Delta(v)$ and \cref{eq:Q_momentmap}.
For the opposite inequality, we use the fact that $\Delta(v)$ is a polytope and the assumption that $\{s(\mu(x\cdot v))\mid x \in P\} \subseteq \intr\dom\phi$; see \cite[Cor.~7.3.2 and Thm.~10.1]{Rockafellar1970}.
\end{proof}

As a consequence, Hadamard mirror flow and descent can be applied to this setting.
First, we compute the $Q$-gradient of~$f_v$ at $x\in P$ concretely in terms of the objective~$\Phi$:
\begin{equation}\label{eq:concrete qgrad}
\begin{aligned}
    \nabla^Q f_v(x)
&= \diff Q_{x}(\diff f_v(x))
= \tau_{I\to x} \diff Q_I(\tau_{I \to x}^* \diff f_v(x))
= \tau_{I\to x} \nabla \Phi(\mu(x^{1/2} \cdot v)) \\
&= x^{1/2} \nabla \Phi(\mu(x^{1/2} \cdot v)) x^{1/2},
\end{aligned}
\end{equation}
where the first equality is the definition of the $Q$-gradient~\eqref{eq:grad},
the second follows from the invariance of~$Q$ under parallel transport;
for the third we use \cref{lem:KNderivative}~(b) and identify~$Q_I$ with~$\Phi$ using the Hilbert--Schmidt inner product on~$\bbH\cong T_IP$;
and the last follows from the concrete formula for the parallel transport on~$P$ (\cref{sub:pd,sub:selfconjugate}).
Using this and the concrete formula for the exponential map on~$P$ (\cref{sub:pd,sub:selfconjugate}), we obtain the following concrete formula for the iteration of the Hadamard mirror descent~\eqref{eq:descent}:
\begin{align}
\label{eq:update_momentpolytope}
    x_{k+1}
\coloneqq \exp_{x_k}\mleft(-\eta\nabla^Q f_v(x_k)\mright)
= {x_k}^{1/2} \, e^{-\eta \nabla \Phi(\mu(x_k^{1/2} \cdot v))} \, {x_k}^{1/2}.
\end{align}
We emphasize that $\nabla \Phi$ denotes the usual Euclidean gradient.

The update~$x_k \mapsto x_{k+1}$ in~\eqref{eq:update_momentpolytope} is easily translated into corresponding updates of group elements~$g_k \mapsto g_{k+1}$ in $G$, and of vectors $v_k \mapsto v_{k+1}$ in the representation space~$\VV$:

\begin{lemma}\label{lem:discrete-iterations}
    Let $0 \neq v \in \VV$.
    Let $\phi \colon \AA \to \overline\RR$ be a function that satisfies \ref{it:F1}--\ref{it:F3} and $s(\mu(g\cdot v)) \in \intr\dom\phi$ for all $g \in G$.
    Let $(x_k)_{k \in \NN}$ be the Hadamard mirror descent sequence~\eqref{eq:update_momentpolytope} starting from $x_0 \coloneqq I$.
    \begin{enumerate}
        \item[(a)] Define a sequence $(g_k)_{k \in \NN}$ in $G$ by
        \begin{equation}\label{eq:update group}
            g_{k + 1} \coloneqq e^{-\frac\eta2\nabla\Phi(\mu(g_k\cdot v))}g_k,
            \qquad g_0 \coloneqq I.
        \end{equation}
        Then, it holds that $g_k^\dagger g_k = x_k$ for all $k \in \NN$.
        That is, $g_k = u_k x_k^{1/2}$ for some $u_k \in K$.
        As a consequence, we have $\nabla\Phi(\mu(g_k\cdot v)) = u_k \nabla\Phi(\mu(x_k^{1/2}\cdot v)) u_k^\dagger$ for all~$k\in\NN$.
        \item[(b)] For any choice of nonzero scalars~$c_k\in\CC^\times$, define a sequence $(v_k)_{k \in \NN}$ in $\VV$ by
        \begin{equation}\label{eq:updatev_discrete}
            v_{k + 1} \coloneqq c_{k+1} \pi\mleft(e^{-\frac\eta2\nabla\Phi(\mu(v_k))}\mright)v_k,
            \qquad v_0 \coloneqq c_0v.
        \end{equation}
        Then, $[v_k] = [g_k \cdot v] \in \bbP(\VV)$ for all $k \in \NN$.
        Moreover, $\phi(s(\mu(v_k))) = \Phi(\mu(v_k)) = Q(\diff f_v(x_k))$.
    \end{enumerate}
\end{lemma}
\begin{proof}
    (a)~We show this by induction.
    The base case $k=0$ holds trivially.
    Thus, assume that~$g_k^\dagger g_k = x_k$.
    By the polar decomposition, this is equivalent to $g_k = u_k x_k^{1/2}$ for some~$u_k \in K$, hence
    \begin{align}\label{eq:K cov grad phi}
        \nabla\Phi(\mu(g_k\cdot v))
        = \nabla\Phi(u_k\mu(x_k^{1/2}\cdot v)u_k^\dagger)
        = u_k\nabla\Phi(\mu(x_k^{1/2}\cdot v))u_k^\dagger,
    \end{align}
    where the first equality holds by the $K$-equivariance property~\eqref{eq:mu equivariance} of the moment map, and the second one follows from the $K$-invariance of~$\Phi$, which implies that~$\nabla\Phi$ is $K$-equivariant.
    Thus,
    \begin{align*}
        g_{k+1}
    = e^{-\frac\eta2 \nabla\Phi(\mu(g_k\cdot v))} g_k
    = e^{-\frac\eta2 u_k\nabla\Phi(\mu(x_k^{1/2}\cdot v))u_k^\dagger} u_k x_k^{1/2}
    = u_k e^{-\frac\eta2 \nabla\Phi(\mu(x_k^{1/2}\cdot v))} x_k^{1/2},
    \end{align*}
    which shows that $g_{k+1}^\dagger g_{k+1} = x_{k+1}$.

    (b)~We again proceed by induction.
    The base case is clear, so assume that $[v_k] = [g_k \cdot v]$.
    Then
    \[
        [v_{k+1}]
    = [e^{-\frac\eta2\nabla\Phi(\mu(v_k))} \cdot v_k]
    = [e^{-\frac\eta2\nabla\Phi(\mu(g_k \cdot v))} g_k \cdot v]
    = [g_{k+1} \cdot v],
    \]
    where the second equality uses the induction hypothesis and the fact that the moment map is insensitive to rescaling its argument.
    This establishes that $[v_k] = [g_k \cdot v]$ for all $k\in\NN$.
    To confirm the second claim, note that
    \begin{align}\label{eq:secondclaim}
        Q(\diff f_v(x_k))
    = \Phi\mleft(\mu(x_k^{1/2} \cdot v)\mright)
    = \Phi\mleft(u_k^\dagger \, \mu(g_k \cdot v) \, u_k\mright)
    = \Phi(\mu(v_k))
    = \phi(s(\mu(v_k))),
    \end{align}
    where we first used \cref{eq:Q_momentmap}, then that $g_k = u_k x_k^{1/2}$ by part~(a), along with the equivariance property~\eqref{eq:mu equivariance} of the moment map, then the rescaling invariance of the moment map and the $K$-invariance of $\Phi$, and finally the definition~$\Phi=\phi\circ s$.
\end{proof}

\noindent We remark that the update formula in \cref{eq:updatev_discrete} can also be written in terms of the gradient of~$\phi$:
\begin{equation}\label{eq:updatev_discrete phi}
    v_{k + 1} \coloneqq c_{k+1} \pi\mleft(u^\dagger e^{-\frac\eta2\nabla\phi(p)} u\mright)v_k
\end{equation}
for any $u \in K$ such that $u \mu(v_k) u^\dagger \eqqcolon p \in \AA$.
This formula is obtained from \cref{eq:updatev_discrete} using the relation between the gradients of~$\Phi$ and~$\phi$ in \cref{prop:FandQ}~(b).
It is used in \cref{sub:entanglementpolytope,sub:tensor_power_action}.

\begin{remark}\label{rem:numerical}
While all three update formulas~\eqref{eq:update_momentpolytope}, \eqref{eq:update group} and \eqref{eq:updatev_discrete} are equivalent by the lemma above, they have different advantages and disadvantages in practice due to numerical errors.
If one stores the matrices~$x_k$ or~$g_k$ as in \cref{eq:update_momentpolytope} or \cref{eq:update group}, these can become increasingly ill-conditioned and may in fact diverge, making the matrix function computation harder.
One can avoid dealing with ill-conditioned matrices by storing vectors~$v_k$ instead, as in \cref{eq:updatev_discrete}, which can even be normalized using the freedom in the scalars~$c_k$.
However, in this case, numerical errors during the update might move~$[v_k]$ outside its orbit~$G\cdot [v]$, resulting in minimization over a different moment polytope.
Developing a numerically stable method remains an open problem.
\end{remark}

Properties similar to those in \cref{lem:discrete-iterations} hold for the Hadamard mirror flow~\eqref{eqn:q-gradientflow}.
By \cref{eq:concrete qgrad}, the Hadamard mirror flow can be written as the following first-order differential equation:
\begin{align}\label{eq:flowupdate_momentpolytope}
    \dot x_t = -x_t^{1/2} \nabla \Phi(\mu(x_t^{1/2} \cdot v)) x_t^{1/2}.
\end{align}
Since Kempf--Ness functions are~$C^\infty$, we only need to require sufficient regularity on~$\Phi$.

\begin{lemma}\label{lem:flow concrete}
    Let $0 \neq v \in \VV$.
    Let $\phi \colon \AA \to \overline\RR$ be a function that satisfies \ref{it:F1}--\ref{it:F3} and $s(\mu(g\cdot v)) \in \intr\dom\phi$ for all $g \in G$.
    Suppose that $\Phi$ is twice continuously differentiable on $\intr\dom\Phi$.
    Let~$(x_t)_{t \in [0,T_{\max})}$ denote the unique maximal solution to the Hadamard mirror flow~\eqref{eq:flowupdate_momentpolytope} with~$x_0 \coloneqq I$.
    \begin{enumerate}
        \item[(a)] The maximal solution $(g_t)$ to the first-order differential equation
        \begin{equation}\label{eq:g flow}
            \dot g_t = -\frac12\nabla\Phi(\mu(g_t\cdot v))g_t, \quad g_0 \coloneqq I
        \end{equation}
        is defined on the same interval~$[0,T_{\max})$, and it holds that $g_t^\dagger g_t = x_t$.
        \item[(b)] The maximal solution $(v_t)$ to the first-order differential equation
        \begin{equation}\label{eq:updatev_continuous}
            \dot v_t = -\frac12\Pi\mleft(\nabla\Phi(\mu(v_t))\mright)v_t,\quad v_0 \coloneqq v
        \end{equation}
        is defined for all times~$t\in[0,T_{\max})$, and it holds that $v_t = g_t \cdot v$, as well as $\phi(s(\mu(v_t))) = \Phi(\mu(v_t)) = Q(\diff f_v(x_t))$.
    \end{enumerate}
\end{lemma}
\begin{proof}
\begin{enumerate}
\item[(a)]
Let $(g_t)_{t \in [0,T'_{\max})}$ denote the unique maximal solution to~\eqref{eq:g flow}.
We will show that~$T_{\max} = T'_{\max}$ and $g_t^\dagger g_t = x_t$ for all $t \in [0,T_{\max})$.
To this end, let $y_t \coloneqq g_t^\dagger g_t$.
Then $(y_t)$ is a solution to~\eqref{eq:flowupdate_momentpolytope} starting from~$y_0 = I$ on the time interval~$[0,T'_{\max})$:
\begin{align*}
    \dot{y_t}
= g_t^\dagger \dot g_t + (\dot g_t)^\dagger g_t
= - g_t^\dagger \nabla\Phi(\mu(g_t\cdot v)) g_t
= - y_t^{1/2} \nabla\Phi(\mu(y_t^{1/2}\cdot v)) y_t^{1/2};
\end{align*}
in the second equality we use that $\nabla\Phi(\mu(g_t\cdot v)) \in \bbH$ is Hermitian, and the last step holds by the same reasoning as in \cref{eq:K cov grad phi}.
By ODE uniqueness, it follows that $y_t = g_t^\dagger g_t$ must agree with the maximal solution~$x_t$ on $[0,T'_{\max})$ and~$T_{\max} \geq T'_{\max}$.
It remains to prove that~$T_{\max} = T'_{\max}$.
Indeed, suppose that $T_{\max} > T'_{\max}$.
Then the solution~$\{g_t\}_{t \in [0,T'_{\max})}$ belongs to a compact subset of~$G$, since~$g_t^\dagger g_t = x_t$, hence $g_t \in K \cdot x_t^{1/2}$, and $\{x_t\}_{t \in [0,T'_{\max}]}$ is compact.
By ODE continuation, this means that~$g_t$ can be extended past $T'_{\max}$, which is a contradiction.

\item[(b)]
Let~$(g_t)$ denote the maximal solution from part~(a).
Then~$v_t \coloneqq g_t \cdot v = \pi(g_t) v$ is a solution to~\eqref{eq:updatev_continuous} starting from~$v_0 = v$:
\begin{align*}
    \dot v_t
=   \frac{\diff}{\diff t} \pi(g_t) v
=   \diff\pi(g_t)[\dot g_t] v
=   \Pi(\dot g_t g_t^{-1}) \pi(g_t) v
=   -\frac12\Pi\mleft(\nabla\Phi(\mu(v_t))\mright) v_t,
\end{align*}
where the third equality follows from the definition of the Lie algebra representation.
By ODE uniqueness, $v_t$ must agree with the maximal solution to~\eqref{eq:updatev_continuous} on~$[0,T_{\max})$.
The second claim is proved just like \cref{eq:secondclaim}.
\end{enumerate}
\end{proof}

We next derive a concrete formula for the step-size condition in terms of $\Phi$ and the sequence $\{v_k\}_{k \in \NN}$ defined in \cref{eq:updatev_discrete}.

\begin{lemma}
Let $v \in\VV\setminus\{0\}$ be such that~$s(\mu(g\cdot v)) \in \intr\dom\phi$ for all~$g \in G$.
Then, the step-size condition~\eqref{eq:step-size-condition} for the Hadamard mirror descent sequence~$(x_k)_{k\in\NN}$ generated by \eqref{eq:update_momentpolytope} starting from $x_0 \coloneqq I$ is equivalent to the following condition:
\begin{equation}\label{eq:step-size-KN}
    \Phi(\mu(v_k)) - \Phi(\mu(v_{k+1})) - \tr\mleft[ \mu(v_k) \nabla\Phi(\mu(v_k)) \mright]
    \geq \frac1\eta \log\frac{\|v_{k+1}\|^2}{|c_{k+1}|^2\|v_k\|^2},
\end{equation}
where $(v_k)_{k \in \NN}$ is any sequence of vectors defined as in \cref{eq:updatev_discrete} with the associated scalars $(c_k)_{k \in \NN}$.
\end{lemma}
\begin{proof}
First, we have
\begin{align*}
    \diff f_v(x_k)\mleft[ \exp_{x_k}^{-1}x_{k+1} \mright]
&= \diff f_v(x_k)\mleft[ -\eta\nabla^Q f_v(x_k) \mright]
= -\eta \, \diff f_v(x_k) \mleft[ \tau_{I\to x_k} \nabla \Phi(\mu(x_k^{1/2} \cdot v)) \mright] \\
&= -\eta \tr\mleft[ \mu(x_k^{1/2} \cdot v) \nabla \Phi(\mu(x_k^{1/2} \cdot v)) \mright]
= -\eta \tr\mleft[ \mu(v_k) \nabla \Phi(\mu(v_k)) \mright],
\end{align*}
where the first equality uses the definition of Hadamard mirror descent;
the second equality rewrites this using the penultimate expression for~$\nabla^Q f_v(x_k)$ from \cref{eq:concrete qgrad};
the third equality follows from \cref{lem:KNderivative}~(b);
for the last equality, we use \cref{lem:discrete-iterations}~(a) and~(b), along with the $K$-equivariance property~\eqref{eq:mu equivariance} of the moment map.

By \cref{lem:discrete-iterations}~(b), we also know that $Q(\diff f_v(x_k)) = \Phi(\mu(v_k))$ for all~$k$.
Finally, we observe that
\begin{align*}
    f_v(x_{k+1}) - f_v(x_k)
&= \log \frac {\braket{v, g_{k+1}^\dagger g_{k+1} \cdot v}}  {\braket{v, g_k^\dagger g_k \cdot v}}
= \log \frac {\|g_{k+1}\cdot v\|^2}  {\|g_k\cdot v\|^2}
= \log \frac {\|e^{-\frac\eta2\nabla\Phi(\mu(g_k\cdot v))}\cdot(g_k\cdot v)\|^2}  {\|g_k\cdot v\|^2}\\
&= \log \frac {\|e^{-\frac\eta2\nabla\Phi(\mu(v_k))}\cdot v_k\|^2}  {\|v_k\|^2}
= \log \frac {\|v_{k+1}\|^2}  {|c_{k+1}|^2\|v_k\|^2},
\end{align*}
using \cref{lem:discrete-iterations}: the first equality follows from $x_k = g_k^\dagger g_k$, the second equality follows from $\pi(g^\dagger) = \pi(g)^\dagger$ for all $g \in G$, the third equality follows from \cref{eq:update group}, the fourth equality follows from $[v_k] = [g_k \cdot v]$, and the last equality follows from \cref{eq:updatev_discrete}.
If we substitute the above into the step-size condition~\eqref{eq:step-size-condition}, we obtain \cref{eq:step-size-KN}.
\end{proof}

We now instantiate our convergence results (\cref{thm:main-continuous,thm:main-discrete}) in this setting.
In particular, we discuss the quantity $\inf_{x \in P}f_v^\xi(x)$, which must be finite to obtain the desired~$\mathcal{O}(k^{-1})$ and~$\mathcal{O}(t^{-1})$ convergence rates.
First, we fix some conventions.
Recall that the Busemann function~$b^\xi$ is only defined up to an additive constant.
We will make the following choice: using the natural identification $CP^\infty \cong T_I P \cong \bbH$, we can associate to any equivalence class~$\xi\in CP^\infty$ the unique tangent vector~$X\in\bbH$ such that~$\xi$ is represented by the geodesic ray $t \mapsto e^{tX}$.
We write~$b^X$ for the Busemann function corresponding to this geodesic ray.
Explicitly:
\begin{align*}
    b^X \colon P \to \RR, \quad
    b^X(x) \coloneqq \lim_{t\to\infty} \parens*{ \norm{X}_{\mathrm F} \norm{\log(x^{-1/2}e^{Xt}x^{-1/2})}_{\mathrm F} - \norm{X}_{\mathrm F}^2 t }.
\end{align*}
In particular, this fixes the additive constant to be~$b^X(I)=0$.
We analogously write~$f_v^X \coloneqq f_v + b^X$.
Then we have, for all~$X \in \bbH$,
\[
  \frac12 f_v^X(I) = \frac12 f_v(I) = \frac12 \log \norm{v}^2 = \log \norm v.
\]
The infimum is the logarithm of the so-called ``capacity''.
We follow the conventions of \cite{FW2020}.

\begin{definition}[Capacity]
Let $X\in\bbH$.
We define the \emph{logarithmic $X$-capacity} of~$v\in\VV\setminus\{0\}$ as
\[
    \log\capacity_X(v) \coloneqq \frac12\inf_{x \in P}f_v^X(x) \in [-\infty,\infty).
\]
\end{definition}
\noindent
This quantity has been introduced in~\cite{BFGOWW2019} as a generalization of the notions of matrix, operator, polynomial and tensor capacity~\cite{gurvits1998deflation,gurvits2004classical,gurvits2006hyperbolic,BFGOWW2018} for~$X\in C$ in a Weyl chamber.
It was subsequently generalized to all of~$\bbH$ in~\cite{FW2020,botero2021large} (cf.\ the \emph{rate function} in the latter work).%
\footnote{We caution that the various definitions in the literature differ from each other by inconsequential conventions. Here we follow the conventions of~\cite{FW2020}.}
The shifted Kempf--Ness function satisfies the equivariance $f^X_{u \cdot v}(x) = f^{u^\dagger X u}_v(u^\dagger x u)$, so
\begin{align}\label{eq:log cap K equivariance}
    \log\capacity_X(u \cdot v)
= \log\capacity_{u^\dagger X u}(v)
\qquad (u \in K, X \in \bbH).
\end{align}

Now consider the choice $X \coloneqq \mu(v) \in \bbH$.
Then, the moment map~$\mu(v) \in \bbH$ equals~$\diff f_v(I)$ (\cref{lem:KNderivative}~(b)), which in turn means that the geodesically convex~$f_v^{\mu(v)}$ has a minimum at~$I$,~hence
\[ \log\capacity_{\mu(v)}(v) = \log\norm v. \]
If we take $u \in K$ such that $p \coloneqq s(\mu(v)) = u\mu(v)u^\dagger = \mu(u\cdot v)$, we have
$\log\capacity_p(u\cdot v) = \log\norm{u \cdot v} = \log\norm v$
by the $K$-invariance of the norm.
By applying this discussion to any~$x\cdot v$ ($x \in P$) in place of~$v$, one finds that for any $p \in \{s(\mu(x\cdot v)) \mid x \in P\}$, there exists $g \in G$ such that $\log\capacity_p(g\cdot v) > -\infty$.
In fact, it is known that this last condition not only characterizes a dense set of the moment polytope, but is \emph{equivalent} to the condition that $p$ belongs to the moment polytope:

\begin{proposition}[\cite{NessMumford1984,Brion-06,BFGOWW2019,FW2020}]\label{prop:shiftingtrick_orig}
    For any $p \in C$, it holds that
    \[
        p \in \Delta(v)
    \iff \exists g \in G,\ \log\capacity_p(g\cdot v) > -\infty
    \iff \exists u \in K,\ \log\capacity_p(u \cdot v) > -\infty.
    \]
\end{proposition}

This equivalence is known as the \emph{shifting trick} and it was originally stated for rational points in the Weyl chamber, where it has a representation-theoretic interpretation;
the generalization to arbitrary points and the second characterization are stated, e.g., in \cite{FW2020}.
It is also known that if there exists $g \in G$ as above, then this is the case for generic~$g \in G$.
We give an elegant proof of the second equivalence in \cref{lem:shiftingtrick} that directly applies to arbitrary points~$p \in C$ (even~$X \in \bbH$).
The shifting trick was used in~\cite{BFGOWW2019} to give algorithms for deciding membership of points in the moment polytope by minimization of the shifted Kempf--Ness function $f_{g\cdot v}^p$ for a random choice of~$g\in G$.
Here, we make use of it for the convergence analysis of Hadamard mirror flow and descent.

Following~\cite{FW2020}, and motivated by \cref{prop:shiftingtrick_orig}, we also define the following quantity.

\begin{definition}[$K$-invariant capacity]
Let~$p\in C$.
Then we define the \emph{$K$-invariant logarithmic $p$-capacity} of $v\in\VV\setminus\{0\}$ as
\[
    \log\Capacity_p(v)
\coloneqq \sup_{u \in K} \log\capacity_p(u \cdot v)
= \sup_{u \in K} \log\capacity_{u^\dagger p u}(v)
= \sup_{X \in \bbH, \, s(X)=p} \log\capacity_X(v).
\]
\end{definition}

\noindent
This quantity is indeed $K$-invariant by definition:
we have $\log\Capacity_p(u \cdot v) = \log\Capacity_p(v)$ for~$u \in K$.
Now it follows from \cref{prop:shiftingtrick_orig} that, for every~$p \in C$,
\begin{align}\label{eq:Delta vs logCap}
    p \in \Delta(v)
\iff \log\Capacity_p(v) > -\infty.
\end{align}

We now state the main result of this section: an $\mathcal{O}(t^{-1})$ and $\mathcal{O}(k^{-1})$ convergence rate for moment polytope optimization using the Hadamard mirror flow and descent, respectively.
Remarkably, the constant in the $\mathcal{O}$ notation depends on the $K$-invariant logarithmic capacity of a minimizer~$p^\star$ of the objective~$\Phi$ over the moment polytope (as well as the initial norm of the vector and, in the discrete-time case, the inverse step size).
As explained following \cref{eq:min_over_momentpolytope}, such a minimizer always exists and by \cref{eq:Delta vs logCap} its $K$-invariant logarithmic capacity~$\log\Capacity_{p^\star}(v)$ is~finite.

\begin{theorem}[Convergence for moment polytope optimization]\label{thm:KN-general}
    Let $0 \neq v \in \VV$.
    Let $\phi \colon \AA \to \overline\RR$ be a function that satisfies \ref{it:F1}--\ref{it:F3} and $s(\mu(g\cdot v)) \in \intr\dom\phi$ for all $g \in G$.
    Let $p^\star$ be a minimizer of~$\phi$ in~$\Delta(v)$.
    Then, the following convergence rates hold:

    \begin{enumerate}
    \item[(a)] \textbf{Discrete-time:} Let $(v_k)_{k \in \NN}$ be the sequence generated by \cref{eq:updatev_discrete}. Under the step-size condition~\eqref{eq:step-size-KN}, it holds that, for all $k\geq1$,
    \[
        \phi\mleft(s(\mu(v_k))\mright) - \min_{p \in \Delta(v)}\phi(p) \le 2 \frac{\log\|v\| - \log\Capacity_{p^\star}(v)}{\eta k}.
    \]
    \item[(b)] \textbf{Continuous-time:}
    Assume that $\Phi$ is twice continuously differentiable on the interior of its domain and that the solution~$(x_t)$ to \cref{eq:flowupdate_momentpolytope} with initial condition~$x_0\coloneqq I$ is defined for all times~$t\geq0$.
    Then the solution $(v_t)_t$ to \cref{eq:updatev_continuous} is also defined for all times and it holds that, for all $t>0$,
    \[
        \phi\mleft(s(\mu(v_t))\mright) - \min_{p \in \Delta(v)}\phi(p) \le 2 \frac{\log\|v\| - \log\Capacity_{p^\star}(v)}t.
    \]
    \end{enumerate}
\end{theorem}
\begin{proof}
For every $u \in K$, choose~$\xi \equiv X \coloneqq u^\dagger p^\star u$ (under the identification $CP^\infty \cong \bbH$).
Then we have
\[ Q(\xi) = \Phi(u^\dagger p^\star u) = \phi(p^\star) = \min_{p \in \Delta(v)}\phi(p) \]
(see \cref{ex:Q_xi} for the first equality).
On the other hand,
\begin{align*}
    f_v^\xi(I) - \inf_{x \in P} f_v^\xi(x) = 2 \left( \log\norm v - \log\capacity_{u^\dagger p^\star u}(v) \right).
\end{align*}
We also know that $\phi(s(\mu(v_k))) = Q(\diff f_v(x_k))$ for the descent sequence by \cref{lem:discrete-iterations}~(b), while $\phi(s(\mu(v_t))) = Q(\diff f_v(x_t))$ for the solution to the flow by \cref{lem:flow concrete}~(b).
Thus the theorem follows by applying \cref{thm:main-discrete,thm:main-continuous} to $\xi \equiv X = u^\dagger p^\star u$, for every $u \in K$, and taking the infimum over~$u \in K$, using $\log\Capacity_{p^\star}(v) = \sup_{u \in K} \log\capacity_{u^\dagger p^\star u}(v)$.
\end{proof}

\Cref{thm:KN-general} depends on two quantities that need to be bounded to obtain concrete convergence guarantees:
the $K$-invariant capacity of the (a priori unknown) minimizer~$p^\star$, and a suitable step size~$\eta$.
For the former, polynomial lower bounds can be shown in important cases by combining known $p$-capacity lower bounds and randomness bounds;
this will be discussed in \cref{sub:entanglementpolytope,sub:GLaction}.
For the latter, a useful ingredient is the dual smoothness of the Kempf--Ness functions, which will be established in the following subsection; this allows choosing an appropriate step size for arbitrary smooth objectives~$\Phi$ by \cref{prop:dual-to-relative}.

Since the moment map $\mu \colon \bbP(\VV) \to \bbH$ is continuous and $\phi$ is l.s.c., the theorem above immediately implies accumulation in projective space:

\begin{corollary}
Under the hypotheses of \cref{thm:KN-general}~(a), any accumulation point of~$[v_k]_{k \in \NN}$ is a minimizer of $\Phi\circ\mu = \phi\circ s\circ\mu$ in the orbit closure~$\overline{G\cdot[v]} \subseteq \bbP(\VV)$.
Under the hypotheses of \cref{thm:KN-general}~(b), the same holds for any accumulation point of $[v_t]_t$ (as $t\to\infty$).
\end{corollary}

\noindent
When $\phi$ is the squared norm, it is known that $[v_t]_t$ converges~\cite[Theorem~3.3]{GRS_Book}.
We do not know whether such convergence also holds in the discrete-time case, or for other objectives:

\begin{question}
    Under which hypotheses do $[v_t]_t$ and $[v_k]_{k \in \NN}$ converge?
\end{question}

Vectors that arise in applications often possess additional symmetries; for example, a tensor may be invariant under certain local unitaries or under a permutation of its factors (\cref{rem:tensor symmetries}).
In this case, it is natural to ask whether the discrete sequence $(v_k)_{k\in\NN}$ and the continuous flow $(v_t)_{t\in[0,T_{\max}) }$ retain such symmetries.
The following proposition shows that every unitary symmetry compatible with the representation and the objective is preserved by both dynamics.

\begin{proposition}[Conservation of symmetries]\label{prop:symmetry}
    Let $0 \neq v \in \VV$.
    Let $\phi \colon \AA \to \overline\RR$ be a function that satisfies \ref{it:F1}--\ref{it:F3} and suppose $s(\mu(g\cdot v)) \in \intr\dom\phi$ for all $g \in G$.
    For the assertions concerning the Hadamard mirror flow, assume additionally that~$\Phi$ is twice continuously differentiable on~$\intr\dom\Phi$.
    Let~$(v_k)_{k\in\NN}, (v_t)_{t\in[0,T_{\max})}$ be as in \cref{eq:updatev_discrete,eq:updatev_continuous}.
    \begin{enumerate}
        \item[(a)] Let $U \in \U(\VV)$ be a unitary such that $Uw = w$ implies $U \Pi(\nabla\Phi(\mu(w))) U^\dagger = \Pi(\nabla\Phi(\mu(w)))$ for all $w \in \VV \setminus \{0\}$ with $\mu(w) \in \intr\dom\Phi$.
        If $Uv = v$, then
        \begin{equation}
            \label{eq:U-fixes}
            Uv_k = v_k\quad (\forall k\in\NN) \quad\text{and}\quad Uv_t = v_t\quad (\forall t\in [0,T_{\max}) ).
        \end{equation}
        \item[(b)]
        The hypothesis on~$U$ in part~(a) is satisfied if there exists $R\in\Or(\bbH)$, orthogonal with respect to the Hilbert--Schmidt inner product, such that $\Phi$ is $R$-invariant and $U \Pi(X) U^\dagger = \Pi(RX)$ for all~$X\in \bbH$.
        Thus, in this case, if $Uv=v$ then \cref{eq:U-fixes} holds.
        \item[(c)]
        In particular, if $u\in K$ satisfies $u \cdot v=v$, then
        \[
         u\cdot v_k = v_k\quad (\forall k\in\NN)\quad\text{and}\quad u\cdot v_t = v_t\quad (\forall t\in [0,T_{\max}) ).
        \]
    \end{enumerate}
\end{proposition}
\begin{proof}
(a)~We prove $U v_k=v_k$ for all~$k$ by induction.
The base case $k=0$ holds by assumption since~$v_0\propto v$.
Thus, assume that $Uv_k=v_k$.
Then,
\begin{align*}
    Uv_{k+1}
&= c_{k+1} U e^{-\frac{\eta}2\Pi(\nabla\Phi(\mu(v_k)))} v_k
= c_{k+1} e^{-\frac{\eta}2U \Pi(\nabla\Phi(\mu(v_k))) U^\dagger} U v_k
= c_{k+1} e^{-\frac{\eta}2 \Pi(\nabla\Phi(\mu(v_k)))} v_k
= v_{k+1},
\end{align*}
where the first and last equalities follow from \cref{eq:updatev_discrete} and~$\pi(e^X) = e^{\Pi(X)}$ for~$X\in\bbH$, the second equality uses $U e^M U^\dagger = e^{U M U^\dagger}$, and the penultimate equality follows from our assumption  and the induction hypothesis.

To prove that $U v_t = v_t$ for all $t\in[0,T_{\max})$, we observe that the vector field defining the flow preserves the fixed-point subspace of~$U$:
if $Uw=w$, then the vector field~$V_w \coloneqq -\frac12\Pi(\nabla\Phi(\mu(w))) w$ at that point satisfies $U V_w = V_w$ directly from our assumption on~$U$.
By ODE uniqueness, it follows that the solution starting at $v = Uv$ remains in this subspace.

(b)~For every $w\in\VV\setminus\{0\}$ and $X\in\bbH$, we have
\[
\tr[\mu(Uw)X]
=\frac{\langle Uw,\Pi(X)Uw\rangle}{\langle w,w\rangle}
=\frac{\langle w,\Pi(R^{-1}X)w\rangle}{\langle w,w\rangle}
=\tr\mleft[\mu(w)\left(R^{-1}X\right)\mright]
=\tr\mleft[\left(R\mu(w)\right)X\mright],
\]
where the first and the penultimate equalities follow from the definition~\eqref{eq:moment map} of the moment map;
the second equality uses the unitarity of~$U$ and our assumption on its relation to~$R$,
and the last equality uses the orthogonality of~$R$ with respect to the Hilbert--Schmidt inner product.
Thus:
\begin{equation}\label{eq:U-R-equiv}
    \mu(Uw) = R\mu(w) \qquad (\forall w\in\VV\setminus\{0\}).
\end{equation}
Next, we observe that the $R$-invariance of $\Phi$ implies the $R$-equivariance of its gradient:%
\footnote{We already used this fact in the proof of \cref{lem:discrete-iterations} in the special case of the $K$-action.}
\begin{equation}\label{eq:grad-Phi-equiv}
\nabla\Phi(RX)= R\,\nabla\Phi(X) \qquad (\forall X \in \intr\dom\Phi).
\end{equation}
To see this, note that for every $Y\in\bbH$,
\[
    \tr[\nabla\Phi(RX) \, Y]
    =\diff\Phi(RX)[Y]
    =\diff\Phi(X)\mleft[R^{-1}Y\mright]
    =\tr\mleft[\nabla\Phi(X) \left(R^{-1}Y\right)\mright]
    =\tr[(R \, \nabla\Phi(X))\, Y],
\]
where the second equality follows from the chain rule applied to $\Phi = \Phi \circ R$.
Using \cref{eq:U-R-equiv,eq:grad-Phi-equiv} and our assumption on the relation between~$U$ and~$R$, we obtain
\begin{equation*}
    U\Pi\mleft(\nabla\Phi(\mu(w))\mright)U^\dagger = \Pi\mleft(R\nabla\Phi(\mu(w))\mright) = \Pi\mleft(\nabla\Phi(R\mu(w))\mright) = \Pi\mleft(\nabla\Phi(\mu(Uw))\mright),
\end{equation*}
for every $w \in \VV \setminus \{0\}$ such that $\mu(w) \in \intr\dom\Phi$.
Thus the hypothesis in part~(a) is satisfied.

(c)~This follows at once from part~(b) applied to $U \coloneqq \pi(u) \in \U(\VV)$ and $R(X) \coloneqq u X u^\dagger$.
The map~$\Phi$ is $R$-invariant by~\ref{it:Phi2}, which is equivalent to assumption~\ref{it:F2}, and we have $U \Pi(X) U^\dagger = \pi(u) \Pi(X) \pi(u)^\dagger = \Pi(u X u^\dagger) = \Pi(RX)$ by $K$-equivariance of the Lie algebra action.
\end{proof}

\begin{remark}\label{rem:tensor symmetries}
In addition to the symmetries induced by elements of~$K$, \cref{prop:symmetry} applies to ``external'' symmetries, such as permutations of tensor factors.
Consider, for example, the tensor action discussed in \cref{sub:entanglementpolytope} when all factors have the same dimension, that is, the canonical action of $G \coloneqq \GL(n)^d$ on $\VV\coloneqq (\CC^n)^{\otimes d}$ (or on tensor tuples).
Any permutation~$\sigma\in S_d$ gives rise to a unitary~$U_\sigma\in \U(\VV)$, which acts by permuting the tensor factors, and to an orthogonal transformation~$R_\sigma \in \mathrm{O}(\bbH)$, which acts by permuting the components of~$\bbH=\Herm(n)^{\op d}$.
These satisfy $U_{\sigma}\Pi(X)U_{\sigma}^\dag=\Pi(R_{\sigma}X)$ for all~$X\in\bbH$.
Thus, whenever a tensor~$v \in \VV$ and the objective $\phi$ (equivalently, $\Phi$) are invariant under a permutation~$\sigma$, this symmetry is preserved throughout the Hadamard mirror descent sequence $(v_k)_{k\in\NN}$ and flow $(v_t)_{t\in [0,T_{\max})}$.
This applies, in particular, to objective functions of the form $\phi_{\Sum}$ that we will consider later; see \cref{eq:phisum}.
\end{remark}

\subsection{Smoothness and dual-smoothness of Kempf--Ness functions}\label{sub:dualsmooth}
It is known that Kempf--Ness functions are smooth.
However, we require \emph{dual-smoothness} as defined in \cref{def:dual smooth} to control the step size of Hadamard mirror descent.
We will now show that Kempf--Ness functions are dual-smooth with a natural parameter (\cref{thm:kn-dual-smooth}).
Although not required for our applications, we also determine the optimal smoothness parameter, improving on the known results in the literature~\cite{BFGOWW2019,HiraiSakabe2024}, and show that in the Hadamard setting, these two parameters are in general distinct (\cref{prop:kn-smooth,ex:countersmooth}).
As in the preceding subsections, we assume that~$\pi\colon G \to \GL(\VV)$ is a representation of a connected self-adjoint subgroup~$G\subseteq\GL(n)$.

\begin{definition}[Weight diameter and norm]\label{def:weight diameter}
We define the \emph{weight diameter} of the representation $\pi \colon G \to \GL(\VV)$, with associated Lie algebra representation~$\Pi \colon \Lie(G) \to \Lin(\VV)$, as
\begin{align*}
    D(\pi)
\coloneqq \max_{X \in \bbH, \, \norm X_{\mathrm F} = 1} \bigl( \lambda_{\max}(\Pi(X)) - \lambda_{\min}(\Pi(X)) \bigr)
= 2 \max_{X \in \bbH, \, \norm X_{\mathrm F} = 1} \min_{s \in \RR} \; \norm{\Pi(X) - s I}_{\ope},
\end{align*}
where $\lambda_{\max}$ and~$\lambda_{\min}$ denote the largest and smallest eigenvalues of a Hermitian operator, respectively.
The \emph{weight norm} of~$\pi$ is defined as
\begin{align*}
    N(\pi)
\coloneqq \norm{\Pi}_{\mathrm{F} \to \ope}
= \max_{X \in \bbH, \, \norm X_{\mathrm F} = 1} \norm{\Pi(X)}_{\ope}.
\end{align*}
\end{definition}

\noindent
The terminology is motivated by the notion of weights in representation theory:
one can associate to~$\pi$ a finite set of vectors in a Euclidean space, called the set of \emph{weights};
then $D(\pi)$ is the diameter of this set, while $N(\pi)$ is the maximal norm of any element in the set.%
\footnote{\label{foot:weights}We briefly sketch the connection.
As discussed in \cref{sub:selfconjugate}, every matrix~$X \in \bbH$ is conjugate to one in the Cartan subspace~$\AA$ by a unitary in~$K$.
Thus~$D(\pi)$ and~$N(\pi)$ can be computed by maximizing over elements~$X \in \AA$.
Now, $\AA$~is a linear space consisting of pairwise commuting Hermitian operators, hence the same is true for its image~$\Pi(\AA)$ under the Lie algebra representation; we may thus jointly diagonalize the operators in~$\Pi(\AA)$.
Because~$\Pi$ is linear, the eigenvalue of~$\Pi(X)$ on each common eigenspace must depend linearly on~$X\in\AA$.
In other words, there exists a finite set of so-called \emph{weights}~$\Omega \subseteq \AA^*$ such that~$\VV = \bigoplus_{\omega\in\Omega} \VV_\omega$ and each $\Pi(X)$ acts on the \emph{weight subspace}~$\VV_\omega$ by multiplication by~$\omega(X)$.
Equip~$\AA^*$ with the inner product inherited from~$\AA$.
Then it is easy to see that~$D(\pi) = \max_{\omega,\omega'\in\Omega} \norm{\omega-\omega'}$, while $N(\pi) = \max_{\omega\in\Omega} \norm\omega$.}
Consistent with this interpretation, but also directly from the definition, these are related by
\begin{align}\label{eq:weight diam vs norm}
    D(\pi) \leq 2 N(\pi).
\end{align}
In general, this bound is strict (see \cref{ex:countersmooth} below).

The weight norm was defined in~\cite{BFGOWW2019}, where it was shown that the Kempf--Ness function~$f_v$ is $N(\pi)^2$-smooth; see \cite[Prop.~3.13]{BFGOWW2019} (but note that their Riemannian metric differs by a constant factor) or \cite[Lem.~4.2~(1)]{HiraiSakabe2024}.
We believe the weight diameter has not been defined before; however, it is a natural quantity that generalizes the diameter of the Newton polytope, which appeared as a condition number in~\cite{BLNW2020}.
We will now show that the Kempf--Ness function~$f_v$ is in fact $D(\pi)^2/4$-smooth, which is optimal.

\begin{proposition}[Smoothness of Kempf--Ness functions]\label{prop:kn-smooth}
Let~$\pi \colon G \to \GL(\VV)$ be as above.
Then the Kempf--Ness function~$f_v$ is $D(\pi)^2/4$-smooth for all~$v \in \VV\setminus\{0\}$.
That is:
\begin{equation}\label{eq:smooth concrete}
    \frac{\diff^2}{\diff t^2} f_v\mleft(g e^{Xt} g^\dagger\mright) \leq \frac{D(\pi)^2}4 \norm{X}_{\mathrm F}^2 \qquad (t \in \RR, g \in G, X \in \bbH).
\end{equation}
Moreover, this is optimal: no smaller constant satisfies~\eqref{eq:smooth concrete} for all~$v \in \VV\setminus\{0\}$.
\end{proposition}
\begin{proof}
By the $G$-equivariance $f_v(g^\dagger xg) = f_{g\cdot v}(x)$, it suffices to show \cref{eq:smooth concrete} for~$g = I$ and~$t=0$.
Thus, define $h \colon \RR \to \RR$ by $h(t) \coloneqq f_v(e^{Xt}) = \log \braket{v, e^{\tilde Xt} v}$, where $\tilde X \coloneqq \Pi(X)$.
We may further assume that~$\norm v=1$, since rescaling~$v$ only changes~$h$ by an additive constant, which does not impact the derivatives.
The first derivative is $h'(t) = \frac {\braket{v, \tilde X e^{\tilde Xt} v}} {\braket{v, e^{\tilde Xt} v}}$ and hence, using $\norm v=1$,
\begin{align*}
    h''(0) = \braket{v, \tilde X^2 v} - \braket{v, \tilde X v}^2.
\end{align*}
This ``variance-type'' expression is invariant under substituting~$\tilde X \mapsto \tilde X - sI$ for any~$s\in\RR$.
Thus:
\begin{align}\label{eq:variance bound}
    h''(0)
\leq \min_{s \in \RR} \braket{v, (\tilde X - s I)^2 v}
\leq \min_{s \in \RR} \, \norm{\tilde X - s I}^2_{\ope}
\leq \frac {D(\pi)^2} 4 \norm{X}_{\mathrm F}^2
\end{align}
from the definition of the weight diameter.
This proves the desired smoothness.

To see that this cannot be improved, choose~$X \in \bbH$ with~$\norm X_{\mathrm F}=1$ that achieves the weight diameter, i.e., $D(\pi) = \lambda_{\max}(\Pi(X)) - \lambda_{\min}(\Pi(X))$.
If $D(\pi)=0$ there is nothing to show.
Otherwise, $\lambda_{\max}(\Pi(X)) > \lambda_{\min}(\Pi(X))$, so we can find orthonormal eigenvectors $v_{\max}$ and $v_{\min}$ of~$\Pi(X)$ corresponding to these eigenvalues.
Then $v = \frac1{\sqrt2} (v_{\max} + v_{\min})$ saturates \cref{eq:variance bound}.
\end{proof}

Next, we show that $f_v$ is also dual-smooth with parameter~$D(\pi)^2/2$ (which is \emph{twice} the optimal smoothness parameter determined above).
We caution that this dual-smoothness parameter is \emph{not} in general optimal: e.g., if~$G$ is commutative, then $P$ is Euclidean, and hence smoothness and dual-smoothness coincide (as discussed in \cref{sub:rel-smooth}).%
\footnote{In fact, since the weight system is the same for a group and its maximally commuting subgroup, the optimal dual-smoothness parameter cannot in general be computed from this data alone (without using, e.g., the root system or other information sensitive to the non-commutativity of the group). See \cref{ex:countersmooth} for an explicit example.}
However, \cref{thm:kn-dual-smooth} is still tight in the sense that there are representations such that no smaller constant is valid for all vectors (we give such an example in \cref{ex:countersmooth} below).

\begin{theorem}[Dual smoothness of Kempf--Ness functions]\label{thm:kn-dual-smooth}
Let~$\pi \colon G \to \GL(\VV)$ be as above.
Then the Kempf--Ness function~$f_v$ is $D(\pi)^2/2$-dual-smooth for all~$v \in \VV\setminus\{0\}$.
That is:%
\footnote{\Cref{eq:dual-smooth repeated} is equivalent to the original definition of dual-smoothness in \cref{eqn:dual-smooth} since parallel transport is an isometry. We have also multiplied both sides of the inequality by a factor of~2.}
\begin{equation}\label{eq:dual-smooth repeated}
    \norm{ \diff f_v(y) - \tau^{\ast}_{y \to x}\diff f_v(x) }^2 \leq D(\pi)^2 \left( f_v(x) - f_v(y) - \diff f_v(y)\mleft[\exp_{y}^{-1}x\mright]\right) \qquad (x,y\in P).
\end{equation}
\end{theorem}
\begin{proof}
By the $G$-equivariance $f_v(g^\dagger xg) = f_{g\cdot v}(x)$, it suffices to show \cref{eq:dual-smooth repeated} for~$y = I$.
We may further assume that~$\norm v=1$, since rescaling~$v$ only changes~$f_v$ by an additive constant, which does not impact~\cref{eq:dual-smooth repeated}.
We start by computing the difference of cotangent vectors on the left-hand side: using \cref{lem:KNderivative}~(b), we find that, for any~$Z \in \bbH$,
\begin{align*}
    \left( \diff f_v(I) - \tau^{\ast}_{I \to x}\diff f_v(x) \right)\mleft[Z\mright]
=   \braket{v, \Pi(Z) v} - \braket{w, \Pi(Z) w}
=   \tr\mleft[ \Pi(Z) \left( v v^\dagger - w w^\dagger \right) \mright],
\end{align*}
where $w = (x^{1/2} \cdot v) / \norm{x^{1/2} \cdot v}$.
Because $v v^\dagger - w w^\dagger$ is traceless, this expression is invariant under substituting~$\Pi(Z) \mapsto \Pi(Z) - sI$ for any~$s\in\RR$.
Thus:
\begin{align*}
    \norm{ \diff f_v(I) - \tau^{\ast}_{I \to x}\diff f_v(x) }^2
= \max_{Z \in \bbH, \, \norm Z_{\mathrm F}=1} \min_{s \in \RR} \tr\mleft[ \left( \Pi(Z) - s I \right) \left( v v^\dagger - w w^\dagger \right) \mright]^2
\leq \frac {D(\pi)^2} 4 \norm{v v^\dagger - w w^\dagger}_{\tr}^2
\end{align*}
by duality and the definition of the weight diameter.
On the other hand:
\begin{align*}
    \frac14 \norm{v v^\dagger - w w^\dagger}_{\tr}^2
=   1 - \abs{\braket{v,w}}^2
=   1 - \frac {\braket{v, x^{1/2} \cdot v}^2} {\braket{v, x \cdot v}}
\leq - \log \frac {\braket{v, x^{1/2} \cdot v}^2} {\braket{v, x \cdot v}},
\end{align*}
where the first equality is the standard formula for the trace norm distance of two rank-one projections in terms of the overlap of the corresponding unit vectors (e.g., \cite[(1.186)]{watrous});
the second equality is obtained by substituting the definition of~$w$;
and the final inequality uses the general fact that $\log s \leq s - 1$ for $s\in\RR_{>0}$.
Write $x = e^X$ for some~$X \in \bbH$, and consider~$h(t) \coloneqq f_v(e^{Xt})$.
This is the restriction of the Kempf--Ness function to a constant-speed geodesic, so it is convex.
Thus:
\begin{align*}
  - \log \frac {\braket{v, x^{1/2} \cdot v}^2} {\braket{v, x \cdot v}}
&= h(1) - 2 h(1/2)
\leq h(1) - 2 \left( h(0) + \tfrac12 h'(0) \right)
= h(1) - h'(0) \\
&= f_v(x) - f_v(I) - \diff f_v(I)[\exp_I^{-1}x],
\end{align*}
where we first use the definition of~$h$, then its convexity, and in the last two steps we substitute
$h(0)=f_v(I)=0$,
$h(1) = f_v(x)$,
and $h'(0) = \diff f_v(I)[X] = \diff f_v(I)[\exp_I^{-1}x]$.
This proves the desired dual-smoothness.
\end{proof}

The dual-smoothness of the Kempf--Ness function implies an immediate corollary of \cref{thm:KN-general}.
When the objective~$\phi$ is $L$-smooth, the step size can be set to $\eta = 2/(D(\pi)^2 L)$:

\begin{corollary}[Convergence for smooth moment polytope optimization]\label{cor:KN-Qsmooth}
    Let $0 \neq v \in \VV$.
    Let $\phi \colon \AA \to \RR$ be an $L$-smooth convex Weyl group invariant function.
    Let $p^\star$ be a minimizer of~$\phi$ in~$\Delta(v)$.
    Let $(v_k)_{k \in \NN}$ be the sequence generated by \cref{eq:updatev_discrete} with step size~$0 < \eta \leq 2/(D(\pi)^2 L)$.
    Then, for all $k\geq1$,
    \[
        \phi\mleft(s(\mu(v_k))\mright) - \min_{p \in \Delta(v)}\phi(p) \le \frac{2 \left( \log\|v\| - \log\Capacity_{p^\star}(v) \right)}{\eta k}.
    \]
    In particular, taking $\eta = 2/(D(\pi)^2 L)$ gives
    \[
        \phi\mleft(s(\mu(v_k))\mright) - \min_{p \in \Delta(v)}\phi(p) \le \frac{D(\pi)^2 L \left( \log\|v\| - \log\Capacity_{p^\star}(v) \right)}k.
    \]
\end{corollary}
\noindent We emphasize that $\log\Capacity_{p^\star}(v)$ is finite by \cref{eq:Delta vs logCap}, as discussed above \cref{thm:KN-general}.
\begin{proof}
Note that any smooth convex Weyl group invariant function satisfies conditions~\ref{it:F1}, \ref{it:F2} and \ref{it:F3}.
Moreover, trivially, we have $s(\mu(g\cdot v)) \in \intr\dom\phi = \AA$ for all $g \in G$ because smooth functions are finite-valued by assumption.
Now, \cref{thm:kn-dual-smooth,prop:dual-to-relative,prop:FandQ}~(c), together with our assumption that~$\phi$ is~$L$-smooth imply that the pair~$(f_v,Q)$ is relatively $L D(\pi)^2/2$-smooth.
By \cref{lem:relsmooth}, it follows that the step-size condition is satisfied with our choice of~$\eta$.
Now the result follows by applying \cref{thm:KN-general}.
\end{proof}

\begin{example}\label{ex:countersmooth}
We now exhibit a group and representation such that, for every~$\delta\in(0,1)$, we can find a vector~$v$ such that the corresponding Kempf--Ness function~$f_v$ is $1/2$-smooth, but not $(1-\delta)$-dual-smooth.
This shows that in the Hadamard setting the two parameters can be distinct and establishes the tightness of \cref{thm:kn-dual-smooth}.
For~$n\geq2$, let $G=\GL(n)$ act on $\VV=\CC^n$ by matrix-vector multiplication, i.e., $\pi(g) = g$ and hence~$\Pi(X)=X$.
Then $P=\PD(n)$, $\bbH=\Herm(n)$, and we see that
\begin{align}\label{eq:fundamental diam}
    D(\pi)
= \max_{X \in \Herm(n), \, \norm X_{\mathrm F} = 1} \left( \lambda_{\max}(X) - \lambda_{\min}(X) \right)
= \sqrt 2,
\end{align}
hence $f_v$ is $\frac12$-smooth and $1$-dual-smooth for every~$v \in \VV\setminus\{0\}$ by \cref{prop:kn-smooth,thm:kn-dual-smooth}.
We also note that $N(\pi) = 1$; thus the inequality~\eqref{eq:weight diam vs norm} is strict.

We will now show that the dual-smoothness parameter cannot in general be improved.
To this end, we consider~$x=\diag(x_1,\dots,x_n)\in\PD(n)$ and~$y=I$; we also assume that~$\norm v=1$.
Then the Bregman divergence reduces to a KL-divergence
\begin{align*}
  D_{f_v}^P(x \Vert I)
= \log \braket{v, x v} - \braket{v, \log(x) v}
= \sum_{j=1}^n p_j \log \frac {p_j} {q_j}
\eqqcolon D_{\mathrm{KL}}(p \Vert q)
\end{align*}
between~$p_j \coloneqq \abs{v_j}^2$ and~$q_j \coloneqq p_j x_j / \sum_{j=1}^n p_j x_j$.
On the other hand, we can compute the left-hand side as in the proof of \cref{thm:kn-dual-smooth}, except that now $\Pi(X)=X$ for all~$X\in\Herm(n)$ and hence we obtain an equality:
if we set $w \coloneqq x^{1/2} v / \norm{x^{1/2} v}$, then
\begin{align*}
    \frac12 \norm{\diff f_v(I) - \tau^{\ast}_{I \to x}\diff f_v(x)}^2
&=   \frac12 \norm{ v v^\dagger - w w^\dagger }_{\mathrm F}^2
=   1 - \abs{\braket{v, w}}^2 \\
&=   1 - \frac{\braket{v, x^{1/2} v}^2}{\braket{v, x v}}
=   1 - \left( \sum_{j=1}^n \sqrt{p_j q_j} \right)^2
\eqqcolon 1 - F^2(p,q),
\end{align*}
where $F$ denotes the fidelity or Bhattacharyya coefficient.
Thus, we see that the optimal dual-smoothness parameter~$L_v$ of~$f_v$ can be lower bounded as
\begin{align*}
    L_v \geq \frac {1 - F^2(p,q)} {D_{\mathrm{KL}}(p \Vert q)},
\end{align*}
where~$q \neq p$ is any probability distribution with the same support as~$p$.
In particular, consider the vector~$v_\eps=(\sqrt{1-\eps^2},\eps,0,\dots,0)$, with associated probability distribution $p_\eps=(1-\eps^2,\eps^2,0,\dots,0)$, as well as $q_\eps=(1-\eps,\eps,0,\dots,0)$ for small~$\eps>0$.
Then:
\begin{align*}
    D_{\mathrm{KL}}(p_\eps \Vert q_\eps)
&= (1-\eps^2) \log(1+\eps) + \eps^2 \log \eps
= \eps + o(\eps),
\\
    1 - F^2(p_\eps,q_\eps)
&=  1 - \left( \sqrt{(1-\eps^2)(1-\eps)} + \eps^{3/2} \right)^2
= \eps + o(\eps),
\end{align*}
which shows that~$L_{v_\eps} \geq 1 - o(1)$.
\end{example}

\begin{remark}
The same reasoning as in \cref{ex:countersmooth} shows that the optimal dual-smoothness constant valid simultaneously for all Kempf--Ness functions~$f_v$, $v \in \VV \setminus \{0\}$, associated with a given representation~$\pi$ can be computed as follows:
\begin{align*}
    L_\pi \coloneqq \sup_{\norm v = 1, \, X \in \AA, D_{\mathrm{KL}}(p\Vert p^X)>0} \frac {\frac12 \norm{\mu(v) - \mu(e^{X/2} \cdot v)}_{\mathrm F}^2} {D_{\mathrm{KL}}(p\Vert p^X)},
\end{align*}
where $p_\omega = \norm{v_\omega}^2$ and $p_\omega^X = e^{\omega(X)} p_\omega / \sum_{\omega'} e^{\omega'(X)} p_{\omega'}$,
with $v=\sum_{\omega\in\Omega}v_\omega$ the weight-space decomposition relative to~$\AA$, as in \cref{foot:weights}.
We will not use this further in this paper.
\end{remark}

Finally, we bound the weight diameter and norm for (multi)homogeneous polynomial representations of products of general linear groups.
A representation is called \emph{polynomial} if each matrix entry of~$\pi(g)$ (the choice of basis is not relevant) is given by a polynomial function of the matrix entries of~$g \in G \subseteq \GL(n)$.
The following lemma improves the degree-based smoothness bounds obtained in~\cite{BFGOWW2019}.%
\footnote{We give an analytic proof, but \cref{lem:multihomogeneous} can be proved geometrically using the notion of weights.
For a multi-homogeneous representation of degrees~$m_1,\dots,m_d$, the weights can be identified with integer tuples~$\omega=(\omega_1,\dots,\omega_d)$ such that each $\omega_\ell$ is in the $m_\ell$-times dilated standard simplex; the latter has diameter at most~$m_\ell \sqrt{2}$ and its vertices have norm~$m_\ell$.}

\begin{lemma}\label{lem:multihomogeneous}
Let $\pi$ be a polynomial representation of $\GL(n_1) \times \cdots \times \GL(n_d)$ that is multi-homogeneous of degrees~$m_1,\dots,m_d$, i.e., $\pi(\lambda_1 g_1, \dots, \lambda_d g_d) = \lambda_1^{m_1} \cdots \lambda_d^{m_d} \pi(g_1,\dots,g_d)$ for all $g_\ell \in \GL(n_\ell)$,~$\lambda_\ell \in \CC^\times$.
Then the weight diameter and the weight norm can be bounded as
\begin{align*}
    D(\pi) \leq \sqrt{2 \sum_{\ell=1}^d m_\ell^2}
    \qquad\text{and}\qquad
    N(\pi) \leq \sqrt{\sum_{\ell=1}^d m_\ell^2}.
\end{align*}
In particular, the Kempf--Ness functions are $\frac12 (\sum_{\ell=1}^d m_\ell^2)$-smooth and $(\sum_{\ell=1}^d m_\ell^2)$-dual-smooth.
\end{lemma}
\begin{proof}
Since $\pi$ is a polynomial function, it has polynomial growth.
To see this, note that as a polynomial function, $\pi$ can be extended continuously to $\CC^{n_1\times n_1} \times \cdots \times \CC^{n_d\times n_d}$.
Then the multi-homogeneity implies that there exists a constant~$c_\pi$ such that, for all $g_1,\dots,g_d$,
\begin{align*}
    \norm{\pi(g_1,\dots,g_d)}_{\ope} \leq c_\pi \prod_{\ell=1}^d \norm{g_\ell}_{\ope}^{m_\ell}.
\end{align*}
Thus, for all~$X=(X_1,\dots,X_d)\in\bbH=\Herm(n_1) \times \dots \times \Herm(n_d)$ and all~$t>0$,
\begin{align*}
    e^{t\lambda_{\max}(\Pi(X))}
= \norm{e^{t \Pi(X)}}_{\ope}
\leq c_\pi \prod_{\ell=1}^d \norm{e^{tX_\ell}}_{\ope}^{m_\ell}
= c_\pi e^{t \sum_{\ell=1}^d m_\ell \lambda_{\max}(X_\ell)}.
\end{align*}
Letting $t\to\infty$, we see that $\lambda_{\max}(\Pi(X)) \leq \sum_{\ell=1}^d m_\ell \lambda_{\max}(X_\ell)$.
Applying the same argument to~$-X$ gives $\lambda_{\min}(\Pi(X)) \geq \sum_{\ell=1}^d m_\ell \lambda_{\min}(X_\ell)$.
Together:
\begin{align*}
    \lambda_{\max}(\Pi(X)) - \lambda_{\min}(\Pi(X))
&\leq \sum_{\ell=1}^d m_\ell \left( \lambda_{\max}(X_\ell) - \lambda_{\min}(X_\ell) \right) \\
&\leq \sqrt 2 \sum_{\ell=1}^d m_\ell \norm{X_\ell}_{\mathrm F}
\leq \sqrt 2 \sqrt{\sum_{\ell=1}^d m_\ell^2} \cdot \norm{X}_{\mathrm F},
\end{align*}
using the Cauchy--Schwarz inequality.
This proves that $D(\pi) \leq \sqrt 2 \sqrt{\sum_{\ell=1}^d m_\ell^2}$.
Then the smoothness and dual-smoothness bounds follow from \cref{prop:kn-smooth,thm:kn-dual-smooth}.
The claim for the weight norm follows in the same manner, using $\norm{X}_{\ope} = \max \{ \lambda_{\max}(X), -\lambda_{\min}(X) \}$.
\end{proof}

We note that our bound on the weight diameter is stronger than what is obtained from the bound on the weight norm and \cref{eq:weight diam vs norm}.
In the single-factor case~$d=1$, the weight norm bound was proved in \cite[Lem.~6.1]{BFGOWW2019}.

\subsection{Tensor action and entanglement polytopes}\label{sub:entanglementpolytope}
In this subsection, we specialize the preceding results to a group action that underlies several important applications:
the \emph{tensor action}.
Let $\Ten(n_0; n_1, \ldots, n_d) \coloneqq \CC^{n_0}\otimes\CC^{n_1}\otimes\cdots\otimes\CC^{n_d}$ denote the space of $(d+1)$-tensors of shape $n_0\times n_1\times\cdots\times n_d$.
Each element $T \in \Ten(n_0; n_1, \ldots, n_d)$ can also be interpreted as an $n_0$-tuple of tensors $T = (t_1, \ldots, t_{n_0})$, where each $t_i \in \CC^{n_1}\otimes\cdots\otimes\CC^{n_d}$ is a tensor of shape~$n_1\times\cdots\times n_d$ ($i \in [n_0]$).
Accordingly, we will call $T$ a \emph{tensor tuple} (unless $n_0 = 1$).
Recall that $G = \GL(n_1)\times\cdots\times \GL(n_d)$ is naturally a self-adjoint subgroup of~$\GL(n_1+\cdots+n_d)$, with maximal compact subgroup $K = \U(n_1)\times\cdots\times\U(n_d)$ (\cref{ex:gl}).

\begin{definition}[Tensor action]
The \emph{tensor action} of $\GL(n_1)\times\cdots\times \GL(n_d)$ is the representation~$\pi$ on $\Ten(n_0;n_1,\ldots,n_d)$ given by
\begin{equation}\label{eq:tensor action}
    (g_1, \ldots, g_d) \cdot T \coloneqq \pi(g_1,\dots,g_d)T \coloneqq (I\otimes g_1\otimes\cdots\otimes g_d)T.
\end{equation}
If we think of $T=(t_i)$ as a tuple, then $g \cdot T=(t'_i)$ is the tuple of tensors $t'_i = (g_1\otimes\cdots\otimes g_d)t_i$.
\end{definition}

The weight diameter of this representation is bounded by~$D(\pi) \leq \sqrt{2d}$ because it is multilinear (\cref{lem:multihomogeneous}).
Recall also that $\bbH = \Herm(n_1) \times \cdots \times \Herm(n_d)$.

The moment map corresponding to the tensor action computes the normalized \emph{(one-body) quantum marginals} of the tensor tuple~$T$ (which can be thought of as a multipartite quantum state):
\begin{equation}\label{eq:momentmap_tensor}
\begin{aligned}
    \mu \colon \Ten(n_0;n_1, \ldots, &n_d)\setminus\{0\} \to \PSD(n_1) \times\cdots\times\PSD(n_d), \\
    \mu(T) &= \left( \frac{T^{(1)}T^{(1)\dagger}}{\norm{T}^2}, \ldots, \frac{T^{(d)}T^{(d)\dagger}}{\norm{T}^2} \right),
\end{aligned}
\end{equation}
where $T^{(\ell)} \colon \bigotimes_{j\in \{0, 1, \ldots, d\}\setminus\{\ell\}}\CC^{n_j} \to \CC^{n_\ell}$ denotes the \emph{$\ell$-th principal flattening} of a tensor $T \in \CC^{n_0} \ot \CC^{n_1} \ot \cdots \ot \CC^{n_d}$.
Regarding~$T=(t_i)_{i \in [n_0]}$ as a tensor tuple, we have $T^{(\ell)}T^{(\ell)\dagger} = \sum_{i\in [n_0]}t_i^{(\ell)}t_i^{(\ell)\dagger}$.
To make the connection to quantum information theory explicit, we note that the moment map only depends on the positive-semidefinite \emph{density operator} $\rho \coloneqq \sum_{i=1}^{n_0} t_i t_i^\dagger / \norm T^2$ on~$\CC^{n_1} \ot \cdots \ot \CC^{n_d}$.
Then
\begin{equation}\label{eq:moment map vs partial trace}
    \mu(T)_\ell = \rho_\ell,
\quad\text{where}\quad \rho_\ell \coloneqq \tr_{\check\ell}\rho
\end{equation}
and the operation $\tr_{\check\ell}$ denotes the \emph{partial trace} over all but the $\ell$-th tensor factor; it is defined for arbitrary operators~$M \in \Lin(\CC^{n_1} \ot \cdots \ot \CC^{n_d})$ by
\begin{align}\label{eq:partial trace}
(\tr_{\check\ell}M)_{i,j} \coloneqq
\qquad\mathclap{\sum_{\substack{i_1\in[n_1],\dots,i_{\ell-1}\in[n_{\ell-1}],\\i_{\ell+1}\in[n_{\ell+1}],\dots,i_d\in[n_d]}}}\qquad
M_{(i_1,\dots,i_{\ell-1},i,i_{\ell+1},\dots,i_d),(i_1,\dots,i_{\ell-1},j,i_{\ell+1},\dots,i_d)}.
\end{align}
The density operator $\rho$ describes a quantum state of~$d$ particles, and $\rho_\ell$ describes the quantum state of the $\ell$-th particle ($\ell\in[d]$);
the normalized tensor~$T/\norm T$ is a so-called \emph{purification} of~$\rho$.

As in \cref{ex:gl}, we take the Cartan subspace~$\AA$ to be the tuples of real diagonal matrices, and the Weyl chamber~$C \subseteq \AA$ to be the tuples of diagonal matrices with nonincreasing entries.
Hence, we may naturally identify
$\AA \cong \RR^{n_1}\times\cdots\times\RR^{n_d}$ and
$C \cong \RR^{n_1}_{\downarrow}\times\cdots\times\RR^{n_d}_{\downarrow}$.
Then the moment polytope~$\Delta(T)$, which in this setting is called the \emph{entanglement polytope}~\cite{WDGC2013_entanglement} of~$T$, is
\begin{equation}\label{eq:entanglementpolytope}
    \Delta(T) = \overline{\left\{\left(\spec\frac{S^{(1)}S^{(1)\dagger}}{\|S\|^2}, \ldots, \spec\frac{S^{(d)}S^{(d)\dagger}}{\|S\|^2}\right)\ \middle|\ S = (I \ot g_1 \ot \cdots \ot g_d) T, \; g_\ell \in \GL(n_\ell) \right\}}.
\end{equation}
As a special case of~\eqref{eq:min_over_momentpolytope}, we now consider convex optimization on the entanglement polytope:
\begin{equation}\label{eq:min_over_entanglementpolytope}
    \minproblem{\phi(p)}{p \in \Delta(T)},
\end{equation}
for an objective~$\phi \colon \RR^{n_1}\times\cdots\times\RR^{n_d} \to \overline\RR$ satisfying \ref{it:phi1gl}--\ref{it:phi3gl} in \cref{ex:P phi}.
Then Hadamard mirror descent leads to the following iteration of tensor tuples $(T_k)_{k \in \NN}$ as a special case of \cref{eq:updatev_discrete,eq:updatev_discrete phi}:
\begin{align}\label{eq:updateTk}
    T_{k + 1}
\coloneqq c_{k+1} \left( I \ot \bigotimes_{\ell=1}^d e^{-\frac\eta2\nabla_\ell \Phi(\mu(T_k))} \right) T_k
= c_{k+1} \left( I \ot \bigotimes_{\ell=1}^d u_\ell^\dagger e^{-\frac\eta2\diag(\nabla_\ell \phi(p))} u_\ell \right) T_k,
\quad
    T_0 \coloneqq c_0 T,
\end{align}
where the $c_k\in\CC^\times$ are arbitrary nonzero scalars, and $\nabla_\ell\Phi$ denotes the $\ell$-th component of the gradient ($\ell \in [d]$);
in the second formula, the unitaries~$u_\ell \in \U(n_\ell)$ are chosen to diagonalize the quantum marginals, i.e., $u_\ell \mu(T_k)_\ell u_\ell^\dagger = \diag(p_\ell)$ for $p_\ell \in \RR^{n_\ell}$, and we denote $p=(p_1,\dots,p_d)$.

Using these formulas, we can now state concrete versions of \cref{thm:KN-general,cor:KN-Qsmooth} in this setting.
For conciseness, we only present the discrete-time results.

\begin{theorem}[Convergence for entanglement polytope optimization]\label{cor:entanglementpolytope}
    Let $0 \neq T \in \Ten(n_0; n_1, \ldots, n_d)$.
    Let $\phi \colon \RR^{n_1} \times \cdots \times \RR^{n_d} \to \overline\RR$ be a function that satisfies~\ref{it:phi1gl}--\ref{it:phi3gl} and $\spec\mu(g\cdot T) \in \intr\dom\phi$ for all~$g \in \GL(n_1)\times\cdots\times\GL(n_d)$.
    Let $p^\star$ be a minimizer of~$\phi$ in~$\Delta(T)$.
    Let $(T_k)_{k \in \NN}$ be the sequence of tensor tuples generated by \cref{eq:updateTk}.
    Under the step-size condition~\eqref{eq:step-size-KN}, it holds that, for all $k\geq1$,
    \[
        \phi(\spec \mu(T_k)) - \min_{p \in \Delta(T)} \phi(p) \leq 2 \frac{\log \norm T - \log\Capacity_{p^\star}(T)}{\eta k}.
    \]
    Moreover, if $\phi$ is $L$-smooth, then the step size can be taken as $\eta = 1/(d L)$.
\end{theorem}
\begin{proof}
This follows immediately from \cref{thm:KN-general,cor:KN-Qsmooth} and the formulas above, including the bound~$D(\pi) \leq \sqrt{2d}$ on the weight diameter from \cref{lem:multihomogeneous}.
\end{proof}

Importantly, for tensor tuples with Gaussian integer entries,%
\footnote{As an easy consequence, the bound on the objective gap also holds for tensors with Gaussian rational entries.}
 i.e., entries in~$\ZZ[\mathrm{i}]$, we are able to give a lower bound on the $K$-invariant logarithmic capacity that is polynomial in the bit-length of the tensor tuples and applies uniformly to any point in the entanglement polytope.

\begin{theorem}[Polynomial iteration bound for Gaussian integer tensors]\label{thm:tensor_f+b}
    Let $n_{\max} \coloneqq \max_{\ell = 1}^dn_\ell$.~Then,
    \[
        -\log\Capacity_p(T) \leq \poly(d, n_{\max})
    \]
    for all tensor tuples $0 \neq T \in \Ten(n_0; n_1, \ldots, n_d)$ with Gaussian integer entries and all $p \in \Delta(T)$.

    Accordingly, \cref{cor:entanglementpolytope} holds with the following convergence bound, where $\braket T$ denotes the bit-length of the tensor tuple~$T$:
    \[
        \phi(\spec \mu(T_k)) - \min_{p \in \Delta(T)} \phi(p) \leq \frac1{\eta k} \poly(d, n_{\max}, \braket T).
    \]
\end{theorem}

The proof is given at the end of this subsection.
It relies on a combination of results from the literature~\cite{BFGOWW2018,BCLNWZ2025computingmomentpolytopes}, as well as the following lemma.

\begin{lemma}\label{lem:busemann_valuebound}
    Let $g \in \GL(n)$ be such that each entry~$g_{ij}$ is an integer with absolute value at most~$S\geq1$.
    Then,
    $\dist(I, g^\dagger g) \leq 2 n^{3/2}\log (nS)$.
\end{lemma}
\begin{proof}
    We first assume that $n\geq2$.
    We have $\norm g_{\ope} \leq \norm g_{\mathrm F} \leq nS$.
    Hence, the singular values~$\sigma_1,\dots,\sigma_n$ of~$g$ are upper bounded by~$nS$.
    Since $\abs{\det g}\geq1$ ($g$ is an integer matrix) and equals the product of all~$n$ singular values, every singular value is also lower bounded by~$(nS)^{-(n-1)}$.
    Thus:
    \begin{align*}
        \dist(I, g^\dagger g)
    = \norm{\log(g^\dagger g)}_{\mathrm F}
    = 2 \sqrt{\sum_{i=1}^n (\log \sigma_i)^2}
    \leq 2 \sqrt n (n-1) \log(nS)
    \leq 2 n^{3/2} \log(nS),
    \end{align*}
    which proves the claim when~$n\geq2$.
    For $n=1$, we have
    $\dist(I, g^\dagger g)
    = \abs{\log (\abs{g_{11}}^2)}
    \leq 2 \log(S)$.
\end{proof}

We will also use the following bound on the $p$-capacity for points in the Cartan subspace.%
\footnote{The theorem in \cite{BFGOWW2018} is stated for integer tensor tuples and~$p$ in a specific Weyl chamber ($-C$ in our conventions, corresponding to reversing each component of~$p$), and the highest-weight argument given there a priori only applies to rational~$p$.
It applies equally to Gaussian integer tensor tuples because a nonzero Gaussian integer has absolute value at least one.
The case of general (irrational)~$p$ follows from the concavity of the logarithmic capacity on a Weyl chamber and the rationality of the Borel polytope; see, e.g., \cite[App.~A.1]{BFGOWW2018}.
Finally, the result extends to the entire Cartan subspace~$\AA$ by applying \cref{eq:log cap K equivariance} with tuples of permutation matrices, which preserve the integrality of the entries.}

\begin{proposition}[{\cite[Thm.~2.13]{BFGOWW2018}}]\label{prop:tensor_capacitylowerbound}
    Let $0 \neq T \in \Ten(n_0; n_1, \ldots, n_d)$ have Gaussian integer entries.
    Then, for all $p \in \AA = \diag(\RR^{n_1})\times\cdots\times\diag(\RR^{n_d})$, it holds that
    \[
        \log\capacity_p(T) > -\infty \iff \log\capacity_p(T) \ge -\sum_{\ell = 1}^d\log n_\ell.
    \]
\end{proposition}

\noindent
We need to bound $\log\Capacity_p(T)$, which is defined as the supremum of $\log\capacity_p(u \cdot T)$ over~$u \in K$.
However, $u \cdot T$ will not in general be integral and hence \cref{prop:tensor_capacitylowerbound} cannot be directly applied.
Instead, we will apply \cref{prop:tensor_capacitylowerbound} to $g \cdot T$, where $g$ is a tuple of integer matrices.
From \cref{prop:shiftingtrick_orig}, we know that there exists $u \in K$ such that $\log\capacity_p(u \cdot T) > -\infty$ if and only if there exists $g \in G$ such that $\log\capacity_p(g\cdot T)>-\infty$.
Prior work has shown that if one chooses~$g$ as a tuple of random integer matrices whose entries have polynomially-bounded bit-length, then $\log\capacity_p(g\cdot T) > -\infty$ follows with high probability; see~\cite[Cor.~2.7]{BFGOWW2018}, \cite[Thm.~3.6]{BCLNWZ2025computingmomentpolytopes}.
The latter reference states such a result uniformly for~$p \in \Delta(T)$.%
\footnote{Although the proof of \cite[Thm.~3.6]{BCLNWZ2025computingmomentpolytopes} is written for $d = 3$ and $n_0 = 1$, it readily generalizes.
We remark that while the proof in \cite{BCLNWZ2025computingmomentpolytopes} uses arbitrary matrices (as in \cref{prop:tensor_randomnessbound}), their theorem statement uses upper triangular random matrices; the statement holds in either case using the same Schwartz--Zippel argument.}

\begin{proposition}[{\cite[Thm.~3.6]{BCLNWZ2025computingmomentpolytopes}}]\label{prop:tensor_randomnessbound}
    There exists an integer $S\geq1$ such that~$\log S \leq \poly(d, n_{\max})$ with the following property:
    Let $0 \neq T \in \Ten(n_0; n_1, \ldots, n_d)$ be a tensor tuple.
    Let $g=(g_1, \ldots, g_d)$ be a random tuple of integer matrices, with each matrix entry drawn independently uniformly at random from~$[S]$.
    Then, with probability at least $1/2$, it holds that $g \in \GL(n_1) \times \cdots \times \GL(n_d)$ and
    \begin{align*}
        \Delta(T) = \{ p \in C \mid \log\capacity_p(g \cdot T) > -\infty \}.
    \end{align*}
\end{proposition}

We now prove the desired lower bound on the $K$-invariant logarithmic $p$-capacity.

\begin{proof}[Proof of \cref{thm:tensor_f+b}]
    Let $S\geq 1$, $\log S \leq \poly(d,n_{\max})$ be the constant from \cref{prop:tensor_randomnessbound}.
    The proposition ensures that for every tensor tuple~$T$ there exists a tuple $g = (g_1, \ldots, g_d)$ of invertible matrices in~$\GL(n_1)\times\cdots\times\GL(n_d)$ whose entries are positive integers at most~$S$ and such that~$\log\capacity_p(g\cdot T) > -\infty$ for all $p \in \Delta(T)$.
    Since $g\cdot T$ has Gaussian integer entries, using \cref{prop:tensor_capacitylowerbound} it follows that, in fact,
    \begin{equation}\label{eq:tensor bound 1}
        \log\capacity_p(g\cdot T) \geq -\sum_{\ell = 1}^d\log n_\ell \qquad (p \in \Delta(T)).
    \end{equation}
    We can obtain from this a lower bound on $\log\capacity_p(u\cdot T)$ for some~$u \in K=\U(n_1)\times\cdots\times\U(n_d)$, which in turn is a lower bound on $\log\Capacity_p(T)$.
    By \cref{lem:asymptoticrays}, the geodesic $t \mapsto g^\dagger e^{tp}g$ is asymptotic to $t \mapsto e^{tu^\dagger pu}$ for some $u \in K$.
    Then, $b^p(x) - b^{u^\dagger pu}(g^\dagger xg) \eqqcolon c$ is a constant, which we can upper bound as
    \begin{align}
    \nonumber
        c
    &= b^p(I) - b^{u^\dagger pu}(g^\dagger g)
    = - b^{u^\dagger pu}(g^\dagger g)
    = - \sum_{\ell=1}^d b^{u_\ell^\dagger p_\ell u_\ell}(g_\ell^\dagger g_\ell) \\
    \label{eq:tensor bound 2}
    &\leq \sum_{\ell=1}^d \norm{p_\ell}_{\mathrm F} \dist(I, g_\ell^\dagger g_\ell)
    \leq \sum_{\ell=1}^d 2 n_\ell^{3/2}\log (n_\ell S)
    \leq 2 d n_{\max}^{3/2}\log (n_{\max} S)
    \leq \poly(d,n_{\max}),
    \end{align}
    where we recall that we fixed the constant in the Busemann functions as~$b^X(I) = 0$; the first inequality follows from the fact that~$b^X$ is $\norm X_{\mathrm{F}}$-Lipschitz continuous, and the second inequality holds because $\norm{p_\ell}_{\mathrm F} \leq \norm{p_\ell}_{\tr} \leq 1$ for every~$\ell$ (from \cref{eq:entanglementpolytope}) and \cref{lem:busemann_valuebound}.
    Thus:
    \begin{align}\label{eq:tensor bound 3}
        2 \log\Capacity_p(T)
    \geq 2 \log\capacity_{u^\dagger p u}(T)
    = \inf_{x \in P} f_T^{u^\dagger p u}(x)
    = \inf_{x \in P} f_{g \cdot T}^p(x) - c,
    \end{align}
    and the right-hand side can be bounded using \cref{eq:tensor bound 1,eq:tensor bound 2}.
\end{proof}

We expect that the bound in \cref{thm:tensor_f+b} may not be tight.
This is because, in the proof above, we used a probabilistic claim merely to certify the existence of a suitably generic group element~$g \in G$.
For instance, in situations where one can choose $g$ to be a tuple of permutation matrices, we have~$c = 0$ in the proof and we obtain a much improved bound.

\subsection{Tensor power action and symmetric entanglement polytopes}\label{sub:tensor_power_action}
Another important action on tensor tuples with $n_1=\dots=n_d\eqqcolon n$ is the tensor power action of~$\GL(n)$.
It corresponds to the restriction of the tensor action to the subgroup $\GL(n) \subseteq \GL(n)^d \coloneqq \GL(n) \times \cdots \times \GL(n)$ under the diagonal embedding.
We annotate most objects associated to this action with a subscript ``$\Sym$'' (for ``symmetric''); however, the tensors do \emph{not} need to be symmetric.

\begin{definition}[Tensor power action]
The \emph{tensor power action} of $\GL(n)$ on $\Ten(n_0;\smash{\underbrace{n,\ldots,n}_{d \text{ times}}})$ is the representation~$\pi_{\Sym}$ given by
\begin{equation}\label{eq:tenpow}
    g \cdot T \coloneqq \pi_{\Sym}(g)T \coloneqq (I \ot g^{\ot d}) T = (I\otimes g\otimes\cdots\otimes g)T.
\end{equation}
If we think of $T=(t_i)$ as a tuple, then $g\cdot T=(t'_i)$ is the tuple of tensors $t'_i = g^{\ot d} t_i$.
\end{definition}

The weight diameter of this representation is bounded by~$D(\pi_{\Sym}) \leq d\sqrt{2}$ because it is homogeneous of degree~$d$ (\cref{lem:multihomogeneous}).

The associated moment map $\mu_{\Sym}$ computes the sum of the one-body quantum marginals, as follows from \cref{eq:momentmap_tensor}:
\begin{align}\label{eq:sym moment map}
    \mu_{\Sym} \colon \Ten(n_0; n,\dots,n) \setminus \{0\} \to \PSD(n),\quad
    \mu_{\Sym}(T) = \sum_{\ell = 1}^d \frac {T^{(\ell)}T^{(\ell)\dagger}}{\norm T^2}.
\end{align}
Thus, the associated moment polytope, called the \emph{symmetric entanglement polytope} of~$T$, is~given by
\begin{equation}\label{eq:sym moment polytope}
    \Delta_{\Sym}(T) = \overline{\left\{\spec\sum_{\ell = 1}^d \frac {S^{(\ell)}S^{(\ell)\dagger}} {\norm S^2} \ \middle|\ S = (I\otimes g\otimes\cdots\otimes g)T, \, g \in \GL(n)\right\}} \subseteq \RR^n_{\downarrow}.
\end{equation}
We consider the problem of convex optimization over the symmetric entanglement polytope:
\begin{equation}\label{eq:min_over_symmetricpolytope}
    \minproblem{\phi(p)}{p \in \Delta_{\Sym}(T)},
\end{equation}
for an objective~$\phi \colon \RR^n \to \overline\RR$ satisfying \ref{it:phi1gl}--\ref{it:phi3gl} in \cref{ex:P phi} (for $P = \PD(n)$).
Then Hadamard mirror descent amounts to the following update formula for tensor tuples $(T_k)_{k \in \NN}$, analogous to \cref{eq:updateTk}, as a special case of \cref{eq:updatev_discrete,eq:updatev_discrete phi}:
\begin{equation}\label{eq:updatesymmetric}
    T_{k+1}
\coloneqq c_{k+1}\left(I\otimes\left(e^{-\frac\eta2\nabla\Phi(\mu_{\Sym}(T_k))}\right)^{\otimes d}\right)T_k
= c_{k+1}\left(I\otimes\left(u^\dagger e^{-\frac\eta2\diag\nabla\phi(p)}u\right)^{\otimes d}\right)T_k,
    \quad T_0 \coloneqq c_0 T,
\end{equation}
where the $c_k \in \CC^\times$ are arbitrary nonzero scalars;
in the second formula, the unitary~$u \in \U(n)$ is chosen to diagonalize the moment map, i.e., $u \mu_{\Sym}(T_k) u^\dagger = \diag(p)$ for $p \in \RR^n$.

\bigskip

We now observe that the tensor power action case can be interpreted as a special case of the tensor action case; this will allow us to leverage the results of the preceding subsection.
For a tensor tuple $T = (t_i)_{i \in [n_0]} \in \Ten(n_0; n,\dots,n)$, consider the cyclic index shifts of each~$t_i$, defined as
\[
    \Shf_a \colon (\CC^n)^{\otimes d} \to (\CC^n)^{\otimes d}, \quad
    (\Shf_a t)_{j_1,\dots,j_d} \coloneqq t_{j_{[1+a]_d},\dots,j_{[d+a]_d}}
    \quad (a \in [d]),
\]
where $[\cdot]_d$ denotes the residue in~$[d]$ modulo~$d$, and stack them all into a single tuple indexed by pairs $(a,i) \in [d] \times [n_0] \cong [dn_0]$:

\begin{definition}\label{def:shift_stack}
We define the \emph{shift stack} of a tensor tuple~$T = (t_i)_{i\in[n_0]} \in \Ten(n_0; n,\dots,n)$ by
\begin{align*}
    T_{\Shf} \coloneqq (\Shf_a t_i)_{a\in[d], i \in [n_0]} \in \Ten(dn_0; n,\dots,n).
\end{align*}
\end{definition}

By an easy calculation, the image of $T_{\Shf}$ under the moment map~\eqref{eq:momentmap_tensor} for the tensor action of~$\GL(n)^d$ on $\Ten(dn_0; n,\dots,n)$ is given by:
\begin{equation}\label{eq:mu_rot}
    \mu(T_{\Shf}) = \frac1d(\smash{\underbrace{\mu_{\Sym}(T), \ldots, \mu_{\Sym}(T)}_d}).
\end{equation}
Furthermore, it is clear that
\begin{align}\label{eq:shift equivariance}
    (g,\ldots,g) \cdot T_{\Shf} = (g\cdot T)_{\Shf} \qquad (g \in \GL(n)).
\end{align}
Therefore, the symmetric entanglement polytope~$\Delta_{\Sym}(T)$ is ``included'' in the entanglement polytope~$\Delta(T_{\Shf})$ in the sense that
\begin{align}\label{eq:inclusion-symmetricpolytope}
    p \in \Delta_{\Sym}(T)
    \quad\Rightarrow\quad
    \frac1d(\smash{\underbrace{p,\dots,p}_d}) \in \Delta(T_{\Shf}).
\end{align}
Let us extend the objective function $\phi \colon \RR^n \to \overline{\RR}$ to $(\RR^n)^d$ by
\begin{equation}\label{eq:phisum}
    \phi_{\Sum} \colon (\RR^n)^d \to \overline{\RR},\quad\phi_{\Sum}(p_1, \ldots, p_d) \coloneqq \frac1d\sum_{\ell=1}^d \phi(dp_\ell);
\end{equation}
then $\phi_{\Sum}$ satisfies the conditions \ref{it:F1}--\ref{it:F3}.
Its unitarily invariant extension~$\Phi_{\Sum}$ is given by the analogous formula.
Moreover, $\phi$ is $L$-smooth if and only if $\phi_{\Sum}$ is $dL$-smooth.
By the inclusion~\eqref{eq:inclusion-symmetricpolytope}, it follows that
\[ \min_{p \in \Delta_{\Sym}(T)}\phi(p) \ge \min_{p \in \Delta(T_{\Shf})}\phi_{\Sum}(p). \]
In fact, the two values always coincide.
This follows from the observation that our algorithms for the two minimization problems coincide:

\begin{theorem}\label{thm:tensor power reduction}
    Let $0 \neq T \in \Ten(n_0; n, \ldots, n)$.
    Let $\phi \colon \RR^n \to \overline{\RR}$ be a function satisfying \ref{it:phi1gl}--\ref{it:phi3gl} and $\spec\mu_{\Sym}(g\cdot T) \in \intr\dom\phi$ for all $g \in \GL(n)$.
    Let $(T_k)_{k \in \NN}$ be the sequence of tensor tuples generated by~\eqref{eq:updatesymmetric}, and let~$((T_{\Shf})_k)_{k \in \NN}$ be the sequence of tensor tuples in $\Ten(d n_0; n, \ldots, n)$ generated by \eqref{eq:updateTk} for the objective~$\phi_{\Sum}$, with the same step size and scalars~$c_k$, starting at $(T_{\Shf})_0 \coloneqq c_0 T_{\Shf}$.
    \begin{enumerate}
    \item[(a)] Then $(T_{\Shf})_k = (T_k)_{\Shf}$ for all $k\in\NN$.
    \item[(b)] The following two minimization problems have the same value:
    \[
        \min_{p \in \Delta_{\Sym}(T)}\phi(p) = \min_{p \in \Delta(T_{\Shf})}\phi_{\Sum}(p).
    \]
    This equality holds even if there does not exist any step size satisfying the step-size condition.
    \item[(c)]
    $\phi(\spec \mu_{\Sym}(T_k)) \to \displaystyle\min_{p \in \Delta_{\Sym}(T)}\phi(p)$ if and only if
    $\phi_{\Sum}(\spec \mu((T_{\Shf})_k)) \to \displaystyle\min_{p \in \Delta(T_{\Shf})}\phi_{\Sum}(p)$.
    \item[(d)] The step-size conditions~\eqref{eq:step-size-KN} for $(T_k)_{k \in \NN}$ and $((T_{\Shf})_k)_{k \in \NN}$ are equivalent.
    \end{enumerate}
\end{theorem}
\noindent
As a consequence of~(b), minimization on the symmetric entanglement polytope reduces to the minimization over the entanglement polytope of the shift-stack tensor tuple.
\begin{proof}
(a)~This follows by induction.
The base case holds by definition.
Thus, assume that $(T_{\Shf})_k = (T_k)_{\Shf}$ for some~$k$.
Then:
\begin{align*}
 (T_{k+1})_{\Shf}
&= c_{k+1} \left(I\otimes\left(e^{-\frac\eta2\nabla\Phi(\mu_{\Sym}(T_k))}\right)^{\otimes d}\right) (T_k)_{\Shf} \\
&= c_{k+1} \left(I\otimes \bigotimes_{\ell=1}^d e^{-\frac\eta2\nabla_\ell\Phi_{\Sum}(\mu((T_{\Shf})_k))} \right) (T_{\Shf})_k = (T_{\Shf})_{k+1},
\end{align*}
using \cref{eq:updatesymmetric,eq:shift equivariance}, then the induction hypothesis, \cref{eq:mu_rot}, and the fact that $\nabla\Phi(d X) = \nabla_\ell\Phi_{\Sum}(X, \ldots, X)$ for $\ell \in [d]$ and~$X\in\Herm(n)$, and finally we use \cref{eq:updateTk}.

(b)~We first assume that~$\phi$ is $L$-smooth.
Then $\phi_{\Sum}$ is $dL$-smooth.
Recall also that $D(\pi_{\Sym}) \leq d\sqrt{2}$, while $D(\pi) \leq \sqrt{2d}$.
Thus, $\eta = 1/(L d^2)$ satisfies the step-size condition for both minimization problems, and our convergence result (\cref{cor:KN-Qsmooth}) implies that
\begin{align*}
    \lim_{k\to\infty} \phi(\spec\mu_{\Sym}(T_k)) = \min_{p \in \Delta_{\Sym}(T)}\phi(p)
    \quad\text{and}\quad
    \lim_{k\to\infty} \phi_{\Sum}(\spec\mu((T_{\Shf})_k)) = \min_{p \in \Delta(T_{\Shf})}\phi_{\Sum}(p).
\end{align*}
But $\phi_{\Sum}(\spec\mu((T_{\Shf})_k)) = \phi(\spec\mu_{\Sym}(T_k))$ for all~$k$ by part~(a) and \cref{eq:mu_rot,eq:phisum}.
We conclude that the two minima must be equal.

We now consider the case where $\phi$ is not smooth.
In this case, we may consider its smoothing, known as the \emph{Moreau envelope} of $\phi$ (e.g., \cite[\S{}1.G]{RW1998}): for $\lambda > 0$, it is defined as
\[
    E_\lambda\phi \colon \RR^n \to \RR,\quad E_\lambda\phi(p) \coloneqq \inf_{q \in \RR^n}\phi(q) + \frac1{2\lambda}\|p - q\|^2.
\]
Since $E_\lambda\phi$ is $1/\lambda$-smooth~\cite[Thm.~2.26]{RW1998}, we have the desired equality for $E_\lambda\phi$ in place of $\phi$.
Moreover, directly from the definition, it follows that $(E_\lambda\phi)_{\Sum} = E_{\lambda/d}(\phi_{\Sum})$.
Then, by taking the limit as $\lambda \to 0$, and using the fact that the minimum of the Moreau envelope on a compact set converges to the minimum of the original convex l.s.c.\ function,%
\footnote{This is proved as follows.
We show that $\lim_{\lambda \to 0}\min_{p \in C}E_\lambda\phi(p) = \min_{p \in C}\phi(p)$ for a compact set $C$.
From $E_\lambda\phi \le \phi$ it follows that $\lim_{\lambda \to 0}\min_{p \in C}E_\lambda\phi(p) \le \min_{p \in C}\phi(p)$.
Let $p^\star_\lambda \in C$ be a minimizer of $E_\lambda\phi$ and $p^\star \in C$ be an accumulation point as~$\lambda\to0$; for simplicity in notation, we assume that $p^\star_\lambda \to p^\star$.
By the definition of the Moreau envelope, there exists $q_\lambda \in \RR^n$ such that $\phi(q_\lambda) \le E_\lambda\phi(p^\star_\lambda) - \|q_\lambda - p^\star_\lambda\|^2/(2\lambda) + \varepsilon$ for any fixed $\varepsilon > 0$.
From this, we can conclude $\|q_\lambda - p^\star_\lambda\|^2 \to 0$.
Since $\phi$ is convex, it is lower bounded by some affine linear function $\phi(p) \ge c + \langle u, p\rangle$.
Then we have $\|q_\lambda - p^\star_\lambda\|^2 \le 2\lambda(E_\lambda\phi(p_\lambda^\star) + \varepsilon - c - \langle u, q_\lambda\rangle)$ and equivalently $\|q_\lambda + \lambda u- p^\star_\lambda\|^2 \le 2\lambda(E_\lambda\phi(p_\lambda^\star) + \varepsilon - c - \langle u, p^\star_\lambda\rangle) + \lambda^2\|u\|^2\to 0$;
$E_\lambda\phi(p_\lambda^\star)$ is uniformly bounded since it is at most~$\min_{p \in C} \phi(p)$, while the compactness of~$C$ uniformly bounds $\langle u,p_\lambda^\star\rangle$.
Hence, $q_\lambda \to p^\star$.
Since $\phi$ is l.s.c., $\phi(p^\star) \le \liminf_{\lambda \to 0}\phi(q_\lambda) \le \liminf_{\lambda \to 0}E_\lambda\phi(p^\star_\lambda) + \varepsilon$ for any~$\varepsilon > 0$.}
we have the desired result.

(c)~This follows immediately from~(a) and~(b).

(d)~The right-hand sides of~\eqref{eq:step-size-KN} coincide for the two sequences, since, for all $k \in \NN$, we have
\begin{align*}
\frac{\|(T_{\Shf})_{k+1}\|^2}{|c_{k+1}|^2\|(T_{\Shf})_k\|^2}
=  \frac{\|(T_{k+1})_{\Shf}\|^2}{|c_{k+1}|^2\|(T_k)_{\Shf}\|^2}
=  \frac{\|T_{k+1}\|^2}{|c_{k+1}|^2\|T_k\|^2},
\end{align*}
where the former equality follows from part~(a) and the latter equality holds because $\norm{S_{\Shf}}^2 = d\norm{S}^2$ for any~$S \in \Ten(n_0; n, \ldots, n)$ directly from \cref{def:shift_stack}.
The left-hand sides of~\eqref{eq:step-size-KN} also coincide for the two sequences, because we have $\Phi_{\Sum}(\mu((T_{\Shf})_k)) = \Phi(\mu_{\Sym}(T_k))$ by part~(a) and \cref{eq:mu_rot,eq:phisum}, and we have $\tr[\mu((T_{\Shf})_k)\nabla\Phi_{\Sum}(\mu((T_{\Shf})_k))] = \tr[\mu_{\Sym}(T_k)\nabla\Phi(\mu_{\Sym}(T_k))]$ by part~(a), \cref{eq:mu_rot}, and the fact that $\nabla\Phi(d X) = \nabla_\ell\Phi_{\Sum}(X, \ldots, X)$ for all $\ell \in [d]$ and~$X\in\Herm(n)$ used already in part~(a).
Thus the step-size condition~\eqref{eq:step-size-KN} coincides for the two sequences of tensor tuples.
\end{proof}

In particular, we obtain the following corollary of \cref{cor:entanglementpolytope,thm:tensor_f+b}.
We emphasize again that the following results do not require the tensors to be symmetric.

\begin{corollary}[Convergence for symmetric entanglement polytope optimization]\label{cor:symmetricentanglementpolytope}
    Let $0 \neq T \in \Ten(n_0; n, \ldots, n)$.
    Let $\phi \colon \RR^n \to \overline\RR$ be a function that satisfies~\ref{it:phi1gl}--\ref{it:phi3gl} and $\spec\mu_{\Sym}(g\cdot T) \in \intr\dom\phi$ for all~$g \in \GL(n)$.
    Let $p^\star$ be a minimizer of~$\phi$ in~$\Delta_{\Sym}(T)$.
    Let $(T_k)_{k \in \NN}$ be the sequence of tensor tuples generated by \cref{eq:updatesymmetric}.
    Under the step-size condition~\eqref{eq:step-size-KN}, it holds that, for all $k\geq1$,
    \[
        \phi(\spec \mu_{\Sym}(T_k)) - \min_{p \in \Delta_{\Sym}(T)} \phi(p) \leq 2 \frac{\log \norm T - \log\Capacity_{p^\star}(T)}{\eta k}.
    \]
    Here, the capacity refers to the tensor power action.
    If $\phi$ is $L$-smooth, we can take the step size~$\eta = 1/(d^2 L)$.

    If $T$ has Gaussian integer entries, then we can also bound
    \[
        \phi(\spec \mu_{\Sym}(T_k)) - \min_{p \in \Delta_{\Sym}(T)} \phi(p) \leq \frac1{\eta k} \poly(d, n, \braket T),
    \]
    where $\braket T$ denotes the bit-length of the tensor tuple~$T$.
\end{corollary}
\begin{proof}
Everything but the last statement follows from \cref{thm:KN-general,cor:KN-Qsmooth}.
For the last statement, first note that $\spec\mu_{\Sym}(g\cdot T) \in \intr\dom\phi$ implies $\spec\mu((g, \ldots, g)\cdot T_{\Shf}) \in \intr\dom\phi_{\Sum}$ for all $g \in \GL(n)$.
By \cref{thm:tensor power reduction}~(a), it follows that $\spec\mu((T_{\Shf})_k) \in \intr\dom\phi_{\Sum}$ for all $k \in \NN$.
Moreover, part~(d) shows that the sequence~$((T_{\Shf})_k)$ satisfies the step-size condition.
Consequently, the proof of \cref{cor:entanglementpolytope} applies to~$T_{\Shf}$ and~$\phi_{\Sum}$, using \cref{thm:main-discrete} with the weakened interior-domain hypothesis explained in \cref{rem:hypothesis}.
Since $T_{\Shf}$ has Gaussian integer entries and its bit-length is polynomially bounded in~$d$ and~$\braket T$, the capacity bound in \cref{thm:tensor_f+b} yields the desired polynomial bound for the shifted problem.
Finally, parts~(a) and~(b) of \cref{thm:tensor power reduction}, together with \cref{eq:mu_rot,eq:phisum}, identify the sequences, objectives, and minima; hence the polynomial bound also holds for the original problem involving~$T$ and~$\phi$.
\end{proof}

In the case of Gaussian integer entries, the capacity of the tensor power action~$\log\Capacity_{p^\star}(T)$ can also be bounded using the general \cref{thm:GLn_f+b} in the next subsection.
This yields a polynomial bound with an additional logarithmic dependence on~$n_0$.

\subsection{Polynomial actions of general linear groups}\label{sub:GLaction}
In this subsection, we consider a general homogeneous polynomial representation $\pi \colon \GL(n) \to \GL(\VV)$ and show a polynomial lower bound on the invariant logarithmic capacity for Gaussian integer vectors.
This yields a polynomial iteration bound analogous to \cref{thm:tensor_f+b} for the tensor action.
As always, $\VV$ comes with a $\U(n)$-invariant inner product and norm~$\norm{\cdot}$.
The notion of a Gaussian integer vector depends on a choice of basis, which identifies~$\VV\cong\CC^N$.
This basis need not be orthonormal with respect to the $\U(n)$-invariant inner product, because we also wish the group action to be given by polynomials with rational coefficients in this basis.
Thus we distinguish the original norm~$\norm{\cdot}$ and the $\ell^2$-norm~$\norm{\cdot}_2$ with respect to the chosen basis.
Following \cite{BFGOWW2019}, we consider the following parameters:%
\footnote{We follow the notation of \cite{BFGOWW2019}. In particular, in this subsection only, $K$ will denote the distortion factor and not the maximal compact subgroup $\U(n)$.}
\begin{itemize}
\item the \emph{degree}~$d\geq1$, i.e., $\pi(z g) = z^d \pi(g)$ for all $z\in\CC^\times$ and $g\in\GL(n)$.
\item the \emph{coefficient size}~$R$, defined as the least positive integer such that there exists a positive integer~$R' \leq R$ such that all the rescaled matrix entries~$R' \pi_{ij}(g)$ are polynomials in the matrix entries of~$g$, with integer coefficients bounded in absolute value by~$R$.
We assume that $\pi$ has rational coefficients, so that $R$ is well-defined.
\item the \emph{distortion factor} $K\geq1$, which is a number such that $K^{-1} \norm v \leq \norm v_2 \leq K \norm v$ for all $v \in \CC^N$.
\end{itemize}
An irreducible homogeneous polynomial representation of~$\GL(n)$ can be specified by a partition of~$d$, and there exists a well-known choice of basis, called the \emph{Gelfand--Tsetlin basis}; this extends to general homogeneous polynomial representations in a natural way by using lists of partitions and the corresponding union of bases.
If the representation~$\pi$ is specified in this manner and the input vector~$v$ has Gaussian integer entries with respect to this basis, then $\log(KRN)$ and the bound in \cref{thm:GLn_f+b} become polynomial in the input size; for details, see~\cite[\S{}1.6, \S{}8]{BFGOWW2019}.
We only state the capacity bound; as in \cref{thm:tensor_f+b}, the polynomial iteration bound is an immediate consequence.
We also note that the weight diameter can be bounded as $D(\pi) \leq d \sqrt 2$.

\begin{theorem}[Polynomial capacity bound for Gaussian integer vectors]\label{thm:GLn_f+b}
    Let $\pi \colon \GL(n) \to \GL(N)$ be a homogeneous polynomial representation with degree~$d$, coefficient size~$R$, and distortion factor~$K$.
    Then,
    \[
        -\log\Capacity_p(v) \leq \mathcal{O}(n^{5.5}d\log(nd) + \log(KRN)),
    \]
    for all vectors $0 \neq v \in \CC^N$ with Gaussian integer entries and all $p \in \Delta(v)$.
    If $\pi$ is specified by a list of partitions and the basis is the Gelfand--Tsetlin basis, then the right-hand side is $\poly(\braket\pi,\braket v,d)$, where $\braket{\cdot}$ denotes the bit-length.
\end{theorem}

The proof of \cref{thm:GLn_f+b} follows the same lines as \cref{thm:tensor_f+b}.
Instead of \cref{prop:tensor_capacitylowerbound,prop:tensor_randomnessbound} we use the following shifted capacity lower bound under randomization:

\begin{proposition}[{\cite[Thm.~7.24]{BFGOWW2019}}]\label{prop:GLn_capacitylowerbound}
Let $S \coloneqq 4n^{3n^3 + 1}d^{n^4}$.
Let $0 \neq v \in \CC^N$ have Gaussian integer entries, and let $p \in \Delta(v) \subseteq \RR^n_{\downarrow}$.
Let $g$ be a random $n\times n$ integer matrix with entries drawn independently uniformly at random from~$[S]$.
Then, with probability at least $1/2$, we have $g \in \GL(n)$ and
\begin{equation}\label{eq:capacity_proof1}
    -\log\capacity_p(g\cdot v) \leq \mathcal{O}(n^5 d \log(nd) + \log(KRN)).
\end{equation}
\end{proposition}

We now prove the desired bound on the unitarily invariant logarithmic $p$-capacity.

\begin{proof}[Proof of \cref{thm:GLn_f+b}]
    This proof proceeds analogously to the proof of \cref{thm:tensor_f+b}.
    \Cref{prop:GLn_capacitylowerbound} ensures the existence of $g \in \GL(n)$ whose entries are positive integers at most~$S$ and which satisfies \cref{eq:capacity_proof1}.
    By \cref{lem:asymptoticrays}, the geodesic $t \mapsto g^\dagger e^{tp}g$ is asymptotic to $t \mapsto e^{tu^\dagger pu}$ for some $u \in \U(n)$, and we can upper bound $b^p(x) - b^{u^\dagger pu}(g^\dagger xg) \eqqcolon c$, which is a constant independent of~$x\in\PD(n)$, as follows:
    \begin{align}
    \nonumber
        c
    &= b^p(I) - b^{u^\dagger pu}(g^\dagger g)
    = - b^{u^\dagger pu}(g^\dagger g) \\
    \label{eq:capacity_proof2}
    &\leq \norm p_{\mathrm F} \dist(I, g^\dagger g)
    \leq 2 d n^{3/2}\log (nS)
    \leq \mathcal{O}(n^{5.5}d\log(nd))
    \end{align}
    where the first inequality follows from the Lipschitz continuity of the Busemann function;
    the second inequality follows from~$\norm p_{\mathrm F} \leq N(\pi) \leq d$ and \cref{lem:busemann_valuebound};
    the last inequality follows from the definition of~$S$
    (except if $n=d=1$; here $\pi(z)=z I$, $p=1$, and the theorem still holds).
    Using \cref{eq:capacity_proof1,eq:capacity_proof2} instead of \cref{eq:tensor bound 1,eq:tensor bound 2}, we can conclude as in \cref{eq:tensor bound 3}.
\end{proof}

\section{Applications}\label{sec:applications}
In this section, we discuss various applications of Hadamard mirror descent.
In \cref{sub:norm minimization}, we address norm minimization over the moment polytope, and show an $\mathcal{O}(k^{-1})$ convergence rate; previously, such a result was only known in the semistable case (i.e., when the norm minimum is zero).
In \cref{sub:quantum}, we show that the quantum functionals of tensors admit a simple scaling algorithm called \emph{entropic tensor scaling}, and we establish polynomial iteration complexity for this algorithm.
No such algorithm was known before our work.
In \cref{sub:symmetric_quantum}, we obtain an analogous algorithm and result for the symmetric quantum functional.
In \cref{sub:G-stable-rk,sub:ncrk}, we address the computation of the $G$-stable ranks and the non-commutative rank; in these cases, the objective functions are not differentiable and we illustrate how to smooth them.

\subsection{Norm minimization over moment polytopes}\label{sub:norm minimization}
We start with a paradigmatic application that is deeply connected to scaling algorithms and to the non-commutative optimization framework~\cite{BFGOWW2019}.
We work in the general setting of \cref{sec:momentpolytope} and consider a representation $\pi \colon G \to \GL(\VV)$ of a self-adjoint subgroup~$G \subseteq \GL(n)$.
Our goal is to minimize the squared Frobenius norm over the moment polytope~$\Delta(v)$ of a given vector $v \in \VV\setminus\{0\}$.
Equivalently, we wish to minimize the squared Frobenius norm of the moment map over the $G$-orbit of~$v$.%
\footnote{Recall that we defined the moment polytope~\eqref{eq:momentpolytope} as a subset of the Weyl chamber~$C \subseteq \AA \subseteq \bbH \subseteq \Herm(n)$, and the moment map~\eqref{eq:moment map} as taking values in~$\bbH \subseteq \Herm(n)$.}
Thus, our goal is to compute:
\begin{align}\label{eq:norm min goal}
    \min_{p \in \Delta(v)} \tfrac12\norm{p}_{\mathrm{F}}^2 = \inf_{g \in G} \tfrac12\norm{\mu(g\cdot v)}_{\mathrm{F}}^2,
\end{align}
corresponding to the objective $\phi(\cdot) = \|\cdot\|_{\mathrm F}^2/2$ (or equivalently, $\Phi(\cdot) = \|\cdot\|_{\mathrm{F}}^2/2$).%
\footnote{The corresponding holonomy-invariant objective is $Q = \frac12\norm{\cdot}^2$, with $\norm{\cdot}$ the dual norm on the cotangent bundle induced by the Riemannian metric. See \cref{eq:norm squared}.}
This problem was studied from the perspective of so-called \emph{scaling problems} (such as matrix, operator and tensor scaling) by B\"{u}rgisser et al.~\cite{BFGOWW2019}; they pointed out that these can be reduced to the problem of setting the moment map to zero for suitable group actions, and proposed minimizing the associated Kempf--Ness functions (whose gradient computes the moment map, see \cref{lem:KNderivative}) as an approach to solving general scaling problems.
Hirai and Sakabe~\cite{HiraiSakabe2024} studied gradient descent in the unbounded case and proved a qualitative convergence result.
We first state our result and then compare it to the results of these prior works at the end of this subsection.

Since $\phi$ is $1$-smooth, the dual smoothness of the Kempf--Ness functions (\cref{thm:kn-dual-smooth}) guarantees that Hadamard mirror descent with the step size $\eta = 2/D(\pi)^2$ solves this problem.
Because $\nabla\Phi$ is the identity map, Hadamard mirror descent reduces to ordinary Riemannian gradient descent.
In particular, the corresponding iteration~\eqref{eq:updatev_discrete} reads as follows:
\begin{equation}\label{eq:moment-map-iter}
    v_{k+1} \coloneqq c_{k+1}\pi\mleft(e^{-\frac1{D(\pi)^2}\mu(v_k)}\mright)v_k,\qquad v_0 \coloneqq c_0v,
\end{equation}
where $0 \neq c_k \in \CC$ are arbitrary scalars.
Then, by applying the general \cref{cor:KN-Qsmooth}, we have the following convergence result for the norm minimization:

\begin{theorem}[Convergence for norm minimization]\label{thm:moment-norm-conv}
    Let $0 \neq v \in \VV$, and let $(v_k)_{k \in \NN}$ be the sequence generated by \cref{eq:moment-map-iter}.
    Then, for every $k\geq1$,
    \[
        \tfrac12\|\mu(v_k)\|_{\mathrm{F}}^2 - \min_{p \in \Delta(v)}\tfrac12\|p\|_{\mathrm{F}}^2
    \leq \frac{D(\pi)^2 \left( \log\|v\| - \log\Capacity_{p^\star}(v) \right)}{k},
    \]
    where $p^\star$ is the minimum-norm point in $\Delta(v)$.

    Moreover, in the case of the tensor action, the tensor power action, or an arbitrary homogeneous polynomial action of~$\GL(n)$ (using a suitable input format), for vectors $v$ with Gaussian-integer entries the numerator can be replaced by a bound that is polynomial in the input bit-length, as in \cref{thm:tensor_f+b,cor:symmetricentanglementpolytope,thm:GLn_f+b}.
\end{theorem}
\begin{proof}
    Since $\phi(\cdot) \coloneqq \|\cdot\|_{\mathrm{F}}^2/2$ is a $1$-smooth convex Weyl group invariant function, the displayed bound follows from \cref{cor:KN-Qsmooth} and the fact that $\|s(\mu(v_k))\|_{\mathrm{F}} = \|\mu(v_k)\|_{\mathrm{F}}$ for all~$k$.
    The ``moreover'' part follows directly from \cref{thm:tensor_f+b,cor:symmetricentanglementpolytope,thm:GLn_f+b}.
\end{proof}

It is instructive to compare this result with prior work.
To minimize Kempf--Ness functions~$f_v$, B\"{u}rgisser et al.~\cite{BFGOWW2019} proposed and analyzed the gradient descent algorithm~\eqref{eq:moment-map-iter} but with step size~$\eta=1/L$, where $L$ is any upper bound on the smoothness parameter of~$f_v$.
They further showed that the Kempf--Ness functions are $N(\pi)^2$-smooth, leading to the step size~$\eta=1/N(\pi)^2$.
This can be improved to $\eta = 4/D(\pi)^2$ using our result on the optimal smoothness parameter (see \cref{prop:kn-smooth,eq:weight diam vs norm}).
Under either choice, their results show that the squared norm of the moment map~$\mu(v_k)$ converges to $0$ at $\mathcal{O}(k^{-1})$ whenever $0 \in \Delta(v)$, which is exactly when $f_v$ is bounded from below (such~$v$ are called \emph{semistable}).%
\footnote{Strictly speaking they proved an $\mathcal{O}(k^{-1})$ bound for the best squared norm within $k$ iterates.
Hirai and Sakabe proved subsequently that the gradient norms along gradient descent with step size~$1/L$ are nonincreasing~\cite[Lem.~3.8~(2)]{HiraiSakabe2024}; this upgrades the best-iterate bound to the last iterate.}
For the same step sizes, Hirai and Sakabe~\cite{HiraiSakabe2024} showed that $\norm{\mu(v_k)}_{\mathrm F}$ converges to its infimum even if the infimum is positive, but did not provide a convergence rate.
Our \cref{thm:moment-norm-conv} proves a convergence rate of~$\mathcal{O}(k^{-1})$ that applies in general, regardless of the value of the optimization problem, provided we choose the step size~$\eta=2/D(\pi)^2$ arising from the use of the dual-smoothness parameter.

\begin{remark}\label{rem:Luna}
In the setting of \cref{thm:moment-norm-conv}, suppose that $v$ is fixed by a self-adjoint Zariski-closed subgroup~$H \subseteq G$, that is, $h \cdot v = v$ for every~$h\in H$.
Consider the \emph{centralizer} $C_G(H)$ of $H$ in $G$, defined by~$C_G(H) \coloneqq \{g\in G\mid ghg^{-1}=h\, (\forall h\in H)\}$.
Luna's criterion \cite[Cor.~2 and Rem.~1]{Luna-75} states that
\begin{equation}\label{eq:luna}
    v \text{ is $G$-semistable}\quad\iff\quad v\text{ is $C_G(H)$-semistable};
\end{equation}
recall that semistability means that the orbit closure does not contain the zero vector (i.e., the Kempf--Ness function is bounded from below); equivalently, the minimum norm over the moment polytope is zero.
For convenience, we assume that $C_G(H)$ is connected.%
\footnote{If the centralizer is not connected, replace it by its identity component; this subgroup defines the same null cone.}

Hadamard mirror descent gives an alternative \emph{algorithmic} proof of this equivalence using conservation of symmetry (\cref{prop:symmetry}).
To see the connection, we choose $c_k \coloneqq 1$ for all $k\in\NN$ in \cref{eq:updatev_discrete}, so~$v_{k} = g_k\cdot v$ remains in the $G$-orbit at all times.
Here, $g_k \in G$ denotes the corresponding sequence of group elements, defined as in \cref{eq:update group} for our objective $\Phi(X)=\norm{X}_{\mathrm{F}}^2/2$ and step size $\eta = 2/D(\pi)^2$.
We claim that the group elements in fact belong to~$C_G(H)$.
To see this, set $L\coloneqq H\cap K$.
Since $L$ is Zariski dense in $H$, we have $C_G(L) = C_G(H)$.
By \cref{prop:symmetry}~(c), we have $u\cdot v_k=v_k$ for every $u\in L$.
Therefore, using the $K$-equivariance of the moment map, we have
\[
    u \mu(v_k) u^{\dag} = \mu(u\cdot v_k) = \mu(v_k) \qquad (u\in L).
\]
Thus, $\mu(v_k)\in\Lie(C_G(H))=\{X\in\Lie(G)\mid u X u^\dag=X\; (u\in L)\}$, hence $e^{-\frac1{D(\pi)^2}\mu(v_k)} \in C_G(H)$.
Since $g_0=I$ and $g_{k+1}=e^{-D(\pi)^{-2}\mu(v_k)} g_k$, induction gives $g_k\in C_G(H)$ for every $k\in\NN$.
In view of \cref{thm:moment-norm-conv}, this already implies Luna's equivalence~\eqref{eq:luna} because norm minimization over the moment polytope is equivalent to minimizing the norm of the gradient of the Kempf--Ness function.

To see this more explicitly, let $\mu_C$ denote the moment map for the restricted $C_G(H)$-action on $\VV$.
It is given by the composition of~$\mu$ with the orthogonal projection onto $\Lie(C_G(H))\cap \Herm(n)$.
Note that every $w\in C_G(H)\cdot v$ is fixed by $H$, and hence also by $L$.
The same equivariance argument therefore shows that $\mu(w) \in \Lie(C_G(H))\cap \Herm(n)$, hence $\mu_C(w)=\mu(w)$.
Finally, by \cref{thm:moment-norm-conv}, we know that~$\lim_{k\rightarrow\infty} \norm{\mu(g_k\cdot v)}_{\mathrm{F}}=\inf_{g\in G}\norm{\mu(g\cdot v)}_{\mathrm{F}}$, so
\[
    \inf_{g\in G}\norm{\mu(g\cdot v)}_{\mathrm{F}}
= \inf_{g\in C_G(H)}\norm{\mu(g\cdot v)}_{\mathrm{F}}
= \inf_{g\in C_G(H)}\norm{\mu_C(g\cdot v)}_{\mathrm{F}}.
\]
The left-hand infimum vanishes exactly when $v$ is $G$-semistable, while the right-hand side infimum vanishes exactly when $v$ is $C_G(H)$-semistable.
Thus we obtain Luna's criterion~\eqref{eq:luna}, concluding the proof.
\end{remark}

\subsection{Quantum functionals and entropic tensor scaling}\label{sub:quantum}
In this subsection, we derive an algorithm that computes the \emph{quantum functionals}, an important class of tensor parameters \cite{CVZ2018}.
The setting is the tensor action in \cref{sub:entanglementpolytope}: the group $G = \GL(n_1)\times\cdots\times\GL(n_d)$ acts on $\VV = \Ten(n_0; n_1, \ldots, n_d)$ as in \cref{eq:tensor action}.
Recall that the associated moment map computes the one-body quantum marginals (\cref{eq:momentmap_tensor})
\[
    \mu \colon \Ten(n_0;n_1, \ldots, n_d)\setminus\{0\} \to \PSD(n_1) \times\cdots\times\PSD(n_d),\quad
    \mu(T) = \left( \frac {T^{(1)}T^{(1)\dagger}} {\norm T^2}, \dots, \frac {T^{(d)}T^{(d)\dagger}} {\norm T^2} \right),
\]
and the moment polytope is known as the \emph{entanglement polytope}~\cite{WDGC2013_entanglement} and is given by \cref{eq:entanglementpolytope}:
\begin{align*}
    \Delta(T) = \overline{\left\{\left(\spec\frac{S^{(1)}S^{(1)\dagger}}{\|S\|^2}, \ldots, \spec\frac{S^{(d)}S^{(d)\dagger}}{\|S\|^2}\right)\ \middle|\ S = (I \ot g_1 \ot \cdots \ot g_d) T, \; g_\ell \in \GL(n_\ell) \right\}}.
\end{align*}
Note that each quantum marginal~$\frac{S^{(\ell)}S^{(\ell)\dagger}}{\|S\|^2}$ is positive semidefinite and has trace one; thus the eigenvalues form a probability distribution and we can consider the entropy.
The logarithmic quantum functionals are defined by maximizing a weighted sum of the marginal entropies over the entanglement polytope.
To define them formally, fix a probability distribution~$\theta$ on~$[d]$ and consider the convex objective
\begin{align}\label{eq:qfunc objective}
    \phi_\theta \colon \RR^{n_1} \times \cdots \times \RR^{n_d} \to \overline\RR, \quad
    \phi_\theta(p_1,\dots,p_d) \coloneqq \begin{cases}-\sum_{\ell \in [d]} \theta_\ell H(p_\ell) & \text{if all } p_\ell \ge \mathbf{0}, \\
    \infty & \text{otherwise},
    \end{cases}
\end{align}
where $H(p)\coloneqq -\sum_ip_i \log p_i$ is the Shannon entropy, with the usual convention that $0\log0 \coloneqq 0$.
It is easily checked that $\phi_\theta$ satisfies \ref{it:phi1gl}--\ref{it:phi3gl}.
We denote by $\Phi_\theta \coloneqq \phi_\theta \circ \spec$ its unitarily invariant extension, obtained by replacing the Shannon entropy by the von Neumann entropy, $H(X) \coloneqq -\tr[X \log X]$; see \cref{ex:F examples}~(a).
The quantum functionals are defined~as~follows.%
\footnote{\label{footnote:qfunc tensor tuple}The quantum functionals were originally proposed for a \emph{single} tensor (i.e., $n_0 = 1$), but the definition generalizes naturally to tensor tuples ($n_0 \ge 2$).
From the perspective of quantum information, this means that we allow $T$ to describe a mixed quantum state of $d$~particles rather than restricting to pure states.
However, identifying $\Ten(n_0; n_1, \ldots, n_d) \cong \Ten(1; n_0, n_1, \ldots, n_d)$, one can show that the quantum functionals of a $d$-tensor tuple~$T$ coincide with the quantum functionals of $T$ regarded as a tensor in~$\Ten(1; n_0, n_1, \ldots, n_d)$ if we extend~$\theta$ by~$\theta_0 \coloneqq 0$.
Hence, the generalization we consider, while natural, does not introduce an essentially new concept.}

\begin{definition}[Quantum functionals, \cite{CVZ2018}]\label{def:qfunc}
For a probability distribution~$\theta$ on~$[d]$, the \emph{logarithmic quantum functional}~$E_\theta$ is, for all tensor tuples~$T \in \Ten(n_0; n_1,\dots,n_d)\setminus\{0\}$, given by
\begin{align*}
    \logqfunc_\theta(T)
\coloneqq \max_{p\in\Delta(T)} \smash{\sum_{\ell=1}^d} \theta_\ell H(p_\ell)
= -\min_{p \in \Delta(T)}\phi_\theta(p).
\end{align*}
Its exponential is called the \emph{quantum functional} $\qfunc_\theta(T)
\coloneqq e^{\logqfunc_\theta(T)}$;
we also set 
$\qfunc_\theta(0) \coloneqq 0$ for the~zero~tensor.
\end{definition}

\noindent
When $T$ is a single matrix ($d=2$, $n_0=1$), the quantum functional coincides with the matrix rank for all $\theta$.
We discuss the motivation for the general definition in \cref{sec:intro}; see also \cref{thm:algorithmicproof} below and the discussion surrounding it.

We will now apply our general results in \cref{sub:entanglementpolytope} to the minimization above.
We may assume that the tensor tuple~$T$ is \emph{concise}: this means that each component of the moment map~$\mu(T)$ is a nonsingular and hence positive definite matrix.%
\footnote{\label{footnote:concise}If $T$ is not concise, its slices in some direction are linearly dependent.
By an invertible transformation in the corresponding tensor factor, we may arrange that the \emph{nonzero} slices are linearly independent.
Deleting the zero slices then allows us to view~$T$ in a smaller tensor subspace.
One obtains a concise tensor~$T'$ by repeating this procedure for every rank-deficient tensor factor.
It is not hard to see that the entanglement polytopes~$\Delta(T)$ and~$\Delta(T')$ coincide (up to padding by zeros), and hence the quantum functionals for the two tensors coincide.
For tensors with Gaussian integer entries, this reduction can be carried out in polynomial time, and the resulting tensor can be arranged to have Gaussian integer entries with bit-lengths polynomially bounded in the input length.}
In this case, $g \cdot T$ is also concise for all~$g \in G$; hence the assumption $\spec\mu(g \cdot T) \in \intr \dom \phi_\theta$ is satisfied for all~$g \in G$.
Thus we can apply Hadamard mirror descent.
Taking~$c_{k+1} \coloneqq e^{\eta/2}$ in \cref{eq:updateTk} and using the formula $-\nabla H(X) = \log X + I$ for the gradient of the von Neumann entropy, we obtain the following iteration:
\begin{align}\label{eq:entropic tensor scaling}
    T_{k+1}
\coloneqq \mleft(I \otimes \bigotimes_{\ell=1}^d \rho_\ell^{-\eta \theta_\ell/2}\mright)T_k,
    \quad\text{where}\quad (\rho_1,\dots,\rho_d) = \mu(T_k),\qquad T_0 \coloneqq T.
\end{align}
For~$\eta=1$, we refer to this algorithm as \emph{entropic tensor scaling}.

We now show that we can always take this step size.
One immediate challenge is that the entropy is not smooth, and hence we cannot apply \cref{cor:KN-Qsmooth}.
Furthermore, numerical experiments suggest that $1$-relative smoothness does not hold either, which would suffice for~$\eta=1$.
Accordingly we establish the step-size condition~\eqref{eq:step-size-KN} directly.
To this end, we first observe that the step-size condition in the present setting amounts to a quantum entropy inequality.

\begin{lemma}\label{lem:quantum reduction}
    To establish the step-size condition~\eqref{eq:step-size-KN} for the sequence~\eqref{eq:entropic tensor scaling}, it suffices to prove that
    for all $\rho \in \PSD(\CC^{n_1} \ot \cdots \ot \CC^{n_d})$ of trace one such that the~$\rho_\ell$ are all invertible,
    \begin{equation}\label{eq:entropic goal}
        \eta \sum_{\ell=1}^d \theta_\ell H\mleft(\frac {\sigma_\ell}{\tr\sigma}\mright) \geq \log \tr \sigma
        \quad\text{where } \sigma \coloneqq \left( \ot_{\ell=1}^d \rho_\ell^{-\eta\theta_\ell/2} \right) \rho \left( \ot_{\ell=1}^d \rho_\ell^{-\eta\theta_\ell/2} \right).
    \end{equation}
    Here, $\rho_\ell \coloneqq \tr_{\check\ell}\rho$ and $\sigma_\ell \coloneqq \tr_{\check\ell} \sigma$ denote the partial traces as defined in \cref{eq:partial trace}.
    Furthermore, in order to establish \cref{eq:entropic goal} for all $\eta\in(0,1]$, all~$d$, all dimensions~$n_1,\dots,n_d$ and all probability distributions~$\theta$, it suffices to establish it in the case~$\eta=1$, in which case \cref{eq:entropic goal} is equivalent to $\sum_{\ell=1}^d \theta_\ell H(\sigma_\ell) \geq 0$.%
\footnote{We emphasize that this inequality is nontrivial: while the~$\sigma_\ell$ are always PSD, their traces will, in general, not be equal to one, so we can have $H(\sigma_\ell) < 0$ for some~$\ell$.
We remark that $\tr\sigma_\ell = \tr\sigma \geq 1$, because $\bigotimes_{\ell\in[d]}\rho_\ell^{-\theta_\ell/2}$ has~all eigenvalues at least $1$ and multiplying it from left and right does not decrease the trace of a positive semidefinite matrix.}
\end{lemma}
\begin{proof}
The left-hand side of the step-size condition~\eqref{eq:step-size-KN} simplifies to
\begin{align*}
    -\sum_{\ell=1}^d \theta_\ell H(\mu(T_k)_\ell) + \sum_{\ell=1}^d \theta_\ell H(\mu(T_{k+1})_\ell) + \sum_{\ell=1}^d \theta_\ell \tr\mleft[ \mu(T_k)_\ell \nabla H(\mu(T_k)_\ell) \mright]
= \left( \sum_{\ell=1}^d \theta_\ell H(\mu(T_{k+1})_\ell) \right) - 1,
\end{align*}
where we used $-\nabla H(X) = \log X + I$ and $\tr \mu(T_k)_\ell = 1$ for all $\ell\in[d]$.
The right-hand side is given by
\begin{align*}
\frac1\eta \log\frac{\norm{T_{k+1}}^2}{\abs{c_{k+1}}^2\norm{T_k}^2}
= \frac1\eta \log\frac{\norm{T_{k+1}}^2}{\norm{T_k}^2} - 1,
\end{align*}
where we used our choice $c_{k+1} \coloneqq e^{\eta/2}$ leading to \cref{eq:entropic tensor scaling}.
Thus, the step-size condition is equivalent to
\begin{align}\label{eq:step size simplified}
    \eta \sum_{\ell=1}^d \theta_\ell H(\mu(T_{k+1})_\ell) \geq \log \frac{\norm{T_{k+1}}^2}{\norm{T_k}^2}.
\end{align}
Because both the left-hand side and the right-hand side are invariant under rescaling~$T_k$, we may assume without loss of generality that~$\norm{T_k} = 1$.
Then \cref{eq:step size simplified} becomes \cref{eq:entropic goal} if we set $\rho \coloneqq \sum_{i\in [n_0]}t_{k,i}t_{k,i}^\dagger$ and $\sigma \coloneqq \sum_{i\in [n_0]}t_{k+1,i}t_{k+1,i}^\dagger$, where $T_k = (t_{k,i})_{i\in[n_0]}$, and use \cref{eq:moment map vs partial trace}.

For the last claim, suppose that $\eta<1$.
Then we can add a one-dimensional tensor factor ($n_{d+1}=1$), use the weights~$\tilde\theta_\ell \coloneqq \eta\theta_\ell$ for $\ell\in[d]$ and $\tilde\theta_{d+1} \coloneqq 1-\eta$, and consider $\tilde\rho \coloneqq \rho \ot e_1 e_1^\dagger$.
Then $\tilde\sigma = \sigma \ot e_1 e_1^\dagger$, and one can verify that the $\eta=1$ result for these data is equivalent to \cref{eq:entropic goal}.
\end{proof}

We now establish this entropy inequality.
To this end, we will use a key notion from quantum information theory~\cite{watrous}:
\emph{quantum relative entropy}, which is defined as%
\footnote{We briefly comment on this definition, since the operator logarithm is only well-defined for positive definite operators.
Restricting $A$ to its image yields a positive definite operator, so $\log A$ is well-defined; the trace of $A\log A$ is understood to be taken on~$\operatorname{im} A$, as in the definition of the von Neumann entropy.
Similarly, the assumption $\operatorname{im} A \subseteq \operatorname{im} B$ allows us to restrict to the image of~$B$, where $B$ is positive definite, and to take the trace of~$A\log B$ on this subspace.
This is analogous to the definition of the KL divergence by the formula $D_{\text{KL}}(p\Vert q) \coloneqq \sum_j p_j \log(p_j/q_j)$.}
\[
D(A \Vert B) \coloneqq \begin{cases}
    \tr[A\log A] - \tr[A\log B] & \text{if } A, B \succeq 0 \text{ and} \operatorname{im} A \subseteq \operatorname{im} B, \\
    \infty & \text{otherwise}.
\end{cases}
\]
It satisfies the following monotonicity property, called the \emph{data-processing inequality}~\cite[Thm.~5.35]{watrous}:
\begin{equation}\label{eq:dpi}
    D(A \Vert B) \geq D(\Lambda(A) \Vert \Lambda(B))
    \qquad (A, B \succeq 0),
\end{equation}
for any completely positive trace-preserving map $\Lambda$, such as the partial trace defined in \cref{eq:partial trace}.

\begin{theorem}\label{thm:quantum}
\Cref{eq:entropic goal} holds for all $\eta\in(0,1]$.
As a consequence, the step-size condition~\eqref{eq:step-size-KN} holds for all $\eta\in(0,1]$.
\end{theorem}
\begin{proof}
By \cref{lem:quantum reduction}, to establish the theorem it suffices to prove $\sum_{\ell=1}^d \theta_\ell H(\sigma_\ell) \geq 0$ for all~$\sigma$ as defined in \cref{eq:entropic goal} for~$\eta=1$.
We will prove that
\begin{equation}\label{eq:entropic subgoal}
    H(\sigma_\ell) \geq \smash{\sum_{k \neq \ell}} \theta_k \tr\mleft[ \sigma_k \log \rho_k \mright] - (1-\theta_\ell) \tr[\sigma_\ell \log \rho_\ell] \qquad (\ell\in[d]),
\end{equation}
which implies \cref{eq:entropic goal} by multiplying the $\ell$-th inequality by~$\theta_\ell$, summing over~$\ell$, and using $\sum_{\ell=1}^d\theta_\ell=1$.
By symmetry, it suffices to prove \cref{eq:entropic subgoal} for~$\ell=1$, which will simplify the notation.
To this end, define
\begin{align}\label{eq:def omega}
    \omega
\coloneqq ( \rho_1^{-\theta_1/2} \ot I^{\ot(d-1)} ) \rho ( \rho_1^{-\theta_1/2} \ot I^{\ot(d-1)} )
= (I \ot Z) \sigma (I \ot Z),
\end{align}
where we have introduced the positive definite operator $Z \coloneqq \otimes_{\ell=2}^d \rho_\ell^{\theta_\ell/2}$.
Let $Z = \sum_\lambda \lambda P_\lambda$ denote the spectral decomposition (i.e., $\lambda$ ranges over the eigenvalues of~$Z$ and each~$P_\lambda$ denotes the orthogonal projection onto the corresponding eigenspace; note that $\sum_\lambda P_\lambda=I$).
Define
\[
    \tilde\sigma \coloneqq \sum_\lambda (I \ot P_\lambda) \sigma (I \ot P_\lambda),
\qquad \tilde\omega \coloneqq \sum_\lambda (I \ot P_\lambda) \omega (I \ot P_\lambda)
= (I \ot Z) \tilde\sigma (I \ot Z),
\]
where we used the second formula in \cref{eq:def omega} and that $Z$ commutes with its eigenprojections~$P_\lambda$.
Note that $\tilde\sigma$ and $\tilde\omega$ have block-diagonal forms with respect to the eigenspaces of $I\otimes Z$, and hence $I \ot Z$ commutes with $\tilde\sigma$ and $\tilde\omega$; therefore,
$\log \tilde\omega = \log \tilde\sigma + I \ot \log(Z^2)$ on the images of~$\tilde\omega$ and~$\tilde\sigma$ (which coincide).
Thus:
\begin{align}\label{eq:DPI left}
    D(\tilde\sigma \Vert \tilde\omega)
= -\tr\mleft[ \tilde\sigma (I \ot \log(Z^2)) \mright]
= -\tr\mleft[ \sigma (I \ot \log(Z^2)) \mright]
= -\sum_{k=2}^d \theta_k \tr\mleft[ \sigma_k \log \rho_k \mright],
\end{align}
where the second equality follows by inserting the definition of~$\tilde\sigma$, commuting~$P_\lambda$ through~$\log Z$, and using the fact that $\sum_\lambda P_\lambda^2 = \sum_\lambda P_\lambda = I$ because the $P_\lambda$ are a complete set of projections;
the last equality holds by $\log(Z^2) = \sum_{k=2}^d \theta_k I^{\ot(k-2)} \ot \log\rho_k \ot I^{\ot(d-k)}$ and the property~$\tr[\sigma(I^{\ot(k-1)} \ot X \ot I^{\ot(d-k)})] = \tr[\sigma_k X]$ of the partial trace~$\sigma_k = \tr_{\check k}\sigma$.
On the other hand, using the cyclicity property of the partial trace and again that~$\sum_\lambda P_\lambda^2 = I$, we see that
\begin{align*}
    (\tilde\sigma)_1
&= \tr_{\check1} \mleft[ \sum_\lambda (I \ot P_\lambda) \sigma (I \ot P_\lambda) \mright]
= \tr_{\check1} \mleft[ \sigma \left(I \ot \sum_\lambda P_\lambda^2 \right) \mright]
= \tr_{\check1} \sigma = \sigma_1, \\
(\tilde\omega)_1
&= \omega_1
= \rho_1^{-\theta_1/2} \rho_1 \rho_1^{-\theta_1/2}
= \rho_1^{1-\theta_1},
\end{align*}
where the first equality on the second line is derived just like the first line, and the second equality follows from the first formula in \cref{eq:def omega}.
Therefore:
\begin{align}\label{eq:DPI right}
    D\mleft( (\tilde\sigma)_1 \Vert (\tilde\omega)_1 \mright)
= \tr[\sigma_1 \log \sigma_1] - \tr[\sigma_1 \log \rho_1^{1-\theta_1}]
= -H(\sigma_1) - (1-\theta_1) \tr[\sigma_1 \log \rho_1].
\end{align}
Thus, \cref{eq:entropic subgoal} (for $\ell=1$) follows from \cref{eq:DPI left,eq:DPI right} and the data-processing inequality~\eqref{eq:dpi} applied to $\tilde\sigma$, $\tilde\omega$, and the partial trace~$\tr_{\check1}$.
\end{proof}

Thus we obtain an $\mathcal{O}(k^{-1})$ convergence rate for entropic tensor scaling towards the logarithmic quantum functional as a consequence of our general results (\cref{cor:entanglementpolytope,thm:tensor_f+b}).

\begin{theorem}[Convergence of entropic tensor scaling]\label{thm:entropic tensor scaling}
Let $\theta$ be a probability distribution on~$[d]$, let $T \in \Ten(n_0; n_1, \ldots, n_d)\setminus\{0\}$ be a concise tensor tuple, and let~$(T_k)_{k\in\NN}$ be the sequence generated by entropic tensor scaling~\eqref{eq:entropic tensor scaling} with step size~$\eta = 1$.
Then, for every $k\geq1$,
    \[
        0 \leq \logqfunc_\theta(T) - \sum_{\ell \in [d]} \theta_\ell H((\rho_k)_\ell)
        \leq 2 \frac{\log\|T\| - \log\Capacity_{p^\star}(T)}{k},
    \]
    where $((\rho_k)_1, \ldots, (\rho_k)_d) \coloneqq \mu(T_k)$ are the one-body quantum marginals of~$T_k$, and $p^\star$ is any minimizer of $\phi_\theta$ in $\Delta(T)$.
    In particular, if $T$ has Gaussian integer entries, then the right-hand side is bounded by $\poly(d, n_{\max}, \braket T)/k$, where $n_{\max} \coloneqq \max_{\ell = 1}^dn_\ell$ and $\braket T$ denotes the bit-length of the tensor tuple.
\end{theorem}
\begin{proof}
Recall that $\phi_\theta$ satisfies \ref{it:phi1gl}--\ref{it:phi3gl} and that the conciseness of~$T$ implies that $\spec\mu(g\cdot T) \in \intr\dom\phi_\theta$ for all~$g \in G$.
Furthermore, the step-size condition is satisfied for~$\eta=1$ by \cref{thm:quantum}.
Hence, \cref{cor:entanglementpolytope} applies with $\phi = \phi_\theta$.
Using the facts that $\phi_\theta(\spec\mu(T_k)) = -\sum_{\ell=1}^d \theta_\ell H((\rho_k)_\ell)$ and $\min_{p \in \Delta(T)}\phi_\theta(p) = -\logqfunc_\theta(T)$, the displayed convergence bound follows.
The ``in particular'' part follows directly from \cref{thm:tensor_f+b}.
\end{proof}

\noindent
Thus, for Gaussian integer tensors, we can approximate the logarithmic quantum functional up to additive error~$\eps>0$ in a number of iterations that is linear in~$1/\eps$ and polynomial in input size.

The algorithmic result above also leads to a new succinct proof of the celebrated fact that the quantum functionals are universal spectral points~\cite{CVZ2018}.
In the literature, this notion is considered for single tensors, so we take $n_0=1$ and denote~$\Ten(n_1,\dots,n_d) \coloneqq \Ten(1;n_1,\dots,n_d)$.\footnote{But see \cref{footnote:qfunc tensor tuple}, which provides a tensor-tuple interpretation of the results below.}

\begin{definition}[Universal spectral point]
Fix any $d\geq2$.
Suppose that $\qfunc \colon \{ \text{$d$-tensors} \} \to \RR_{\geq0}$ is a function that assigns to any $d$-tensor (with arbitrary dimensions $n_1,\dots,n_d$) a nonnegative number.
Then $\qfunc$ is called a \emph{spectral point} for the semiring of complex $d$-tensors if the following hold:
\begin{itemize}
\item \emph{Monotonicity:} $\qfunc(S) \le \qfunc(T)$ if $S = (A_1\otimes\cdots\otimes A_d)T$ for some matrices $A_1, \ldots, A_d$.
\item \emph{Multiplicativity:} $\qfunc(S\otimes T) = \qfunc(S)\qfunc(T)$, where $S \ot T$ is regarded as a $d$-tensor.
\item \emph{Additivity:} $\qfunc(S\oplus T) = \qfunc(S)+\qfunc(T)$, where $S \op T$ is regarded as a $d$-tensor.
\item \emph{Normalization:} $\qfunc(\langle r\rangle) = r$ for the unit tensor $\langle r\rangle \coloneqq \sum_{j=1}^r e_j^{\ot d}$.
\end{itemize}
\end{definition}

\noindent This terminology is due to Strassen, who defined the \emph{asymptotic spectrum} of the semiring of tensors (in fact, for general sub-semirings of $d$-tensors) and showed that the universal spectral points, i.e., the points in the asymptotic spectrum of the set of all $d$-tensors, characterize asymptotic transformations of tensors~\cite{Strassen-88}.
Accordingly, the asymptotic spectrum has deep applications in algebraic complexity and quantum information theory~\cite{zuiddam2018algebraic}.

For matrices, i.e., 2-tensors, the matrix rank is the only spectral point.
For tensors of higher order, the ranks of their flattenings are trivial examples of universal spectral points.
In a breakthrough work, Christandl, Vrana, and Zuiddam gave the first nontrivial family: they introduced the quantum functionals~$F_\theta$ (as in \cref{def:qfunc}) and proved that they are universal spectral points for any fixed choice of~$\theta$~\cite{CVZ2018}.
Monotonicity, normalization, super-multiplicativity~$\qfunc_\theta(S\otimes T) \geq \qfunc_\theta(S)\qfunc_\theta(T)$, and super-additivity~$\qfunc_\theta(S\oplus T) \ge \qfunc_\theta(S)+\qfunc_\theta(T)$ are straightforward consequences of the definition and elementary properties of the entanglement polytope.
In contrast, sub-multiplicativity $\qfunc_\theta(S\otimes T) \le \qfunc_\theta(S)\qfunc_\theta(T)$ and sub-additivity~$\qfunc_\theta(S\oplus T) \le \qfunc_\theta(S)+\qfunc_\theta(T)$ require more substantial arguments.
In~\cite{CVZ2018}, these two properties were proved via a representation-theoretic dual description of the entanglement polytope.
Only recently have geometric proofs been found that avoid representation theory~\cite{Hirai2025,SDW2026}.

We now give a new \emph{algorithmic} proof of the theorem by Christandl, Vrana, and Zuiddam that uses the convergence of entropic tensor scaling to obtain the multiplicativity and additivity of the quantum functionals.
In contrast, the recent proof in~\cite{Hirai2025} relies on properties of the continuous-time flow to establish sub-additivity, but uses duality for sub-multiplicativity, whereas the proof in~\cite{SDW2026} relies on a minimax formula based on that duality.

\begin{theorem}[\cite{CVZ2018}]\label{thm:algorithmicproof}
The quantum functional~$\qfunc_\theta$ (\cref{def:qfunc}) is a universal spectral point for any fixed~$\theta$.
\end{theorem}
\begin{proof}
We first show multiplicativity.
Let $S\in\Ten(m_1,\dots,m_d)$ and $T\in\Ten(n_1,\dots,n_d)$ be two $d$-tensors; we may assume without loss of generality that both are nonzero and concise.\footref{footnote:concise}
We can think of $S \ot T \in \Ten(m_1n_1,\dots,m_dn_d)$.
Then the flattenings are given by $(S\ot T)^{(\ell)} = S^{(\ell)}\ot T^{(\ell)}$ ($\forall \ell\in [d]$) and the moment map is given by
\[
    \mu(S\ot T) = \left( \mu(S)_1 \ot \mu(T)_1 , \dots,\mu(S)_d \ot \mu(T)_d \right).
\]
Thus, we see inductively that the entropic tensor scaling sequences~\eqref{eq:entropic tensor scaling} started at $S$, $T$, and~$S \ot T$, respectively, are related by
\[
    (S\ot T)_{k}=S_k\ot T_k \qquad (\forall k \in \NN).
\]
Consequently, $\mu((S \ot T)_k)_\ell = \mu(S_k)_\ell \ot \mu(T_k)_\ell$ for all $k\in\NN$ and $\ell\in[d]$, and
\[
    \qfunc_\theta(S\ot T)
= \lim_{k\to\infty} e^{\sum_{\ell \in [d]}\theta_\ell H(\mu( (S\ot T)_k )_\ell)}
= \lim_{k\to\infty} e^{\sum_{\ell \in [d]}\theta_\ell (H(\mu( S_{k} )_\ell)+H(\mu( T_{k} )_\ell))}
= \qfunc_\theta(S)\qfunc_\theta(T),
\]
where the first and last equalities hold because entropic tensor scaling computes the quantum functionals (\cref{thm:entropic tensor scaling}) and the middle equality follows from the additivity of the von Neumann entropy: $H(\rho\otimes\sigma)=H(\rho)+H(\sigma)$.

We next show sub-additivity.
Regarding $S \op T$ as a tensor in $\Ten(m_1+n_1,\dots,m_d+n_d)$, the flattenings of $S\op T$ are given by $(S\op T)^{(\ell)} = S^{(\ell)}\op T^{(\ell)}$ ($\forall \ell\in [d]$) up to the removal of zero columns.
Since zero columns in the flattenings do not impact the marginals, we have
\[
    \mu(S\op T) = \mleft( \lambda \,\mu(S)_1\op (1-\lambda) \, \mu(T)_1, \dots, \lambda \, \mu(S)_d \op (1-\lambda) \,\mu(T)_d \mright),
\]
where $\lambda \coloneqq\norm{S}^2/(\norm{S}^2+\norm{T}^2)$.
Inductively, we see that the entropic tensor scaling sequences~\eqref{eq:entropic tensor scaling} started at $S$, $T$, and~$S \op T$, respectively, are related by
\[
    (S\op T)_{k} \propto \sqrt{\lambda_k} \frac{S_k}{\norm{S_k}} \op \sqrt{1-\lambda_k} \frac{T_k}{\norm{T_k}} \qquad (\forall k \in \NN)
\]
for some~$\lambda_k\in [0,1]$.
Thus, $\mu((S \op T)_k)_\ell = \lambda_k \mu(S_k)_\ell \op (1-\lambda_k) \mu(T_k)_\ell$ for all $k\in\NN$ and $\ell\in[d]$, and
\begin{align*}
  \qfunc_\theta(S\op T)
&= \lim_{k\to\infty} e^{\sum_{\ell \in [d]}\theta_\ell H(\mu( (S\op T)_k )_\ell)} \\
&= \lim_{k\to\infty} e^{\sum_{\ell \in [d]}\theta_\ell H(\lambda_k \mu(S_k)_\ell \,\op\, (1-\lambda_k) \mu(T_k)_\ell )} \\
&= \lim_{k\to\infty} e^{\lambda_k \sum_{\ell \in [d]}\theta_\ell H(\mu(S_k)_\ell )
\,+\, (1-\lambda_k) \sum_{\ell \in [d]}\theta_\ell H(\mu(T_k)_\ell )
\,+\, h(\lambda_k) } \\
&\leq \lim_{k\to\infty} \left( e^{\sum_{\ell \in [d]}\theta_\ell H(\mu(S_k)_\ell )} + e^{\sum_{\ell \in [d]}\theta_\ell H(\mu(T_k)_\ell )} \right)
= \qfunc_\theta(S) + \qfunc_\theta(T),
\end{align*}
where the first and the last equalities hold again because entropic tensor scaling computes the quantum functionals (\cref{thm:entropic tensor scaling}),
the third equality follows from the standard recursion property of entropy, with $h(\lambda) \coloneqq -\lambda \log\lambda-(1-\lambda)\log(1-\lambda)$ denoting the binary entropy function,
and the inequality holds because $e^{\lambda x+(1-\lambda)y+h(\lambda)} \le e^x+e^y$ for all $x, y \in \RR$ and $\lambda \in [0, 1]$ by Jensen's inequality (e.g., \cite[(2.13)]{Strassen1991}).
This proves sub-additivity.

As discussed above, the remaining properties of a spectral point are easily satisfied.
We briefly recall their standard proof.
Super-additivity follows readily from the observation that $\Delta(S \op T)$ contains all points of the form $\lambda p \op (1-\lambda) q$ (up to sorting nonincreasingly) for $p \in \Delta(S)$, $q \in \Delta(T)$ and $\lambda\in[0,1]$.
Monotonicity holds because $\Delta(S) \subseteq \Delta(T)$ for any $S \in \overline{G \cdot T}$, so in particular for $S = (A_1 \ot \cdots \ot A_d) T$ if $A_1,\dots,A_d$ are square matrices (and we can always reduce to this situation by padding with zeros).
Finally, normalization holds because the marginals of the unit tensor are uniform ($\rho_\ell = I/r$), and uniform distributions have maximum entropy ($\log r)$.
\end{proof}

\subsection{Symmetric quantum functional}\label{sub:symmetric_quantum}
In this subsection, we give an algorithm for the \emph{symmetric quantum functional}~\cite{CFTZ-22}.
The setting is the tensor power action in \cref{sub:tensor_power_action}:
the group $G = \GL(n)$ acts on tuples of $d$-tensors in $\VV = \Ten(n_0; n, \ldots, n)$ as in \cref{eq:tenpow}.
The moment map for this action (\cref{eq:sym moment map}) computes the sum of the $d$ quantum marginals, hence the symmetric entanglement polytope (\cref{eq:sym moment polytope}) consists of vectors~$p\in\RR^n$ such that $p/d$ is a probability distribution.
The logarithmic symmetric quantum functional is then defined as the maximum entropy of the distributions thus obtained.%
\footnote{Like the ordinary quantum functionals, the symmetric quantum functionals were originally proposed for a single tensor ($n_0 = 1$), but the definition generalizes naturally to tensor tuples.}
Recall the Shannon entropy, defined by $H(p)\coloneqq -\sum_ip_i \log p_i$, with the usual convention that $0\log0 \coloneqq 0$.

\begin{definition}[Symmetric quantum functional, \cite{CFTZ-22}]
The \emph{logarithmic symmetric quantum functional}~$\logqfunc_{\Sym}$ is, for all tensors~$T \in \Ten(n_0; n,\dots,n)\setminus\{0\}$, given by
\[
    \logqfunc_{\Sym}(T) \coloneqq \max_{p \in \Delta_{\Sym}(T)} H(p/d).
\]
The \emph{symmetric quantum functional} is defined as $\qfunc_{\Sym}(T) \coloneqq e^{\logqfunc_{\Sym}(T)}$; we also set 
$\qfunc_{\Sym}(0) \coloneqq 0$.
\end{definition}

\noindent
Hence, this falls into our minimization framework with the objective function $\phi(p) \coloneqq -H(p/d)$; as usual we set $H(p)\coloneqq-\infty$ if $p\not\in\RR^n_{\geq0}$.
Instead of analyzing this problem from scratch, we can use \cref{thm:tensor power reduction} to relate the minimization of~$\phi$ on the symmetric entanglement polytope to the minimization of the objective~\eqref{eq:phisum}
\[
    \phi_{\Sum}(p_1, \ldots, p_d) \coloneqq \frac1d\sum_{\ell=1}^d\phi(dp_\ell) = -\frac1d\sum_{\ell=1}^dH(p_\ell)
\]
on the entanglement polytope $\Delta(T_{\Shf})$ of the shift stack tensor tuple~$T_{\Shf}$ from \cref{def:shift_stack}.
Note that this is precisely the objective function~\eqref{eq:qfunc objective} for the logarithmic quantum functional with the uniform distribution $\theta = (1/d, \ldots, 1/d)$.
Consequently, the symmetric quantum functional of~$T$ can be computed as the quantum functional of $T_{\Shf}$ with the uniform distribution.
We may assume that the moment map image~$\mu_{\Sym}(T)$ is invertible and hence positive definite.%
\footnote{\label{footnote:symmetric-concise}To see that this can be assumed without loss of generality, note that $\mu_{\Sym}(T) = \sum_{\ell=1}^d T^{(\ell)} T^{(\ell)\dagger} / \norm T^2$ is singular if and only if their common kernel $K \coloneqq \bigcap_{\ell=1}^d \ker T^{(\ell)\dagger}$ is nontrivial.
Analogously to \cref{footnote:concise} we can project~$T$ onto~$(K^\perp)^{\ot d}$ to obtain a tensor~$T'$ of smaller format such that $\mu_{\Sym}(T')$ is invertible.}
Since $\mu_{\Sym}(T)$ is the sum of the one-body quantum marginals, this is in general a weaker assumption than requiring that $T$ is concise (in fact it is equivalent to requiring that $T_{\Shf}$ is concise).
If we describe the iterations in terms of tensors~$T_k \in \Ten(n_0;n,\dots,n)$ instead of their shifts~$(T_{\Shf})_k = (T_k)_{\Shf}$, we obtain the following iteration for step size~$\eta=1$, which we call \emph{symmetric entropic tensor scaling}:
\begin{equation}\label{eq:symmetric_tensor_scaling}
    T_{k + 1} \coloneqq \left(I\otimes\rho_\text{avg}^{-1/(2d)}\otimes\cdots\otimes\rho_\text{avg}^{-1/(2d)}\right)T_k, \quad \text{where}\quad\rho_\text{avg} \coloneqq \mu_{\Sym}(T_k)/d,\qquad T_0 \coloneqq T.
\end{equation}
Note that $\rho_\text{avg} = (\rho_1+\ldots+\rho_d)/d$ for $(\rho_1,\dots,\rho_d)\coloneqq\mu(T_k)$, hence the notation.
Of course, \cref{eq:symmetric_tensor_scaling} can also be obtained from \cref{eq:updatesymmetric} and the same step size~$\eta=1$.

The above computes the logarithmic symmetric quantum functional with an $\mathcal{O}(k^{-1})$ convergence rate:

\begin{theorem}[Convergence of symmetric entropic tensor scaling]\label{thm:symmetric-quantum}
    Let $0 \neq T \in \Ten(n_0; n, \ldots, n)$ be a tuple of $d$-tensors such that $\mu_{\Sym}(T)$ is nonsingular, and let $(T_k)_{k\in\NN}$ denote the sequence generated by symmetric entropic tensor scaling~\eqref{eq:symmetric_tensor_scaling}.
    Then, for every~$k\geq1$,
    \[
        0 \leq \logqfunc_{\Sym}(T) - H((\rho_k)_\text{avg}) \leq 2 \frac{\log\|T\| - \log\Capacity_{p^\star}(T)}{k},
    \]
    where $(\rho_k)_\text{avg} \coloneqq \mu_{\Sym}(T_k)/d$ is the average one-body marginal of~$T_k$, $p^\star$ is any minimizer of $\phi$ in $\Delta_{\Sym}(T)$, and the capacity refers to the tensor power action.
    Moreover, if $T$ has Gaussian integer entries, then the objective gap is also bounded by $\poly(d, n, \braket T)/k$, where $\braket T$ denotes the bit-length of the tensor tuple.
\end{theorem}
\begin{proof}
    Recall that~$\phi$ satisfies \ref{it:phi1gl}--\ref{it:phi3gl} and that the assumption on~$T$ implies that $\spec\mu_{\Sym}(g\cdot T) \in \intr\dom\phi$ for all~$g \in G$.
    By \cref{thm:tensor power reduction,thm:quantum}, since $\phi_{\Sum}$ is the objective function for the quantum functional with $\theta=(1/d,\dots,1/d)$, the step-size condition is satisfied for~$\eta=1$.
    Hence, the theorem follows from \cref{cor:symmetricentanglementpolytope}, using that $\phi(\spec\mu_{\Sym}(T_k)) = -H((\rho_k)_\text{avg})$ and $\min_{p \in \Delta_{\Sym}(T)}\phi(p) = -\logqfunc_{\Sym}(T)$.
\end{proof}

We remark that one can also obtain a convergence result as a direct consequence of \cref{thm:entropic tensor scaling} (involving the capacity of $T_{\Shf}$ and an additional term $\log d$ in the numerator).

\subsection{\texorpdfstring{$G$}{G}-stable rank via smoothing}\label{sub:G-stable-rk}
Next, we turn to the computation of the \emph{$G$-stable rank}, a tensor invariant introduced by Derksen~\cite{Derksen-22}.
As in \cref{sub:quantum}, the setting of this subsection is the tensor action in \cref{sub:entanglementpolytope}: the group $G = \GL(n_1)\times\cdots\times\GL(n_d)$ acts on $\VV = \Ten(n_0; n_1, \ldots, n_d)$ as in \cref{eq:tensor action}.
This $G$-stable rank is defined for tensors over arbitrary perfect fields, but for complex tensors it can also be expressed in terms of the entanglement polytope~\cite[Thm~5.2]{Derksen-22}, as follows.%
\footnote{The $G$-stable rank was originally introduced for a single tensor (i.e., $n_0 = 1$), but the definition generalizes naturally to tensor tuples.}

\begin{definition}[{$G$-stable rank, \cite{Derksen-22}}]
For a parameter $\alpha \in \RR_{>0}^d$, the $G$-stable rank~$\rk_{\alpha}^G$ is, for all $T \in \Ten(n_0; n_1, \ldots, n_d)\setminus\{0\}$, defined as
\[
    \rk_{\alpha}^G(T)
\coloneqq \sup_{g \in G} \min_{\ell\in[d]} \frac {\alpha_\ell \norm{g \cdot T}^2} {\norm{(g \cdot T)^{(\ell)}}_{\ope}^2}
= \sup_{g \in G} \min_{\ell\in[d]} \frac {\alpha_\ell} {\norm{\mu(g \cdot T)_\ell}_{\ope}}
= \max_{p \in \Delta(T)} \min_{\ell\in[d]} \frac {\alpha_\ell} {\norm{p_\ell}_\infty},
\]
where $\norm{(g \cdot T)^{(\ell)}}_{\ope}$ is the operator norm of the $\ell$-th flattening of the tensor~$g \cdot T$.
\end{definition}

\noindent
The first equality follows immediately from the moment map for the tensor action (\cref{eq:momentmap_tensor}).
Accordingly, the inverse of the $G$-stable rank can be written as a convex minimization over the entanglement polytope~$\Delta(T)$ (\cref{eq:entanglementpolytope}): 
\begin{align}\label{eq:Gstable-rank}
    \frac 1 {\rk_{\alpha}^G(T)} = \min_{p \in \Delta(T)} \max_{\ell\in[d]} \frac {\norm{p_\ell}_\infty} {\alpha_\ell} = \min_{p\in\Delta(T)} \phi(p),
\end{align}
for the convex objective
\begin{align*}
    \phi \colon \RR^{n_1} \times \cdots \times \RR^{n_d} \to \RR, \quad \phi(p_1,\dots,p_d) \coloneqq \max_{\ell\in[d]} \max_{i\in[n_\ell]} \frac {p_{\ell,i}} {\alpha_\ell}.
\end{align*}
However, $\phi$ is not differentiable.
In order to treat this problem in our framework, we utilize the \emph{smoothing technique} in \cite[\S{}4]{Nesterov-smoothing}.
Concretely, we approximate~$\phi$ with the following objective of log-sum-exp type, parametrized by a smoothing parameter~$\delta > 0$:
\[
    \phi_\delta \colon \RR^{n_1} \times \cdots \times \RR^{n_d} \to \RR, \quad \phi_\delta(p) \coloneqq \delta\log\sum_{\ell \in [d], i \in [n_\ell]}e^{p_{\ell, i}/(\delta \alpha_\ell)}.
\]
The function $\phi_\delta$ is convex, symmetric in each component, and $L$-smooth, where~$L \coloneq 1/(\delta\alpha_{\min}^2)$, with $\alpha_{\min} \coloneqq \min_{\ell \in [d]}\alpha_\ell$~\cite[Lem~3.10]{BLNW2020}.
Moreover, it gives a controlled additive approximation to~$\phi$~\cite[(17)]{Nesterov-smoothing}:
\begin{equation}\label{eq:nesterov guarantee}
    \phi(p) \le \phi_\delta(p) \le \phi(p) + \delta\log n_{\Sum}
    \quad (n_{\Sum} \coloneqq \sum_{\ell=1}^d n_\ell).
\end{equation}
Hence, by minimizing $\phi_\delta$ for sufficiently small $\delta$, the inverse of the $G$-stable rank can be approximately computed with an arbitrarily small additive error.

We now derive the Hadamard mirror descent iteration for this objective function.
First, the unitarily invariant extension of $\phi_\delta$ to tuples of Hermitian matrices is given by (see \cref{ex:F examples}~(b))
\[
    \Phi_\delta \colon \Herm(n_1) \times \cdots \times \Herm(n_d) \to \RR, \quad
    \Phi_\delta(X) \coloneqq \phi_\delta(\spec X) = \delta\log\sum_{\ell=1}^d\tr e^{X_\ell/(\delta\alpha_\ell)},
\]
with the following gradient:
\[
    \nabla_\ell\Phi_\delta(X)
= \frac1{\alpha_\ell} \frac {e^{X_\ell/(\delta\alpha_{\ell})}} {\sum_{\ell'=1}^d\tr e^{X_{\ell'}/(\delta\alpha_{\ell'})}}.
\]
By substituting this into \cref{eq:updateTk}, and taking the step size $\eta = 1 / (d L) = \delta\alpha_{\min}^2/d$, we obtain the following iteration:
\begin{equation}\label{eq:g stable update}
    T_{k+1} \coloneqq \left(I \otimes \bigotimes_{\ell=1}^d\exp\left(-\frac{\delta\alpha_{\min}^2}{2d\alpha_\ell}\frac {e^{\rho_\ell/(\delta\alpha_{\ell})}} {\sum_{\ell'=1}^d \tr e^{\rho_{\ell'}/(\delta\alpha_{\ell'})}}\right)\right)T_k,\quad\text{where}\ (\rho_1,\dots,\rho_d) = \mu(T_k).
\end{equation}
By our general convergence results (\cref{cor:entanglementpolytope,thm:tensor_f+b}), this computes the inverse of the $G$-stable rank approximately:

\begin{theorem}[Approximation algorithm for inverse $G$-stable rank]\label{thm:Gstable}
Let $\alpha\in\RR^d_{>0}$, and let $T \in \Ten(n_0; n_1, \ldots, n_d)\setminus\{0\}$ be a tensor tuple.
Fix $\delta>0$, and let~$(T_k)_{k\in\NN}$ be the sequence generated by the iteration~\eqref{eq:g stable update} starting from~$T_0 \coloneqq T$.
Then, for every $k\geq1$,
\[
    0 \leq \phi_\delta(\spec\mu(T_k)) - \frac1{\rk_\alpha^G(T)} \leq \delta\log n_{\Sum} + \frac{2d\left(\log\|T\| - \log\Capacity_{p^\star}(T)\right)}{\delta\alpha_{\min}^2k},
\]
where $p^\star$ is any minimizer of $\phi_\delta$ in $\Delta(T)$.
Thus, to compute the inverse $G$-stable rank with additive error~$\eps>0$, we may choose%
\footnote{If $n_{\Sum}=1$, then $\phi_\delta$ is independent of $\delta$ and $\delta\in\RR_{>0}$ can be chosen arbitrarily. In this case, we have $\phi_{\delta}(\spec\mu(T_k))=\frac{1}{\rk_\alpha^G(T)}$ for every $k\geq 1$.}%
~$\delta\coloneqq\eps/(2\log n_{\Sum})$ and~$k\geq 8d\log n_{\Sum}(\log\|T\| - \log\Capacity_{p^\star}(T))/(\alpha_{\min}\varepsilon)^2$ iterations.
In particular, if $T$ has Gaussian integer entries, then $k = \poly(n_{\Sum}, \braket T)/(\alpha_{\min}\varepsilon)^2$ iterations suffice, where $\braket T$ denotes the bit-length of the tensor tuple~$T$.
\end{theorem}
\begin{proof}
    Since $\phi_\delta$ is $1/(\delta\alpha_{\min}^2)$-smooth convex and symmetric in each argument, $\phi_\delta$ satisfies \ref{it:phi1gl}--\ref{it:phi3gl} and it holds that $\spec \mu(g\cdot T) \in \intr\dom\phi_\delta$ for all $g \in G$.
    Hence, \cref{cor:entanglementpolytope} applies with~$\phi = \phi_\delta$ and the step size~$\eta = \delta\alpha_{\min}^2/d$, and it gives
    \[
        0 \leq \phi_\delta(\spec\mu(T_k)) - \min_{p \in \Delta(T)}\phi_\delta(p) \le \frac{2d\left(\log\|T\| - \log\Capacity_{p^\star}(T)\right)}{\delta\alpha_{\min}^2k}.
    \]
    By \cref{eq:nesterov guarantee}, we know that $0 \le \min_{p \in \Delta(T)}\phi_\delta(p) - 1/\rk_\alpha^G(T) \le \delta\log n_{\Sum}$, and hence the displayed convergence bound follows.
    For the subsequent claim, note that our choice of~$\delta$ and~$k$ leads to an additive error
    \[
        \phi_\delta(\spec\mu(T_k)) - \frac1{\rk_\alpha^G(T)} \le \delta\log n_{\Sum} + \frac{2d\left(\log\|T\| - \log\Capacity_{p^\star}(T)\right)}{\delta\alpha_{\min}^2k} \leq \frac\varepsilon2 + \frac\varepsilon2 = \varepsilon.
    \]
    The claim for Gaussian integer tensors follows from the capacity lower bound of~\cref{thm:tensor_f+b}.
\end{proof}

\subsection{Non-commutative rank by \texorpdfstring{$1$}{1}-norm minimization}\label{sub:ncrk}
Finally, we revisit the problem of computing the \emph{non-commutative rank} of a matrix tuple.
Let $\mathcal{A} = (A_1, \ldots, A_{n_0}) \in (\CC^{n \times n})^{\op n_0}$ be a tuple of square matrices.
The \emph{non-commutative rank} of $\mathcal{A}$, denoted by $\ncrk(\mathcal{A})$, is given by the rank of the following linear symbolic matrix (see, e.g., \cite{Fortin-Rautenauer-04})
\[
    A(x) \coloneqq \sum_{i \in [n_0]}x_iA_i,
\]
in \emph{non-commuting} indeterminates $(x_i)_{i \in [n_0]}$ (that is, $x_ix_j \neq x_jx_i$ for $i\neq j$), defined formally over the \emph{free skew field}~\cite{Cohn1995}.
See also \cite[Thm.~1.17]{GGOW20} and \cite[Def.~2.1]{FSG-22} for several equivalent definitions and motivations for this fundamental notion.
Celebrated results state that the non-commutative rank is computable in deterministic polynomial time (e.g., when the matrix entries are rational numbers encoded in binary)~\cite{GGOW20,IQS-18,Hamada-Hirai-21}.

In this section, we derive a new algorithm for computing the non-commutative rank.
It is based on Hadamard mirror descent and a recent formula by Hirai.
Before stating our results, we introduce some identifications to connect with the literature.
We shall think of $\Ten(n_0; n, n) \cong (\CC^{n \times n})^{\op n_0}$ as the space of $n\times n$-matrix tuples $\mathcal{A} = (A_1, \ldots, A_{n_0})$.
Under this identification, the tensor action reads~$(g_1,g_2) \cdot \mathcal A\coloneqq (g_1A_1g_2^\top, \ldots, g_1A_{n_0}g_2^\top)$ for $g_1,g_2\in\GL(n)$; this is commonly called the \emph{left-right action}.
Then the two components of the moment map~\eqref{eq:momentmap_tensor} take the form
\begin{align*}
    \mu_{\LR}(\mathcal A)_1 = \sum_{i=1}^{n_0} \frac {A_i A_i^\dagger} {\norm{\mathcal A}^2}, \qquad
    \mu_{\LR}(\mathcal A)_2 = \left( \sum_{i=1}^{n_0} \frac {A_i^\dagger A_i } {\norm{\mathcal A}^2} \right)^\top,
\end{align*}
where $\norm{\mathcal A}^2 \coloneqq \sum_{i=1}^{n_0} \norm{A_i}_{\mathrm F}^2$.
The transpose in the right-hand side formula does not impact the spectrum.
Hence in the present setting, the entanglement polytope~\eqref{eq:entanglementpolytope} can be written as
\begin{align*}
    \Delta_{\LR}(\mathcal A) = \overline{\left\{ \left( \spec \sum_{i=1}^{n_0} \frac {B_i B_i^\dagger} {\norm{\mathcal B}^2}, \spec \sum_{i=1}^{n_0} \frac {B_i^\dagger B_i} {\norm{\mathcal B}^2}\right) \,\middle|\, g_1, g_2 \in \GL(n), \; B_i = g_1 A_i g_2^\top \, (\forall i \in[n_0]) \right\}}.
\end{align*}
Without loss of generality, we assume that at least one of $\mu_{\mathrm{LR}}(\mathcal{A})_1$ or $\mu_{\mathrm{LR}}(\mathcal{A})_2$ is invertible; otherwise, we may represent $\mathcal{A}$ as a tuple of smaller square matrices.%
\footnote{For a matrix tuple $\mathcal{A} \in \Ten(n_0; n, n)$, let $\mathcal A' \in \Ten(n_0'; n_1', n_2')$ be its concise reduction obtained by the method in \cref{footnote:concise}.
If $n_1' > n_2'$, by adding auxiliary zero slices to the last mode, we may make it a tuple of \emph{square} matrices $\mathcal A'' \in \Ten(n_0'; n_1', n_1')$ while keeping $\mu_{\mathrm{LR}}(\mathcal{A}'')_1$ invertible.
Both the concise reduction and the padding by zero slices preserve the non-commutative rank.
The case $n_1' < n_2'$ is handled symmetrically.}
This condition is equivalent to the following \emph{kernel condition}:
\begin{equation}\label{eq:kernel-condition}
    \bigcap_{i \in [n_0]}\ker A_i = \{0\} \qquad\text{or}\qquad \bigcap_{i \in [n_0]}\ker A^\dagger_i = \{0\}.
\end{equation}
Under this hypothesis, Hirai~\cite{Hirai-ncrank,Hirai2025} derived the following expression for the non-commutative rank as an $\ell^1$-norm minimization over the entanglement polytope.

\begin{proposition}[{\cite[Thm.~1.4]{Hirai-ncrank}, \cite[Thm.~4.10]{Hirai2025}}]\label{prop:hirai}
    Suppose that $\mathcal{A} = (A_1, \ldots, A_{n_0}) \in (\CC^{n\times n})^{\op n_0}$ satisfies the kernel condition~\eqref{eq:kernel-condition}.
    Then, its non-commutative rank is given by
    \[
        \ncrk(\mathcal{A}) = n - \frac{n}2\min_{(p_1, p_2) \in \Delta_{\mathrm{LR}}(\mathcal{A})}\left(\left\|p_1 - \frac{\bm{1}_n}n\right\|_1 + \left\|p_2 - \frac{\bm{1}_n}n\right\|_1\right).
    \]
\end{proposition}

\noindent
We note that his original characterization used the \emph{unnormalized} moment map~\cite[Thm.~1.4]{Hirai-ncrank}.
The above expression then follows under the kernel condition~\eqref{eq:kernel-condition}; see the discussion following~\cite[Thm.~4.10]{Hirai2025}.

Thus, the non-commutative rank is obtained by minimization over the entanglement polytope (\cref{sub:entanglementpolytope}) of the following convex objective:
\[
    \phi \colon \RR^n \times \RR^n \to \RR, \quad \phi(p_1, p_2) = \left\|p_1 - \frac{\bm{1}_n}n\right\|_1 + \left\|p_2 - \frac{\bm{1}_n}n\right\|_1.
\]
As in the previous subsection, $\phi$ is not differentiable, so we construct a smooth approximation.
Observe that $\phi$ is the sum of $|p_{\ell, i_\ell} - 1/n|$ for all $\ell \in \{1,2\}$ and~$i_\ell \in [n]$.
We can approximate $|\cdot| \approx \sqrt{(\cdot)^2 + \delta^2}$, where $\delta > 0$ controls the error.
Thus we arrive at the following objective
\[
    \phi_\delta \colon \RR^n \times \RR^n \to \RR, \quad\phi_\delta(p_1, p_2) \coloneqq \sum_{\ell \in \{1, 2\}}\sum_{i \in [n]}\sqrt{\left(p_{\ell, i} - \tfrac1n\right)^2 + \delta^2}.
\]
Since $\sqrt{(\cdot)^2 + \delta^2}$ is convex and $1/\delta$-smooth, $\phi_\delta$ is also convex and $1/\delta$-smooth, and it is clearly permutation-symmetric in both arguments.
Moreover, from $|\cdot| \le \sqrt{(\cdot)^2 + \delta^2} \le |\cdot| + \delta$, it follows that $\phi_\delta$ approximates $\phi$ up to additive error~$2n\delta$:
\begin{equation}\label{eq:ncrank error}
    \phi(p_1, p_2) \le \phi_\delta(p_1, p_2) \le \phi(p_1, p_2) + 2n\delta.
\end{equation}
Hence, by minimizing $\phi_\delta$, the original minimization can be solved with an additive error $\mathcal{O}(n\delta)$.

We now derive the Hadamard mirror descent iteration for this objective.
The unitarily invariant extension of $\phi_\delta$ to pairs of Hermitian matrices is given by (see \cref{ex:F examples}~(c))
\[
    \Phi_\delta \colon \Herm(n) \times \Herm(n) \to \RR, \quad
    \Phi_\delta(X_1, X_2) \coloneqq \tr\sqrt{\left(X_1 - \tfrac{I}n\right)^2 + \delta^2I} + \tr\sqrt{\left(X_2 - \tfrac{I}n\right)^2 + \delta^2I},
\]
with the following gradient:
\[
    \nabla\Phi_\delta(X_1, X_2) = \left(\left(X_1 - \tfrac{I}n\right)\left(\left(X_1 - \tfrac{I}n\right)^2 + \delta^2I\right)^{-\frac12}, \left(X_2 - \tfrac{I}n\right)\left(\left(X_2 - \tfrac{I}n\right)^2 + \delta^2I\right)^{-\frac12}\right).
\]
By substituting this into \cref{eq:updateTk} with the step size $\eta = \delta/2$, we obtain the following iteration:
\begin{equation*}
    \mathcal{A}_{k+1} \coloneqq \left(I \otimes\bigotimes_{\ell \in \{1, 2\}} e^{-\frac{\delta}4\left(\rho_\ell -\frac{I}n\right)\left(\left(\rho_\ell - \frac{I}n\right)^2 + \delta^2I\right)^{-\frac12}}\right)\mathcal A_k,
    \quad ((\rho_1,\rho_2) = \mu_{\mathrm{LR}}(\mathcal{A}_k)),
\end{equation*}
or, concretely in terms of the components of the matrix tuples~$\mathcal A_k = (A^{(k)}_1,\dots,A^{(k)}_{n_0})$:
\begin{equation}\label{eq:ncrank update concrete}
    A^{(k+1)}_i \coloneqq e^{-\frac{\delta}4 \sigma_L \left(\sigma_L^2 + \delta^2I\right)^{-\frac12}} A^{(k)}_i \, e^{-\frac{\delta}4\sigma_R\left(\sigma_R^2 + \delta^2I\right)^{-\frac12}},
    \quad
    \sigma_L \coloneqq \sum_{j=1}^{n_0} \tfrac {A_j^{(k)} A_j^{(k)\dagger}} {\norm{\mathcal A_k}^2} - \tfrac I n, \;
    \sigma_R \coloneqq \sum_{j=1}^{n_0} \tfrac {A_j^{(k)\dagger} A_j^{(k)}} {\norm{\mathcal A_k}^2} - \tfrac I n.
\end{equation}
By our general convergence results (\cref{cor:entanglementpolytope,thm:tensor_f+b}), this allows computing the non-commutative rank approximately.
Furthermore, since the non-commutative rank takes an integer value, the exact value can be computed by rounding.

\begin{theorem}[Non-commutative rank computation]\label{thm:ncrk}
    Let $\mathcal{A} \in (\CC^{n\times n})^{\op n_0} \cong \Ten(n_0; n, n)$ satisfy the kernel condition~\eqref{eq:kernel-condition}.
    Fix $\delta>0$, and let~$(\mathcal A_k)_{k\in\NN}$ be the sequence generated by the iteration~\eqref{eq:ncrank update concrete} starting from~$\mathcal A_0 \coloneqq \mathcal A$.
    Then, for every $k\geq1$,
    \[
        0 \leq \ncrk(\mathcal{A}) - \left( n - \frac n2 \phi_\delta(\spec\mu_{\mathrm{LR}}(\mathcal{A}_k)) \right)  \leq n^2\delta + \frac{2n\left(\log\|\mathcal{A}\| - \log\Capacity_{p^\star}(\mathcal{A})\right)}{\delta k},
    \]
    where $p^\star$ is any minimizer of $\phi_\delta$ in $\Delta_{\mathrm{LR}}(\mathcal{A})$.
    Thus, to compute the non-commutative rank of the matrix tuple~$\mathcal A$, we may choose~$\delta=1/(2n^2)$ and $k > 8n^3(\log\|\mathcal{A}\| - \log\Capacity_{p^\star}(\mathcal{A}))$ iterations; then we have
    \begin{equation}\label{eq:ncrank by rounding}
        \ncrk(\mathcal{A}) = n - \left\lfloor (n/2)\phi_\delta(\spec\mu_{\mathrm{LR}}(\mathcal{A}_k))\right\rfloor.
    \end{equation}
    If $\mathcal{A}$ has Gaussian integer entries, then $k = \poly(n, \braket{\mathcal A})$ iterations suffice to compute the non-commutative rank exactly, where $\braket{\mathcal A}$ denotes the bit-length.
\end{theorem}
\begin{proof}
    Since $\phi_\delta$ is $1/\delta$-smooth convex and symmetric in both arguments, $\phi_\delta$ satisfies \ref{it:phi1gl}--\ref{it:phi3gl} and it holds that $\spec \mu_{\mathrm{LR}}((g_1, g_2)\cdot \mathcal A) \in \intr\dom\phi_\delta$ for all $g_1, g_2 \in \GL(n)$.
    Hence, \cref{cor:entanglementpolytope} applies with~$\phi = \phi_\delta$ and the step size~$\eta = \delta/2$, and it gives
    \[
        0 \leq \phi_\delta(\spec\mu_{\mathrm{LR}}(\mathcal{A}_k)) - \min_{p \in \Delta_{\mathrm{LR}}(\mathcal{A})}\phi_\delta(p) \leq 4 \frac{\log\|\mathcal{A}\| - \log\Capacity_{p^\star}(\mathcal{A})}{\delta k}.
    \]
    From \cref{eq:ncrank error}, we see that $0 \le \min_{p \in \Delta_{\mathrm{LR}}(\mathcal A)}\phi_\delta(p) - \min_{p \in \Delta_{\mathrm{LR}}(\mathcal A)}\phi(p) \le 2n\delta$;
    moreover we know that $\min_{p \in \Delta_{\mathrm{LR}}(\mathcal A)}\phi(p) = 2 - (2/n)\ncrk(\mathcal{A})$ from \cref{prop:hirai}.
    Thus we obtain the displayed convergence bound.
    For the rounding claim, note that our choice of~$\delta$ and~$k$ leads to an additive error
    \begin{align*}
        0 \leq \ncrk(\mathcal{A}) - \left(n - \frac n2\phi_\delta(\spec\mu_{\mathrm{LR}}(\mathcal{A}_k))\right) \leq n^2\delta + 2n \frac{\log\|\mathcal{A}\| - \log\Capacity_{p^\star}(\mathcal{A})}{\delta k} < \frac12 + \frac12 = 1;
    \end{align*}
    hence \cref{eq:ncrank by rounding} follows because $\ncrk(\mathcal{A})$ is an integer.
    Lastly, the claim for Gaussian integer matrix tuples follows from the capacity lower bound of~\cref{thm:tensor_f+b}.
\end{proof}

While our algorithm has polynomial iteration complexity for Gaussian integer (or rational) inputs, it does not directly yield a polynomial-time algorithm in the Turing model due to its use of exact real arithmetic; see \cref{rem:numerical}.
We plan to carry out a precision analysis to obtain such an algorithm in future work.
We still find it interesting to compare the method with the existing algorithms for non-commutative rank computation~\cite{GGOW20,IQS-18,Hamada-Hirai-21} (see also \cite{FSG-22}), as they are all based on different approaches.
The algorithm in \cite{GGOW20} solves the decision problem of whether $\ncrk(\mathcal{A}) \ge r$ by operator scaling on suitable matrix tuples constructed from~$\mathcal A$; then the non-commutative rank is obtained by a search over~$r$.
The other two algorithms~\cite{IQS-18,Hamada-Hirai-21} compute the non-commutative rank more directly, but they involve more complex steps.
In contrast, our algorithm iterates a single simple update formula and it targets the value of the non-commutative rank directly---combining these two desirable properties of prior work in one algorithm.

Another way to compute the non-commutative rank is via the $G$-stable rank.
Derksen~\cite[Prop.~2.9]{Derksen-22} showed that if one identifies the matrix tuple~$\mathcal A \in \Ten(n_0; n, n)$ with a $3$-tensor $T_{\mathcal A}\in \Ten(n_0, n, n) \coloneqq \Ten(1; n_0, n, n)$, the non-commutative rank equals the $G$-stable rank with $\alpha = (\ell, 1, 1)$ for all $\ell \ge n$.
By considering the limit $\ell \to \infty$, the corresponding minimization problem for the inverse of the $G$-stable rank~\eqref{eq:Gstable-rank} reduces to one that does not depend on the first mode.
For the resulting ``$\alpha = (\infty,1,1)$'' objective function, we can construct log-sum-exp approximations that do not depend on the first mode.
Then, the corresponding Hadamard mirror descent update for the tensor~$T_{\mathcal A}$ serves as the one for the matrix tuple~$\mathcal A\in\Ten(n_0; n, n)$ and $\alpha = (1, 1)$ (but with a smaller step size than the natural choice).
Thus, $\rk^G_{(1,1)}(\mathcal A) = \ncrk(\mathcal A)$, and the method in \cref{thm:Gstable} can also be used for non-commutative rank computation.

However, note that to compute $\rk^G_{(1,1)}(\mathcal A) \in \{1,\dots,n\}$, we have to approximate the inverse $G$-stable rank $1/\rk^G_{(1,1)}(\mathcal A)$ to additive error $\eps < \frac1{n-1} - \frac1n = \Theta(n^{-2})$.
To achieve this, \cref{thm:Gstable} shows that $k = \mathcal{O}\bigl((n^4 \log n) (\log\|\mathcal A\| - \log\Capacity_{p^\star}(\mathcal A))\bigr)$ iterations suffice.
In contrast, \cref{thm:ncrk} shows that $k = \mathcal{O}\bigl(n^3 (\log\|\mathcal A\| - \log\Capacity_{p^\star}(\mathcal A))\bigr)$ iterations suffice to compute the non-commutative rank exactly.
The $p^\star$ are not necessarily the same (they are minimizers for different objective functions), but we can upper bound both capacity terms by the common uniform estimate from \cref{thm:tensor_f+b}.
Then the resulting upper bounds on the iteration complexities can be compared: for the $G$-stable rank based method, we find an additional factor~$n \log n$ compared to the method presented in this section.

\section*{Acknowledgments}
We thank Hiroshi Hirai for sharing an early version of his manuscript~\cite{Hirai2025} and for suggesting~\cite{Nesterov-smoothing}.
We also thank Maxim van den Berg for valuable discussions and Frank Verstraete for asking about rigorous step sizes for inverse marginal iterations as in \cref{eq:intro-entropic}.
We acknowledge support by the European Union (ERC Grant SYMOPTIC, 101040907), by the Deutsche Forschungsgemeinschaft (DFG, German Research Foundation, 556164098), by the Deutsche Forschungsgemeinschaft under Germany's Excellence Strategy~--~EXC-2111~--~390814868, and by the German Federal Ministry of Research, Technology and Space (QuSol, 13N17173), and by the Klaus Tschira Foundation.
We thank the Simons Institute for the Theory of Computing at UC Berkeley and Q-FARM and the Leinweber Institute for Theoretical Physics at Stanford University for hospitality.

\appendix
\addtocontents{toc}{\protect\setcounter{tocdepth}{1}}

\section{Matrix scaling via marginal entropy maximization}\label{app:matrix scaling}
Let $A = (A_{ij}) \in \RR^{n \times m}$ be an entrywise nonnegative matrix.
The \emph{matrix scaling problem}~\cite{Sinkhorn1964,Idel2016} asks to transform it into an approximately doubly-stochastic matrix\footnote{The term ``doubly-stochastic'' is usually used for square matrices, but we also use it for rectangular matrices under appropriate normalization as in the text: namely, that the row and the column sums are uniform distributions.} by scaling with positive definite diagonal matrices:
\begin{equation}\label{eq:matrix scaling goal}
    B\bm{1} \approx \frac{\bm{1}}n, \quad B^\top\bm{1}\approx \frac{\bm{1}}m,\qquad \text{where}\quad B = XAY, \quad X, Y:\ \text{diagonal},\quad X, Y \succ 0.
\end{equation}
We may assume that each row and column of~$A$ has at least one positive entry, since otherwise this problem cannot be solved.

In this appendix, we recast the matrix scaling problem as a marginal entropy maximization problem, and obtain an algorithm that we call \emph{entropic matrix scaling}.
Its iterations take the~simple~form:
\begin{equation}\label{eq:entropic-matrix-scaling}
A_{k+1} \propto \diag(A_k\bm{1})^{-\frac{1}{2}} A_k \diag( A_k^{\top}\bm{1})^{-\frac{1}{2}}, \qquad A_0\coloneqq A.
\end{equation}
This algorithm maximizes the sum of the marginal entropies at a rate inversely proportional to the iteration count.
Because uniform distributions maximize the entropy, this solves problem~\eqref{eq:matrix scaling goal} whenever possible, but the algorithm also applies to ``non-scalable matrices'' for which~\eqref{eq:matrix scaling goal} cannot be achieved for arbitrarily small error.
We note that the iteration~\eqref{eq:entropic-matrix-scaling} has previously been analyzed in \cite{knight2014symmetry} (for scalable square matrices, and without the connection to mirror descent or entropy maximization).
It can also be obtained as a special case of the SMART iteration, which has a mirror descent interpretation~\cite{raus2024accelerated} (but on different variables from our mirror descent formalization below).

The problem of maximizing the sum of marginal entropies can be formalized as follows in a simplified manner.%
\footnote{\label{footnote:precise matrix scaling}
Strictly speaking, the rigorous way to formalize this problem is to restrict $f_A$ to a subspace $\MM \coloneqq \{(x, y) \in \RR^n \times \RR^m \mid \bm{1}^\top x = \bm{1}^\top y = 0\} \cong \RR^{n - 1} \times \RR^{m - 1}$, and define $\phi(p', q') \coloneqq -\frac12(H(p' + \bm{1}/n)+H(q' + \bm{1}/m))$ for $(p', q') \in \MM$; as usual we define~$H(p) \coloneqq -\infty$ for~$p\not\in\RR^n_{\geq0}$.
This ensures that $\dom\phi$ has a nonempty interior and the assumption in \cref{def:grad} is satisfied.
However, to avoid cumbersome notation, we proceed as if $f_A$ is defined on~$\RR^n\times\RR^m$.}
Let $W \coloneqq \{(p, q) \in \RR^n \times \RR^m \mid p \ge \bm{0}, \, q \ge \bm{0}, \, \bm{1}^\top p = \bm{1}^\top q = 1\}$ denote the set of pairs of probability distributions, and consider the following convex functions $f_A\colon \RR^n\times\RR^m \to \RR$ and $\phi\colon \RR^n\times\RR^m \to \overline\RR$:
\[
    f_A(x, y) \coloneqq \log\sum_{i, j}A_{ij}e^{x_i +y_j},\qquad
    \phi(p, q) \coloneqq \begin{cases}
        -\frac{1}{2}\left(H(p) + H(q)\right) & \text{if } (p, q) \in W,\\
        +\infty& \text{otherwise},
    \end{cases}
\]
where $H(p) \coloneqq -\sum_ip_i\log p_i$ with $0\log0 \coloneqq 0$ denotes the Shannon entropy.%
\footnote{The function~$f_A$ can be identified with the Kempf--Ness function for the action of~$G=(\CC^\times)^n  \times (\CC^\times)^m$ on~$\VV = \CC^{n \times m}$ by left-right-multiplication with diagonal matrices and the vector~$v = \sum_{i,j} \smash{A_{ij}^{1/2}} e_i e_j^\top$; see \cref{eq:kn,ex:euclidean as selfadjoint}. We will not use this connection in the subsequent discussion.}
The components of the Euclidean gradient~$\nabla f_A(x,y)$ are the row and column sums of the normalized rescaled matrix $B_{ij} \coloneqq A_{ij}e^{x_i +y_j} / \sum_{i', j'}A_{i'j'}e^{x_{i'} +y_{j'}}$; hence $\nabla f_A(x,y)$ is an element of~$W$.
Hence, the problem of minimizing~$\phi\circ\nabla f_A$ captures precisely the problem of maximizing the sum of the marginal entropies over all scalings of the matrix~$A$.

Hadamard mirror descent~\eqref{eq:descent} for minimizing $\phi\circ\nabla f_A$, which---as explained in \cref{sub:descent and euclidean}---coincides with ordinary mirror descent~\eqref{eq:mirror in grad space} for $h=f_A^*$~and~$\phi$ in gradient coordinates, is~given~by
\begin{equation}\label{eqn:matrixscaling-mirror}
    (x_{k + 1})_i \coloneqq (x_k)_i - \frac{\eta}{2}\log\mleft({\sum_{j = 1}^m}(B_k)_{ij}\mright) + c_k,\quad (y_{k + 1})_j \coloneqq (y_k)_j - \frac{\eta}{2}\log\mleft({\sum_{i = 1}^n}(B_k)_{ij}\mright) + c'_k,
\end{equation}
where $(B_k)_{ij} \coloneqq A_{ij}e^{(x_k)_i+(y_k)_j}/\sum_{i',j'}A_{i'j'}e^{(x_k)_{i'}+(y_k)_{j'}}$, and $c_k$ and $c'_k$ are any constants originating from the indeterminacy of $\nabla\phi$ in the directions normal to the affine hull of $W \subseteq \RR^n \times \RR^m$.
These constants do not change the matrices~$B_k$, which are normalized, nor the objective values~$\phi(\nabla f_A(x_k,y_k))$.
We take $x_0\coloneqq0$ and~$y_0\coloneqq0$ as the starting point.

We determine the relative smoothness parameter~$L$ in this setting by an analysis similar to \cite{raus2024accelerated}.
We use the characterization in \cref{eq:relsmooth-Bregman dual}, which is a precise special case of the Riemannian definition in \cref{eq:relsmooth} and in the present setting reads as follows:
\begin{align}\label{eq:matrix scaling smoothness goal}
    D_{\phi}(\nabla f_A(x',y')\Vert \nabla f_A(x,y)) \leq L \cdot D_{f_A}(x,y\Vert x',y') \qquad (x,x'\in\RR^n, y,y' \in \RR^m).
\end{align}

\begin{lemma}
Relative smoothness holds with~$L=1$.
Moreover, this cannot be improved for all matrices~$A$ (if $n,m\geq2$).
\end{lemma}
\begin{proof}
For any $x,y$, we define $B_{ij} \coloneqq A_{ij}e^{x_i +y_j} / \sum_{i', j'}A_{i'j'}e^{x_{i'} +y_{j'}}$ and denote by $p \coloneqq B\bm{1}$ and $q \coloneqq B^\top\bm{1}$ the row and column sums, respectively.
We similarly define~$B'$, $p'$ and $q'$ for $x',y'$.
Then, $\nabla f_A(x, y) = (p, q)$ and we can take $\nabla\phi(p,q) = \frac12 (\log p, \log q)$ (recall that~$\nabla\phi$ is only determined up to irrelevant multiples of the all-ones vector in each component).
By direct computation:
\begin{align*}
    D_\phi((p',q')\Vert (p,q)) &= \frac12 \left( D_{\mathrm{KL}}(p' \Vert p) + D_{\mathrm{KL}}(q' \Vert q) \right), \\
    D_{f_A}((x,y) \Vert (x',y')) &= D_{\mathrm{KL}}(B'\Vert B),
\end{align*}
where $D_{\mathrm{KL}}$ is the KL-divergence (we regard $B$ and $B'$ as probability distributions on $[n]\times[m]$).
As a consequence of the data processing inequality, which gives $D_{\mathrm{KL}}(B'\Vert B) \geq D_{\mathrm{KL}}(p'\Vert p)$ and $D_{\mathrm{KL}}(B'\Vert B) \geq D_{\mathrm{KL}}(q'\Vert q)$, we conclude that~\eqref{eq:matrix scaling smoothness goal} holds with~$L=1$.

To see that this cannot be improved for all matrices if $n,m\geq2$, consider $B=\begin{psmallmatrix}t&0\\0&(1-t)U\end{psmallmatrix}$ and $B'=\begin{psmallmatrix}t'&0\\0&(1-t')U\end{psmallmatrix}$, where $U \coloneqq \frac1{(n-1)(m-1)} \bm{1}_{n-1}\bm{1}_{m-1}^\top$ is a normalized all-ones matrix and $t, t'\in(0,1)$ with $t \neq t'$.
These clearly arise as scalings of the same matrix~$A$.
Then $D_{\mathrm{KL}}(B'\Vert B)
= D_{\mathrm{KL}}(p' \Vert p)
= D_{\mathrm{KL}}(q' \Vert q)
> 0$
are all equal (to the binary KL-divergence between~$t$ and~$t'$).
\end{proof}

Taking the step size $\eta=1/L=1$ in \cref{eqn:matrixscaling-mirror}, we see that the corresponding sequence of matrices
\begin{equation}\label{eq:matrixscaling-correspondence}
    A_k \propto e^{\diag x_k}\, A \, e^{\diag y_k}
\end{equation}
amounts precisely to the entropic matrix scaling algorithm~\eqref{eq:entropic-matrix-scaling}.
As a consequence of \cref{thm:main-discrete,lem:relsmooth}, we have the following convergence result:

\begin{theorem}[Convergence of entropic matrix scaling]\label{thm:matrix scaling}
    Let $A \in \RR^{n \times m}$ be an entrywise nonnegative matrix such that each row and column has at least one positive entry.
    Let $(A_k)_{k \in \NN}$ be a matrix sequence generated by the entropic matrix scaling~\eqref{eq:entropic-matrix-scaling}, and let $B_k \coloneqq A_k / \sum_{i, j}(A_k)_{ij}$.
    Then, for $k\geq1$,
    \[
        \phi\mleft(B_k\bm{1} , B_k^{\top}\bm{1}\mright) - \min_{(p, q) \in \overline{\nabla f_A(\RR^n \times \RR^m)}}\phi(p, q) \le \frac{\log \kappa_A}{k},
    \]
    where $\kappa_A \coloneqq  \frac{\sum_{i,j} A_{ij}}{\min_{i,j : A_{ij}>0}\; A_{ij}}$ is a natural condition number of the matrix~$A$.
\end{theorem}
\begin{proof}
    One can observe that $\phi(B_k\bm{1},B_k^{\top}\bm{1}) = \phi(\nabla f_A(x_k, y_k))$ from \cref{eq:matrixscaling-correspondence}.
    From the general theory of log-sum-exp functions, we have $f_A^{(p, q)}(0, 0) - \inf_{x, y} f_A^{(p, q)}(x, y) \le \log\kappa_A$ for all $(p, q) \in \overline{\nabla f_A(\RR^n \times \RR^m)}$; see \cite[Eq.~(1.9)]{BLNW2020} and the discussion following it.
    This bounds the numerator in the right-hand side of \cref{eq:main discrete first}.
    Thus, we obtain the upper bound from \cref{thm:main-discrete,lem:relsmooth}.
\end{proof}

We remark that entropic matrix scaling is a special case of entropic tensor scaling discussed in \cref{sub:quantum} applied to the matrix tuple $(\sqrt{A_{ij}}e_i\otimes e_j)_{i \in [n], j \in [m]}$ and~$\theta=(1/2,1/2)$.

\section{Technical proofs}
\label{app:technical proofs}
In this section, we prove \cref{lem:asymptoticrays,lem:Q-busemann} and the second equivalence in \cref{prop:shiftingtrick_orig}.
The first two results rely on the following lemma:

\begin{lemma}
    \label{lem:longdistanceparalleltransport}
    Let $\MM$ be a Hadamard manifold, and let $\gamma(t) \coloneqq \exp_x(tu)$ and $\beta(t) \coloneqq \exp_y(tv)$ $(u \in T_x\MM, v \in T_y\MM)$ be asymptotic geodesic rays.
    Then, it holds that
    \[
        v = \lim_{t \to \infty}\tau_{\gamma(t) \to y}\circ\tau_{x \to \gamma(t)}(u).
    \]
\end{lemma}
\begin{proof}
    The case of $u = 0$ is trivial; hence, we assume $u \neq 0$.
    Since $\gamma \sim \beta$, we have $\|v\| = \|u\| = \|\tau_{\gamma(t) \to y}\circ\tau_{x \to \gamma(t)}(u)\|$ for all $t \in \RR_{\ge 0}$.
    Therefore, it suffices to show convergence of the angle: $\lim_{t \to \infty}\angle_{y}(v, \tau_{\gamma(t) \to y}\circ\tau_{x \to \gamma(t)}(u)) = 0$.
    Let $\alpha_t: [0, 1] \to \MM$ be a geodesic such that $\alpha_t(0) = y$ and $\alpha_t(1) = \gamma(t)$.
    Let $\theta_t \coloneqq \angle_{\gamma(t)}(\dot\gamma(t), \dot\alpha_t(1))$ and $\varphi_t \coloneqq \angle_{y}(v, \dot\alpha_t(0))$.
    The law of cosines on nonpositively-curved spaces \cite[II.1.7~Proposition]{BridsonHaefliger1999} (or the $\mathrm{CAT}(0)$ inequality) yields,
    \begin{align*}
        \cos\theta_t
        &\ge \frac{\dist(x, \gamma(t))^2 + \dist(y, \gamma(t))^2 - \dist(x, y)^2}{2\dist(x, \gamma(t))\dist(y, \gamma(t))}\\
        &= \frac12\left(\frac{t\|u\|}{\dist(y, \gamma(t))} + \frac{\dist(y, \gamma(t))}{t\|u\|}\right) - \frac{\dist(x, y)^2}{2t\|u\|\dist(y, \gamma(t))} \ge 1 - \frac{\dist(x, y)^2}{2t\|u\|(t\|u\| - \dist(x, y))}\ \underset{t \to \infty}{\longrightarrow} 1 \\
    \intertext{and}
        \cos\varphi_t
        &\ge \frac{\dist(y, \beta(t))^2 + \dist(y, \gamma(t))^2 - \dist(\beta(t), \gamma(t))^2}{2\dist(y, \beta(t))\dist(y, \gamma(t))}\\
        &= \frac12\left(\frac{t\|u\|}{\dist(y, \gamma(t))} + \frac{\dist(y, \gamma(t))}{t\|u\|}\right) - \frac{\dist(\beta(t), \gamma(t))^2}{2t\|u\|\dist(y, \gamma(t))} \ge 1 - \frac{\dist(\beta(t), \gamma(t))^2}{2t\|u\|(t\|u\| - \dist(x, y))}\ \underset{t \to \infty}{\longrightarrow} 1,
    \end{align*}
    where the final inequalities hold for all sufficiently large~$t$ so that the denominators are positive;
    in the last limit we also used that $\dist(\beta(t), \gamma(t)) < C$ for some $C \in \RR$.
    Then,
    \[
        \angle_{y}(v, \tau_{\gamma(t) \to y}\circ\tau_{x \to \gamma(t)}(u)) \le \angle_{y}(v, \dot\alpha_t(0)) + \angle_y(\dot\alpha_t(0), \tau_{\gamma(t) \to y}\circ\tau_{x \to \gamma(t)}(u)) = \varphi_t + \theta_t\ \underset{t \to \infty}{\longrightarrow} 0,
    \]
    which completes the proof.
\end{proof}

Then we show \cref{lem:asymptoticrays}.
\begin{proof}[Proof of \cref{lem:asymptoticrays}]
    We denote by $\gamma(t) \coloneqq g^\dagger e^{tX}g$ a geodesic on $P$.
    By the concrete formula of parallel transport, we have $\tau_{\gamma(0) \to I}\dot\gamma(0) = \hat u^\dagger X\hat u$, where $\hat u \coloneqq g(g^\dagger g)^{-1/2} \in K$.
    Define $Y \in \bbH$ such that $t \mapsto e^{tY}$ is asymptotic to $\gamma$.
    By \cref{lem:longdistanceparalleltransport}, $Y$ is given by
    \[
        Y = \lim_{t \to \infty}\tau_{\gamma(t) \to I}\circ\tau_{\gamma(0) \to \gamma(t)}\dot\gamma(0)
        = \lim_{t \to \infty}\tau_{\gamma(t) \to I}\circ\tau_{\gamma(0) \to \gamma(t)}\circ\tau_{I \to \gamma(0)}(\hat u^\dagger X\hat u).
    \]
    Since $\tau_{\gamma(t) \to I}\circ\tau_{\gamma(0) \to \gamma(t)}\circ\tau_{I \to \gamma(0)}$ is an element of $\Hol_I(P)$, this is given by the adjoint action of some ${\tilde u}_t \in K$.
    Hence $Y = \lim_{t \to \infty}{\tilde u}_t^\dagger \hat u^\dagger X\hat u{\tilde u}_t$.
    Since every conjugated matrix lies in the $K$-orbit of~$X$, which is compact and hence closed, so does~$Y$.
\end{proof}

We now prove the second equivalence in \cref{prop:shiftingtrick_orig}.
The result itself is known (see, e.g., \cite{FW2020}), but we give a geometric proof that directly applies to arbitrary points~$X \in \bbH$.

\begin{lemma}\label{lem:shiftingtrick}
In the setting of \cref{sub:convex_optimization_on_moment_polytopes}, let $v \in \VV\setminus\{0\}$ and $X \in \bbH$.
Then:
    \[
        \exists g \in G,\ \log\capacity_X(g\cdot v) > -\infty \iff \exists u \in K,\ \log\capacity_X(u \cdot v) > -\infty.
    \]
\end{lemma}
\begin{proof}
Only the forward direction needs proof.
Let $g \in G$.
By \cref{lem:asymptoticrays}, there exists $u \in K$ such that the geodesic ray $\gamma(t) \coloneqq g^\dagger e^{tX}g$ is asymptotic to~$t \mapsto e^{t u^\dagger Xu}$, i.e., both determine the same equivalence class.
Then we have
\begin{align*}
    b^X(x)
= b^\gamma(g^\dagger x g)
= b^{u^\dagger X u}(g^\dagger x g) + c_g \qquad (x \in P),
\end{align*}
where the former equality follows from the isometry $x \mapsto g^\dagger xg$, and the latter from the fact that the Busemann functions of equivalent geodesics differ by a constant, which we call~$c_g \in \RR$.
Hence,
\begin{align*}
    f_{g\cdot v}^X(x)
= f_{g\cdot v}(x) + b^X(x)
= f_v(g^\dagger x g) + b^{u^\dagger X u}(g^\dagger x g) + c_g
= f_v^{u^\dagger X u}(g^\dagger x g) + c_g
\qquad (x \in P).
\end{align*}
By dividing by two and taking the infimum over $x \in P$, we obtain
\begin{align*}
    \log\capacity_X(g \cdot v)
= \log\capacity_{u^\dagger X u}(v) + \tfrac{c_g}2
= \log\capacity_X(u \cdot v) + \tfrac{c_g}2,
\end{align*}
where the second equality follows from \cref{eq:log cap K equivariance}.
Thus, if the left-hand side is finite, then so is the right-hand side.
\end{proof}

Next, we show \cref{lem:Q-busemann}.

\begin{proof}[Proof of \cref{lem:Q-busemann}]
    We denote $\gamma(t) \coloneqq \exp_x(tu)$.
    We let $R \coloneqq Q \circ \flat\colon T\MM \to \overline{\RR}$;
    then it is easy to see that $R$ is parallel transport invariant, and that $R_x$ is (Euclidean) l.s.c.\ convex.
    We have
    \[
        R_y(v) = R_y\left(\lim_{t \to \infty}\tau_{\gamma(t) \to y}\circ\tau_{x \to \gamma(t)}(u)\right) \le \lim_{t \to \infty}R_y\left(\tau_{\gamma(t) \to y}\circ\tau_{x \to \gamma(t)}(u)\right) = R_x(u),
    \]
    where the first equality follows from \cref{lem:longdistanceparalleltransport}, the inequality follows from the lower semicontinuity of $R_y$, and the last equality follows since $R$ is invariant under parallel transport.
    By symmetry, we also have $R_x(u) \le R_y(v)$.
    Hence, we have $Q_x(\flat u) = Q_y(\flat v)$.

    For $\xi \in C\MM^\infty$, the differential of the Busemann function at $x \in \MM$ is given by \cref{eq:diff busemann}: $\diff b^\xi(x) = -\flat u$, where $u \in T_x\MM$ is the unique tangent vector such that the geodesic $t\mapsto \exp_x(tu)$ belongs to~$\xi$.
    Therefore, if the differential at another point $y \in \MM$ is given by $\diff b^\xi(y) = -\flat v$ ($v \in T_y\MM$), then $\exp_x(tu)$ and $\exp_y(tv)$ are asymptotic.
    Thus, it holds that $Q_x(-\diff b^\xi(x)) = Q_x(\flat u) = Q_y(\flat v) = Q_y(-\diff b^\xi(y))$, which implies that $Q(\xi) \coloneqq Q_x(-\diff b^\xi(x))$ is well-defined.
\end{proof}

\bibliographystyle{amsalpha}
\bibliography{references}

\end{document}